\documentclass[a4paper, 10pt,reqno]{report}
\title{Computing Entanglement Polytopes}

\author{Konstantin Wernli}
\usepackage{fullpage}
\usepackage[english]{babel}
\usepackage[utf8]{inputenc}
\usepackage{cite}
\usepackage{multirow}
\usepackage{amsmath}
\usepackage{amssymb}
\usepackage{cleveref}
\usepackage{amsthm}
\usepackage{amsfonts}
\usepackage{paralist}
\usepackage{braket}
\usepackage{bbold}
\usepackage{graphicx}
\usepackage{placeins}
\usepackage{wrapfig}
\usepackage{dsfont}
\usepackage{multicol}
\usepackage{nccmath}
\usepackage{longtable}
\usepackage{caption}
\usepackage{subcaption}
\usepackage{hyperref}

\theoremstyle{plain} 
\newtheorem{thm}{Theorem}[section]
\newtheorem{lem}[thm]{Lemma}
\newtheorem{cor}[thm]{Corollary}
\newtheorem{prop}[thm]{Proposition}
\newtheorem{clm}[thm]{Claim}
\theoremstyle{definition}
\newtheorem{defn}{Definition}[section]
\newtheorem{expl}{Example}[section]
\theoremstyle{remark}
\newtheorem{rem}[thm]{Remark}

\newcommand{\C}{\mathds{C}}
\newcommand{\R}{\mathds{R}}

\newcommand{\Z}{\mathds{Z}}

\newcommand{\End}{\mathrm{End}}
\newcommand{\ketbra}[1]{\ket{#1}\bra{#1}}
\newcommand{\tr}{\mathrm{tr}}
\newcommand{\supp}{\mathrm{supp}}

\begin{document}
\begin{titlepage}
\begin{center}

\textsc{\LARGE ETH Zürich}\\[1.5cm] 
\textsc{\Large Master Thesis}\\[1cm] 

\huge \textbf{Computing Entanglement Polytopes}\\[0.75cm] 

\large \textit{using geometrical, algebraical and numerical methods} \\[2cm]
\begin{minipage}{0.4\textwidth}
\begin{flushleft} \large
\emph{Author:}\\
{Konstantin Wernli} 
\end{flushleft}
\end{minipage}
\begin{minipage}{0.4\textwidth}
\begin{flushright} \large
\emph{Supervisor:} \\
{Matthias Christandl} 
\end{flushright}
\end{minipage}\\[3cm]
 
\vfill
 
{\large \today}\\[4cm] 

\vfill
\end{center}

\end{titlepage}
\begin{abstract}
In \cite{Walter2012} entanglement polytopes where introduced as a coarsening of the SLOCC classification of multipartite entanglement. The advantages of classifying entanglement by entanglement polytopes are a finite hierarchy for all dimensions and a number of parameters linear in system size. In \cite{Walter2012} a method to compute entanglement polytopes using geometric invariant theory is presented. In this thesis we consider alternative methods to compute them. Some geometrical and algebraical tools are presented that can be used to compute inequalities giving an outer approximation of the entanglement polytopes. Furthermore we present a numerical method which, in theory, can compute the entanglement polytope of any given SLOCC class given a representative. Using it we classify the entanglement polytopes of $2 \times 3 \times N$ systems.
\end{abstract}
\tableofcontents

\chapter{Introduction}
Entanglement is one of the key features that distinguishes classical (and relativistic) physics from quantum physics in a way that has puzzled also the greatest minds in history \cite{Einstein1935}. In more recent years, it has gained interest as a fundamental resource in Quantum Information Protocols such as Quantum Key Distribution \cite{Bennett1984}, Quantum Teleportation \cite{Bennett1993}, and Superdense Coding \cite{Bennett1992}.
In \cite{Greenberger1990}, Greenberger, Zeilinger, Horne and Shimony argued that multipartite entanglement violates classical principles even stronger than bipartite entanglement as proposed by Bell \cite{Bell1964}. Moreover, \cite{Raussendorf2001} suggests that multiparticle entanglement could also be used as a resource of Quantum Information Theory. 
It is therefore desirable to study the properties of multipartite entanglement. This has been done by attempting to classify multipartite entanglement systematically by equivalence under Stochastic Local Operations and Classical Communication, or SLOCC for short (a detailed list of references can be found in \cite{Walter2012}). However, this way of classification is faced with the problem that in general there are infinitely many classes distinguished by an amount of parameters exponential in the particle number. \\
In \cite{Walter2012}, Christandl, Walter, Doran and Gross suggested entanglement polytopes as a way of significantly coarse-graining the classification, always providing a finite hierarchy in which states can be placed by measuring their local eigenvalues.  They also proposed a method which can theoretically be used to compute all polytopes in general dimensions based on covariants of the group action used to define SLOCC, and used it to compute entanglement polytopes of systems of three and four qubits, see also Figure \ref{fig:3qubitpol}. \\

The aim of thesis is to try and develop alternative methods to compute entanglement polytopes. More precisely, we are looking for answers to the following questions: 
\begin{itemize}
\item First of all, are there any other ways to gain information about the entanglement polytopes? 
\item If so, how efficient and how general are they?
\item How much information about the polytopes can be gained from those methods? 
\end{itemize}
In the first part of the thesis we show that the answer to the first question is positive, providing some geometrical and algebraic method to retrieve information about the polytopes. We then apply the methods to some low-dimensional systems. 
In the second part of the thesis, we discuss a numerical method to compute entanglement polytopes, and use it to completely calculate all polytopes of systems of size $2 \times 3 \times N$. Moreover, some results on higher dimensional systems are given. 
\begin{figure}
\includegraphics[width=\textwidth]{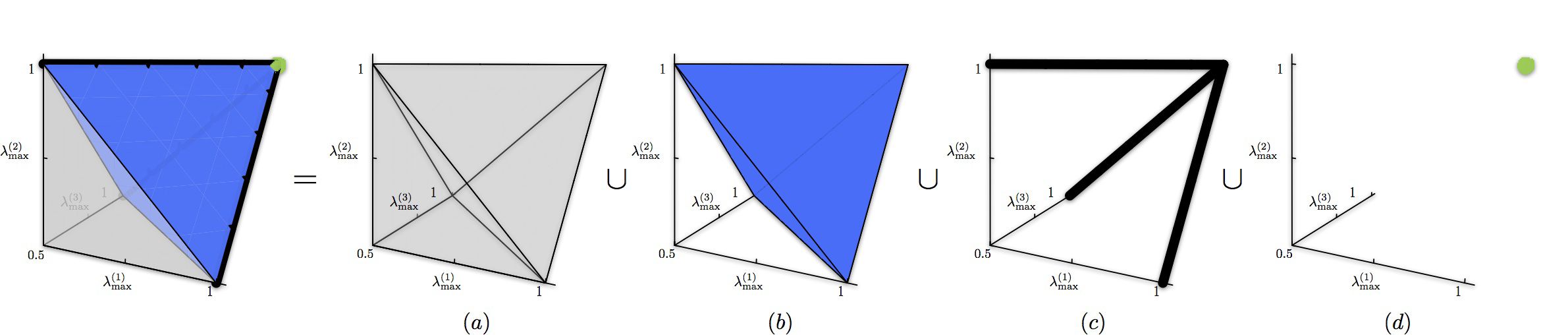}
\caption[Entanglement Polytopes for Three Qubits]{Entanglement polytopes for three qubits: (a) GHZ polytope (entire polytope, i.e., upper and lower pyramid), (b) W
polytope (upper pyramid), (c) three polytopes corresponding to EPR pairs shared between any two of the three parties (three
solid edges in the interior), (d) polytope of the unentangled states (interior vertex). Figure taken from \cite{Walter2012}}
\label{fig:3qubitpol}
\end{figure}
\chapter{Preliminaries}\label{ch:entpoly}
In this chapter we will fix some notations and review the basic facts about SLOCC orbits and entanglement polytopes  needed for this thesis along the lines of \cite{Walter2012}. Furthermore we give an ever so brief introduction to convexity, a feature which will prove very important in the calculation of entanglement polytopes. 
\section{SLOCC entanglement classes}

\subsection{Pure states}
Let $\mathcal{H} = \C^{d_1} \otimes \C^{d_2} \otimes \cdots \otimes \C^{d_N}$ be the space of pure states of a system of $N$ distinguishable particles, where particle $i$ has $d_i$ degrees of freedom. In this thesis we will consider the $N = 3$ case and the $N$-qubit case $(\C^2)^{\otimes N}$. Mathematically, it is convenient\footnote{It means one does not have to assume normalised states all the time which turns out to be handy sometimes. We will treat this issue very sloppily in general, differing between vectors in Hilbert space, lines in projective space and projectors as density operators only when necessary.} to consider a pure state to be an element of the projective space $P(\mathcal{H})$ over $\mathcal{H}$, i.e a line $l$ through the origin. Usually we identify such a line with any non-zero vector $\ket{\psi} \in l$ (and refer to it as $\ket{\psi} \in \mathcal{H}$), such a vector is called a representative of the line. The projetive space $P(\mathcal{H})$ can be seen as a subset of the space of all density operators (positive hermitian operators of trace 1) on $\mathcal{H}$ by identifying a line through $\ket{\psi}$ with the normalised projector $\frac{\ketbra{\psi}}{|\braket{\psi|\psi}|}$ onto it, thus we are able to identify a pure state with the corresponding density operator $\rho = \ketbra{\psi}$. Throughout this thesis we will use both viewpoints without further notice. 
 
\subsection{Entanglement and entanglement classes}
A pure state $\ket{\psi}$ is called \emph{entangled} if it cannot be written as a product state $\ket{\psi_1} \otimes \cdots \otimes \ket{\psi_N}$ with $\ket{\psi_i} \in \C^{d_i}$ (\cite{Nielsen2000})\footnote{In our convention one actually would have to check that this definition and all the others to follow are independent of the representative. However, this is a dull task from which want to save the reader.}. Two pure states $\ket{\psi}$ and $\ket{\varphi}$ are considered  \emph{equivalently entangled under Stochastic Local Operations and Classical Communication} if each can be converted into the other with a non-zero probability using only local operations and classical communication, see \cite{Duer2000}, where it has also been shown that this is equivalent to the existence of an Invertibe Local Operator (ILO) $A^{(1)} \otimes \cdots \otimes A^{(N)}$ such that 
\begin{equation}\label{eq:ILO} A^{(1)} \otimes \cdots \otimes A^{(N)}\ket{\psi} = \ket{\varphi}.\end{equation}
This defines an equivalence relation the set of pure states and therefore divides it into equivalence classes $\mathcal{C}$ called \emph{SLOCC entanglement classes}. Mathematically, these classes can be understood as the \emph{orbits} of the action of $G = SL(d_1) \times \ldots \times SL(d_N)$ on the set of pure states. If $g = (g^{(1)},\ldots,g^{(N)})$, then this action is given by \footnote{It is convenient to define it in this way even though normalisation is of course unnecessary.}
\begin{equation}g \cdot \ket{\psi} = \frac{ (g^{(1)} \otimes \cdots \otimes g^{(N)})\ket{\psi} } {\|(g^{(1)} \otimes \cdots \otimes g^{(N)})\ket{\psi}\|} \label{eq:action} \end{equation}
We will therefore also use the name SLOCC orbits from time to time. The orbit $G \cdot \ket{\psi}$ of a state $\ket{\psi}$ is called the \emph{entanglement class of } $\ket{\psi}$. \\
 
While the entanglement classes $\mathcal{C} = G \cdot\ket{\psi}$ partition the set of pure states, there is a natural hierarchy on the closures of the entanglement classes $\overline{\mathcal{C}} = \overline{G \cdot\ket{\psi}}$ coming from their nature as orbit closures: If $\ket{\psi'} \in \overline{G \cdot \ket{\psi}}$ then $\overline{G \cdot\ket{\psi'}} \subseteq \overline{G\cdot\ket{\psi}}$ is immediate. \\
A considerable amount of research has been put into the classification of multipartite entanglement under SLOCC, see the references in \cite{Walter2012}. In this thesis we will use results by Miyake published in  \cite{Miyake2004} on the SLOCC classes of $ 2 \times 2 \times N $ systems and work of Chen et al. \cite{Chen2006, Chen2009} on the classification of orbits in more general $2 \times M \times N$ systems. 
The main problem of the SLOCC classification is that even though it is already a coarsening of the LOCC classification it is still ``too fine'' for larger systems. As shown by a quick dimension argument, the dimension of the orbit space and thus the number of parameters labeling the SLOCC classes will grow exponentially for large N, but also for higher-level systems of a small number of particles. For example, the dimension of the orbit space is lower bounded by $2^N -1 - 3N$ for $N$ qubit systems.  
Here the entanglement polytopes described in \cite{Walter2012} come in. Before turning to them let us make a brief digression into convexity.
\section{A brief introduction to convexity}
We will state some definitions and results in this section that will provide useful in the rest of the thesis. For a detailed introduction to convexity and proofs see \cite{Barvinok2002}.
\subsection{Basics}
Let $V$ be a real vector space and $A = \{x_1,\ldots,x_n\}\subseteq V$ be a finite subset. A convex combination of $x_1,\ldots,x_n$ is a sum $\sum_{i=1}^n \lambda_i x_i$ where $ 0 \leq \lambda_i \leq 1$ and $\sum_{i} \lambda_i = 1$. 
The set of all convex combinations of the $x_i$ is called the \emph{convex hull} $\mathrm{conv}(A)$, and for any subset $B \subseteq V$, the \emph{convex hull} $\mathrm{conv}(B)$ is the set of all convex combinations of a finite number of points in $B$. \\
The set of all convex combinations of two points $x,y$ is called the \emph{segment} between $x$ and $y$ and is denoted by $[x,y]$. A subset $A \subseteq V$ is \emph{convex} if for any two points $x,y \in A$, $[x,y] \in A$ 
The \emph{extremal points} or \emph{vertices} of a convex set $A$ are those points that cannot be expressed as a proper convex combination of two other points. Under reasonable assumptions and in particular for finite-dimensional vector spaces, 
the so-called Krein-Milman theorem says that any convex set equal the (closure of the) convex hull of its extremal points. 
\subsection{Polytopes}
A (convex)\footnote{In this thesis, we use the term ``Polytope'' as equivalent to ``convex Polytope'', since we are only dealing with convex polytopes.}\emph{Polytope} in euclidean space is a finite intersection of half-spaces given by linear inequalities $l(x)+a \leq 0$  . This does not require the polytopes to be bounded, however, our polytopes always will be. The intersection of a hyperplane $l(x)+a = 0$ with the polytope is called a \emph{facet} or \emph{face} of the polytope. \\

When trying to compute an entanglement polytope $\Delta$ in this thesis, we typically are confronted with the following situation: We have computed an \emph{outer} approximation of a polytope $\Delta' \supseteq \Delta$ by means of a set of linear inequalities $\mathcal{I}$ , and an \emph{inner} approximation by means of a list of vertices $\{v_1,\ldots,v_n\}$, and we would like to show that the two coincide: $\Delta' = \mathrm{conv}(\{v_1,\ldots,v_n\})$. To this end we have to compute the extremal points $\{w_1,\ldots,w_m\}$ of the outer approximation of our polytope and show that they are actually contained in the convex hull of the vertices we already know are included. Since we know they are extremal points we also know that if we have equality, then we must actually have $\{w_1,\ldots,w_m\} = \{v_1,\ldots,v_n\}$. Alternatively, one can calculate $\{w_1,\ldots,w_m\}$ and try to show they are included in the polytope. If this is true, we know that the polytope is exactly equal to $\bigcap \mathcal{I}$. In theory, one can compute the vertices by hand: Since they must lie in the boundaries of the halfspaces, we can take any subset of $\mathcal{I}$ of cardinality $dim V$, and then check whether the corresponding hyperplanes intersect in a single point. If so, we can check the point against all other inequalities in $\mathcal{I}$. Of course this is an inhuman amount of calculation. Fortunately, there exists a variety of elaborate algorithms and programs. In this thesis we have used \texttt{qhull}\footnote{see http://www.qhull.org and \cite{Barber1996}.}. \\

To compute such inequalities we sometimes use the following result: Given a convex set $A$ and a point $p \notin A$, there exists a seperating linear functional $l(x)$ such that $l(p) < c$ and $l(A) \subseteq [c,\infty)$. Now in  euclidean space there is a canonical way of computing such a linear functional and the corresponding $c$. Namely, assume that we have found a point $x \in A$ such that $\|x-p\| = \min_{y \in A}\|y-p\|$\footnote{It can be shown such  points always exists. One of the goals of this thesis is to develop methods to find them in the entanglement polytopes.}. Denote $n:= x-p$. Then we have that for all $y \in p$, $\langle y-p , x - p \rangle \geq \langle x - p , x - p\rangle.$\\

\section{Entanglement Polytopes}
Entanglement Polytopes are a tool introduced in \cite{Walter2012} to classify entanglement in multiparticle systems. In this section, we will define them in physical and mathematical terms and analyse their advantages compared the full SLOCC classification.
\subsection{A physical definition of entanglement polytopes}
Entanglement polytopes have a very nice description by means of physical parameters. Namely, let $\ket{\psi}\in\mathcal{H}$ be a pure state and $\overline{\mathcal{C}} = \overline{G \cdot \ket{\psi}}$. Denote by $${\lambda(\ket{\psi})} = \left(\lambda^{(1)}, \ldots,\lambda^{(N)}\right)$$ the ordered local spectra of $\rho = \ketbra{\psi}$, i.e. $\lambda^{(i)} = \sigma(\rho^{(i)})$, the spectrum of th $i$-th particle reduced density matrix $\rho^{(i)} = tr_{j\neq i}\rho$. The \emph{entanglement polytope} of an entanglement class $\overline{\mathcal{C}}$ is by definition 
\[ \Delta_{\mathcal{C}}=\{\lambda(\ket{\psi}): \ket{\psi} \in \overline{\mathcal{C}}\} \]
Also, for $\overline{\mathcal{C}} = \overline{G\cdot\ket{\psi}}$ we often denote $\Delta_{\mathcal{C}}$ by $\Delta_{\ket{\psi}}$ or simply $\Delta_{\psi}$.
A priori, it is not clear that this set should have any reasonable structure at all. 
\subsection{Mathematical Description of entanglement polytopes}
The mathematical description of the entanglement polytopes is a bit more involved but also gives more results. We first provide the mathematical framework, which is geometric invariant theory and its study of the moment map. Our main sources are the work by Kirwan\cite{Kirwan1984}, Brion \cite{Brion1987}, and Ness \cite{Ness1984}. \\

Let $K$ be a compact Lie Group with Lie Algebra $Lie(K)$ which acts  on a Kähler manifold\footnote{A symplectic manifold with compatible Riemannian and almost complex structures. For our purposes we always have $X \subseteq  P(\mathcal{H}). $ See e.g. the lecture notes by Ballmann \cite{Ballmann2006}} $X$ preserving the symplectic form. 
\begin{defn}
A \emph{momemt map} for the action of $K$ on $X$ is a map $\Phi: X \to Lie(K)^*$ satisfying the following to conditions: 
\begin{enumerate}[i)]
\item It is $K$-equivariant with respect to the coadjoint action $\mathrm{Ad}^*$ of $K$ on $Lie(K)*$, i.e for all $x \in X$
$$\Phi(k\cdot x) = \mathrm{Ad^*}(k)\Phi(x)$$
\item If $\omega$ denotes the symplectic form on $X$, then for any $a \in Lie(K),x \in X$ and $\xi \in T_xX$ we have that 
\[ d\Phi_x(\xi)(a) =\omega_x(\xi,a_x) \] 
where $a_x$ denotes the vector field on $X$ generated by $a$, i.e. $a_x = \frac{d}{dt}|_{t=0}\exp(ta)\cdot x$. Equivalently, one could say that for any $a \in Lie(K)$, the function $x \mapsto \Phi(x)(a)$ is a hamiltonian function for the vector field $a_x$.
\end{enumerate}
\end{defn}
Now let $G$ be the complexification of $K$ and $T$ a maximal torus in $G$. We now choose a postive Weyl chamber $\mathfrak{t}^+$ in $Lie(T)$\footnote{Details can be found in any good text on representation theory of Lie Algebra and Lie Groups, e.g. \cite{Fulton1997} or \cite{Fulton1991}.} Then, as stated by Kirwan in \cite{Kirwan1984}, for every $x \in X$ we have that $\Phi(\overline{G\cdot x})\cap \mathfrak{t}^+ \subseteq Lie(K)^*$ is a compact convex  polytope, the so-called \emph{moment polytope}. \\

We now to fit our entanglement polytopes in this setting. Therefore, we consider as above the Lie Group $G = SL(d_1)\times\ldots \times SL(d_N)$ with maximal compact subgroup $K = SU(d_1) \times \ldots \times SU(d_N)$ and $X = P(\mathcal{H})$ considered as a Kähler manifold with the Fubini-Study metric. The postive Weyl chamber $\mathfrak{t}^+$ then coincides with the set of diagonal matrices with non-increasing entries. Now we define a map $\Phi: X \to i\cdot Lie(K)$ which maps a pure state $\rho = \ketbra{\psi}$ to its tuple of reduced density matrices $
\left(\rho^{(1)},\ldots,\rho^{(N)}\right)$\footnote{Notice a slight imprecision here: The reduced density operators have trace 1, whereas the elements of $i\cdot Lie(K)^*$ are supposed to have trace 0. This can be dealt with by subtracting a constant term from the moment or simply by shifting the coordinates. The ``origin'' in our language therefore is actually a tuple of maximally mixed density matrices.}.  We want to show this is can be seen as a moment map for the action of $K$ on $X$ defined by restricting the action \ref{eq:action}. To do so we use a Lemma quoted from Kirwan's book\cite[2.7]{Kirwan1984} (adjusted to our notation):
\begin{lem}
For a compact group $K$ acting complex projective variety $ X \subseteq \mathds{P}_n$ by a homomorphism $\varphi: K \to U(n+1)$, a moment map $\Phi: X \to Lie(K)^*$ is given by 
\[ \Phi(\ket{x})(a) = \frac{\braket{x|\varphi_*(a)|x}}{2\pi i |\braket{x|x}|}\]
\end{lem}
In our case, the homomorphism $\varphi$ is given by $\varphi((U^{(1)},\ldots U^{(N)})) = U^{(1)} \otimes \cdots \otimes U^{(N)}$ and its differential $\varphi_*$ simply takes \[(u^{(1)},\ldots,u^{(N)}) \mapsto u^{(1)}\otimes \mathds{1} \cdots \mathds{1} + \ldots + \mathds{1}\otimes \cdots \otimes \mathds{1}\otimes u^{(N)},\]
i.e. the above says in our case (choose a normalised representative), by the definition of reduced density matrices, 
\begin{align*}\Phi(\ket{x})(u^{(1)},\ldots,u^{(N)}) &= \frac{1}{2\pi i}\sum_{i=1}^N\braket{x|\mathds{1}_{1\cdots i-1}\otimes u^{(i)} \otimes \mathds{1}_{i+1\cdots N} |x}\\
&= \frac{1}{2\pi i}\sum_{i=1}^N \tr(\rho^{(i)}u^{(i)}) = \frac{1}{2\pi i}\sum_{i=1}^N \tr((u^{(i)})^{\dagger}\rho^{(i)})\\
&= \frac{1}{2\pi }\left\langle \left(u^{(1)},\ldots,u^{(N)}\right),-i \left( \rho^{(1)},\ldots,\rho^{(N)}\right),\right\rangle 
\end{align*}
where $\langle U_1,U_2\rangle = tr(U_1^{\dagger}U_2)$ denotes the scalar product on $Lie(K)$. Note that the $i$ serves a ``conversion factor'' between the hermitian reduced density matrices and the antihermitian elements of $Lie(K)$. 
The only thing left show to identify entanglement polytopes as moment polytopes is that $\Phi(\overline{G\cdot \ket{\psi}}) \cap \mathfrak{t}^+$ can be identified with the collection of ordered local eigenvalues. We have already stated, however, that $\mathfrak{t}^+$ consists of diagonal matrices with weakly decreasing entries. Moreover, since for any collection of local unitaries  $U=(U^{(1)},\ldots U^{(N)}$ we have that
 $$\Phi(U\cdot\ket{\psi}) = Ad^*(U)(\Phi(\ket{\psi})) = \left(U^{(1)}\rho^{(1)}(U^{(1)})^{\dagger},\ldots,U^{(N)}\rho^{(N)}(U^{(N)})^{\dagger}\right)$$
 the orbit of $\ket{\psi}$ will always contain states $\ket{\psi'}$ for whom the reduced density matrices are exactly given by $\lambda(\ket{\psi})$ (identifying a  diagonal matrix with the vector of it entries). We conclude that the entanglement polytopes are precisely the momemt polytopes of the closures $G$-action \eqref{eq:action} on the set of pure states. \\
 
 Brion investigated on moment polytopes in \cite{Brion1987} using covariants. In \cite{Walter2012} these results are used to deduce important properties of entanglement polytopes. Three important results for this thesis are the following\cite[Corollary 2,p.8]{Walter2012}: 
 \begin{itemize}
 \item The entanglement polytopes are compact convex polytopes.
 \item There are always only finitely many entanglement polytopes.
 \item The entanglement polytope of the orbit closure of a generic state always is the full polytope $\Delta$ given by the Quantum Marginal Inequalities (the set of all possible local spectra of a quantum state).

 \end{itemize} 
 The second part is important for us since it means we can use methods put together in the previous section on convexity. \\
 
 Using the description of the entanglement polytopes by covariants,the authors of  \cite{Walter2012} describe an algorithm how to compute them. In the next chapter we will present methods on how to gain information on the polytopes without using covariants. First, however, we briefly discuss the advantages and disadvantages of entanglement polytopes in comparison with the usual SLOCC classification.
 \subsection{Advantages of entanglement polytopes}
 The first major advantage of classifying entanglement via entanglement polytopes is that there are only finitely many of them no matter how large our system becomes. 
 It follows directly from the definition of entanglement polytopes that 
 $$ \overline{\mathcal{C}} \subseteq \overline{\mathcal{D}} \Rightarrow \Delta_{\mathcal{C}} \subseteq \Delta_{\mathcal{D}},$$ i.e the inclusion hierarchy of the polytopes respects the inclusion hierarchy of the SLOCC orbit closures. Therefore, the classification of entanglement by entanglement polytopes allows for a significant coarse-graining of the SLOCC hierarchy. \\
 
 The second big advantage is their relatively simple description once they have been computed. Since they form convex polytopes, they are given by an intersection of halfspaces, i.e. giving an description of an entanglement polytope of a certain entanglement class is equivalent to giving a set of linear inequalities for the local spectra of all states in an entanglement class. Therefore, if one wants to check whether a state $\ket{\psi}$ is contained in an entanglement class $\mathcal{C}$, one measures its local eigenvalues and compares them against all the inequalities for a given entanglement polytope. If the local eigenvalues violate one of them (i.e. $\lambda(\ket{\psi}) \notin \Delta_{\mathcal{C}}$)then we conclude that $\ket{\psi} \notin \mathcal{C}$. This only requires the measurement of linearly many parameters in $N$, as opposed to the exponentially many one needs to measure to locate the entanglement class the state belongs to. Moreover, this suggests that also partial knowledge of the entanglement polytopes can be useful, as every inequality we know serves as an entanglement witness.

\chapter{Geometric and algebraic ways to compute entanglement polytopes}\label{geometrictricks}
Here we present some alternative methods on how to find any information (that is, inner and outer estimates) on entanglement polytopes. 
\section{A word on coordinates}
To give results on the polytopes we use as coordinates in the polytope the ``Most Local Eigenvalues''\footnote{named after the Mathematica function Most.}: A point in the polytope, which is a collection of diagonal matrices with trace 1, is specified by the collection of tuples of the $d_i-1$ \emph{highest} eigenvalues denoted by $x_{i,j}, j=1,\ldots,d_i -1, i = 1,\ldots , N$ the last diagonal entry on each system then equals one minus the sum of the others.
For example, the well-known GHZ state $\frac{1}{\sqrt{2}}(\ket{000}+\ket{111})$ has coordinates $(\frac{1}{2},\frac{1}{2},\frac{1}{2})$. While unpractical for calculation, these coordinates remove $N$ dimensions from the problem, thus making it more visualisable for low dimensions. However, one has to pay attention to calculate the scalar product of two such points by using \emph{all} diagonal entries. 
\section{A method to find the closest point to the origin in an entanglement polytope}
The aim of this section is to prove the following 
\begin{prop}\label{eigenvectorprop}
Let $\ket{\psi}\in \mathcal{H}$,$\rho = \ketbra{\psi}$ and 
$$ X_{\rho}:=\rho^{(1)}\otimes \mathds{1} \cdots \otimes \mathds{1}+\mathds{1}\otimes\rho^{(2)}\otimes\mathds{1}\otimes\cdots\otimes\mathds{1}+\ldots+\mathds{1}\otimes\cdots\otimes\mathds{1}\otimes\rho^{(N)}
$$
Assume $\ket{\psi}$ is an eigenvector of $X_\rho$, i.e. \begin{equation}X_{\rho}\ket{\psi}=\lambda\ket{\psi}\label{psieigenvector} \end{equation} Then $\Phi(K\cdot\ket{\psi})\cap \Delta_{\psi}$ is the (unique) closest point to the origin in $\Delta_{\psi}$.
\end{prop}
I.e. if $\psi$ satisfies equation \eqref{psieigenvector}, then the locally ordered collection of local eigenvalues will give the point in the entanglement polytope of $\ket{\psi}$ closest to the origin. The proof relies on Kirwan's book \cite{Kirwan1984}. 
\begin{clm}
Let $f = \|\Phi\|^2$. Assume equation \eqref{psieigenvector}. Then $\ket{\psi}$ is a critical point of f. 
\end{clm}
\begin{proof}
In Kirwan's book \cite[Lemma 6.6]{Kirwan1984},  we find the equation \[\mathrm{grad} f(x) = 2i\Phi(x)_x\label{gradienteq}\] (for compact K\"ahler manifolds) where we identify $\Phi(x) \in Lie(G)^*$ with an element of $Lie(G)$ using duality in euclidean space. $\Phi(x)_x$ then denotes the tangent vector of the infinitesimal action of $\Phi(x)$, i.e. for any $y\in X$
\[\Phi(x)_y = \frac{d}{dt}exp(t\Phi(x))\cdot y\big|_{t=0}\] 
In our case we have 
$e^{tX_{\rho}} = e^{t\rho^{(1)}}\otimes\cdots\otimes e^{t\rho^{(N)}}$, so by assuming equation \eqref{psieigenvector} we also get $e^{t\rho^{(1)}}\otimes\cdots\otimes e^{t\rho^{(N)}}\ket{\psi} = e^{t\lambda} \ket{\psi}$. Thus we see that 
\begin{align*}
\mathrm{grad}f(\rho) 	&= 2i\Phi(\rho)_{\rho} \\
						&= 2i\frac{d}{dt}exp(t\Phi(\rho))\cdot \rho\big|_{t=0} \\
						&= 2i\frac{d}{dt}\bigg|_{t=0}\frac{e^{t\rho^{(1)}}\otimes\cdots\otimes e^{t\rho^{(N)}}\ketbra{\psi}(e^{t\rho^{(1)}}\otimes\cdots\otimes e^{t\rho^{(N)}})^{\dagger}}{\|e^{t\rho^{(1)}}\otimes\cdots\otimes e^{t\rho^{(N)}}\ket{\psi}\|^2} \\
						&= 2i\frac{d}{dt}\bigg|_{t=0}|\frac{e^{t\lambda}|^2\ketbra{\psi}}{\|e^{t\lambda}\ket{\psi}\|^2}\\
						&=2i\frac{d}{dt}\ketbra{\psi} = 0. \\
\end{align*}
\end{proof}
Next let $T \leq G$ be the subgroup of diagonal matrices and $T\cap K = (S^1)^{d_1-1} \times \cdots \times (S^1)^{d_N-1}\subset K$ be a maximal torus in $K$. 
\begin{clm}
The set of fixed points of the action of $T\cap K$ on $P(\mathcal{H})$ is the set of computational basis vectors of $\mathcal{H}$: $$Fix(T)=\{\ket{\mathbf{j}} = \ket{j_1}\otimes\cdots\otimes\ket{j_N}: 1 \leq j_i \leq d_i, i=1,\ldots,N\}\ .$$\end{clm}
\begin{proof}
Since $T\cap K$ acts by matrices which are diagonal in the computational basis, it acts by multiplying the basis vectors by a phase. Therefore the corresponding density operator $\rho=\ketbra{\psi}$ is invariant under the action of $T$. Now let $\ket{\psi} =\sum_{\mathbf{j}} c_{\mathbf{j}}\ket{\mathbf{j}}$ be an entangled (non-seperable) state expanded in the computational basis. Then there exist $\mathbf{j} \neq \mathbf{j'}$ such that both $c_{\mathbf{j}}$ and $c_{\mathbf{j'}}$ are nonzero. Assume $j_1=0 \neq j'_1=1$(for simplicity,this can be achieved by permuting the systems and basis vectors) and take $t \in T$ to be $(\mathrm{diag}(i,-i,1,\ldots,1),\mathds{1},\ldots,\mathds{1})$.
Then $t \cdot \ket{\psi} = ic_{\mathbf{j}}\ket{j}-ic_{\mathbf{j'}}+\sum_{\mathbf{k}\neq\mathbf{j},\mathbf{j'}}c_{\mathbf{k}}t\cdot\ket{\mathbf{k}}$. Hence there is no complex constant $a$ such that $a\ket{\psi}=t\cdot\ket{\psi}$, that is, $\ket{\psi} \neq t\cdot\ket{\psi}$ and the claim is proven. 
\end{proof}
Let us denote, following Kirwan, by $\Phi_T$ the composition $X \rightarrow Lie(K)^* \rightarrow Lie(T)^*$ of the moment map with the restriction map $Lie(G)^* \rightarrow Lie(T)^*$ which maps a linear functional $\phi \in Lie(G)^*$ to its restriction $\phi\big|_{Lie(T)}$, and the set $\Phi_T(Fix(T))$ by $\mathds{A}$ According to a theorem by Atiyah \cite{Atiyah1982}, this is always a finite set and $\Phi_T(X)=\mathrm{Conv}(\mathds{A})$. In our case an element of $A$ is just given by a collection of matrices which have a single 1 on the diagonal. \\
Now we define a finite set $\mathcal{B}$ as follows. Let $\beta \in t^*$ be an element of the convex hull of some elements of $\mathds{A}$ which minimises the distance to the origin in this convex hull. For every subset of $\mathds{A}$ there is a unique such $\beta$. $\mathcal{B}$ is the set of all such $\beta$ which also lie in $t^+$, our choice of positive Weyl chamber, that is, those $\beta$ which consist of diagonal matrices with weakly decreasing entries. 
The following Lemma is from the Kirwan's book and characterises the critical points of the norm square of the moment map: 
\begin{lem}\cite[Lemma 3.5]{Kirwan1984}\label{critpts} Let $\Phi_{\beta}(x)=\Phi(x)(\beta)$. 
\[\mathrm{Crit}(f)= \bigsqcup_{\beta \in \mathcal{B}} K\cdot (Z_{\beta} \cap \Phi^{-1}(\beta))=: \bigsqcup_{\beta \in \mathcal{B}}C_{\beta}\] where $Z_{\beta} = \mathrm{Crit}(\Phi_{\beta}) \cap \Phi_{\beta}^{-1}(\|\beta\|^2)$
\end{lem}
Denote by $W^S(C_{\beta})$ the stable manifold of $C_{\beta}$ under the negative gradient flow of f, i.e. the points in $P(\mathcal{H})$ satisfying $\lim_{t \to \infty}\varphi_{\mathrm{grad}f}^t(x) \in C_{\beta}$ where $\varphi_{\mathrm{grad}f}^t$ denotes the flow of the vector field $\mathrm{grad}f$. In chapter 6 of Kirwan's book, we find the following result \cite[6.2]{Kirwan1984} 
\begin{prop}\label{mindistprop}
A point $x \in X$lies in $W^S(C_{\beta})$ if and only if $\beta$ is closest to the origin in $\Phi(\overline{G\cdot x})\cap \mathfrak{t}^+$.
\end{prop}
Assume now that $x$ is a critical point of $f$ with $\Phi(x) \in \mathfrak{t}^*$. Now $\Phi(C_{\beta})\cap \mathfrak{t^+}$ consists of a single element since the action of $K$ on the reduced density matrices can be interpreted as a local change of basis, thus leaving the eigenvalues invariant, and there is only a single element in $t^*$ for a given set of eigenvalues. We conclude, using Proposition $\ref{critpts}$, that there exists a $U\in K$ such that $\Phi(U \cdot x) = \beta$. A fortiori this means $x \in W^S(C_{\beta})$, and therefore $\beta$ is the closest point to the origin in the entanglement polytope of $x$. Since $\Phi(U \cdot x ) = \beta$, we have proved the claim of proposition \ref{mindistprop}. \\
If we have found the closest point $\lambda^*$ to the origin $O$, we automatically get an inequality for the polytope using convexity and the scalar product in $Lie(K)$, namely, for all $\lambda$ in the polytope,
$$\langle\lambda - O, \lambda^*-O\rangle \geq \langle\lambda^*-O, \lambda^*-O\rangle$$ 
However, it is a priori unclear how to guess a state satisfying \eqref{psieigenvector} . In some cases, the following method can be used.
 
\section{Free Quantum states}
Let $\mathcal{H} = \C^{d_1} \otimes \ldots \otimes \C^{d_N}$ and $\rho = \ket{\psi}\bra{\psi} \in \End(\mathcal{H})$
be a pure state. Let
\[ \ket{\mathbf{j}}:= \ket{j_1 j_2 \ldots j_N} = \ket{j_1} \otimes \ket{j_2} \otimes \ldots \otimes  \ket{j_N} \]
 where $\mathbf{j}=j_1\ldots j_N$ ranges over $j_i = 0,1,\ldots, n-1; \quad i=1, \ldots, N$, denote the canonical basis for $\mathcal{H}$. We then can expand our stat
\[ \ket{\psi} = \sum_{\mathbf{j}}c_{\mathbf{j}}\ket{\mathbf{j}}=\sum_{j_1=0}^{d_1-1}\cdots\sum_{j_N=0}^{d_N-1}c_{j_1j_2\ldots j_n}\ket{j_1\cdots j_N} \] 
in this basis, the normalisation condition tells us that $\sum_{\mathbf{j}}|c_{\mathbf{j}}|^2 = 1$.

\begin{defn}  If $\ket{\psi} = \sum c_{\mathbf{j}}\ket{\mathbf{j}}$ then the set $ \supp(\ket{\psi}):= \{\mathbf{j}:c_{\mathbf{j}}\neq 0\}$ is called the \emph{support} of $\ket{\psi}$.\end{defn} 
 
\begin{defn} A pure state $\ket{\psi}$ is said to be \emph{free} if for all $\mathbf{j},\mathbf{j'}$ such that $\ket{\mathbf{j}}, \ket{\mathbf{j'}} \in \supp(\ket{\psi})$ there exist distinct indices $i_1$ and $i_2$ such that $j_{i_1} \neq j'_{i_1}$ and $j_{i_2} \neq j'_{i_2}$,  \end{defn}
or in words: If we expand $\ket{\psi}$ in a basis, every two appearing basis vectors differ in at least two "slots". 
\begin{rem}
This is a special case of the property of a vector to have \emph{no adjacent weights} (since our weight spaces are generated by the computational basis vectors), which is a standard tool to compute inner approximations to polytopes, see \cite{Sjamaar1998} and \cite{Smirnov2004}.
\end{rem}
\begin{prop}[cf \cite{Sjamaar1998}]\label{prop:free1} Let $\ket{\psi} \in P(\mathcal{H})$ be free
 and $\rho = \ket{\psi}$. Then all its reduced density matrices
 \[ \rho^{(i)}=\tr_{l \neq i}\rho \] are diagonal in the computational (canonical) basis .
In this case,
\[\rho^{(i)} = \sum_{\mathbf{j \in \supp(\ket{\psi})}}|c_{\mathbf{j}}|^2\ketbra{j_i}\] 
\end{prop}

\begin{proof}
 Let $\ket{\psi} \in P(\mathcal{H}).$ Expanding $\ket{\psi}$ in the computational basis, we have that 
 \[ \rho = \ketbra{\psi} = \sum_{\mathbf{j},\mathbf{k}}c_{\mathbf{j}}c_{\mathbf{k}}^*\ket{\mathbf{j}}\bra{\mathbf{k}}\]
 Now let's look at $\rho^{(1)}$ for simplicity. Then we can compute the reduced density matrix to be, using linearity of the partial trace
 \[ \rho^{(1)}= \sum_{\mathbf{j},\mathbf{k}}c_{\mathbf{j}}c_{\mathbf{k}}^* \tr_{2\ldots n} \ket{\mathbf{j}}\bra{\mathbf{k}}
 =\sum_{\mathbf{j},\mathbf{k}}c_{\mathbf{j}}c_{\mathbf{k}}^*\tr_{2\ldots n}\ket{j_1\ldots j_N} \bra{k_1 \ldots k_N} 
= \sum_{\mathbf{j},\mathbf{k}}c_{\mathbf{j}}c_{\mathbf{k}}^*\ket{j_1}\bra{k_1}\braket{j_2\ldots j_N|k_2 \ldots k_N}\]
where the second equality follows from the definition of the partial trace. Now if $\ket{\psi}$ is free, every $\mathbf{j}$ and $\mathbf{k}$ must differ in at least \emph{two} slots, we have that $\braket{j_2\ldots j_N|k_2 \ldots k_N} = \delta_{\mathbf{j}\mathbf{k}}$. Indeed if $\mathbf{j} \neq \mathbf{k}$ then by freeness there exists at least one $i \in {2,\ldots,N}$ with $j_i \neq k_i$ so that $\braket{j_2\ldots j_N|k_2 \ldots k_N} = 0$. Hence only the terms with $\mathbf{j} = \mathbf{k}$ contribute to the above sum and we get that
\[ \rho^{(1)}=\sum_{\mathbf{j}}|c_{\mathbf{j}}|^2 \ketbra{j_1}\]
is diagonal in the computational basis. The same argument leads to the formula in the proposition.\end{proof}
This proposition implies in particular that the image of a free state under the moment map $\Phi$ lies in $Lie(T)^*$ (since all reduced density matrices are diagonal). Therefore, for free states $\ket{\psi}$, $\Phi_T(\ket{\psi}) = \Phi(\ket{\psi})$. The theorem by Atiyah quoted above says $\Phi_T(X) = \mathrm{conv}(\mathds{A})$. So, the convex hull of the images $K(\ket{\psi}):=\mathrm{conv}(\Phi_T(\{\ket{j}:\ket{j}\in\supp{\ket{\psi}}\}))$ lies inside $\Phi_T(X) \subseteq \Phi(X)$, therefore $K(\ket{\psi}) \cap \mathfrak{t}^+$ will lie inside the generic polytope. Even more is true:
\begin{prop}\label{prop:free2}
Let $\ket{\psi} = \sum_{\mathbf{j}}c_{\mathbf{j}}\ket{\mathbf{j}}$ be a free state. Then 
\[K(\ket{\psi}):=\mathrm{conv}(\Phi_T(\{\ket{\mathbf{j}}:\ket{\mathbf{j}}\in\supp{\ket{\psi}}\})) \cap \mathfrak{t}^+ \subseteq \Delta_{\ket{\psi}} = \Phi(\overline{G\cdot\ket{\psi}})\cap\mathfrak{t}^+\ ,\] i.e. the set $K(\ket{\psi})$ is completely contained in the entanglement polytope belonging to the entanglement class of $\ket{\psi}$. 
\end{prop}
\begin{proof}
By convexity, it is enough to show that $\ket{\mathbf{j}} \in \overline {G\cdot \ket{\psi}} $ for all $\ket{\mathbf{j}}\in \supp(\ket{\psi})$. This can be seen as follows: Let $\ket{\mathbf{k}} \in\supp(\ket{\psi}) $ and look at the one-parameter subgroup $H_{\mathbf{k}}$ generated by the  difference $Z_{\mathbf{k}}:=\Phi(\ketbra{\mathbf{k}})-\Phi(\ket{\psi})$ which we interpret as an element of $Lie(K)$ by identifying $Lie(K) \cong Lie(K)^*$. By \ref{prop:free1} all the reduced density matrices are diagonal 
$$\rho^{(i)} = \sum_{l=1}^{d_i}\sum_{\mathbf{j}:j_i = l} |c_{\mathbf{j}}|^2 \ketbra{l}_{i}$$
so that the exponential is simply 
$$\exp(-t\rho^{(i)}) =  \sum_{l=1}^{d_i}\exp\left(-t\sum_{\mathbf{j}:j_i = l} |c_{\mathbf{j}}|^2\right) \ketbra{l}_{i}$$
Therefore, 
\begin{align*} 
e^{-t\rho^{(1)}}\otimes\cdots\otimes e^{-t\rho^{(N)}} \ket{\mathbf{m}} &=  e^{-t\rho^{(1)}}\otimes\cdots\otimes e^{-t\rho^{(N)}}\ket{m_1}\otimes\cdots\otimes\ket{m_N} \\ &= \sum_{l=1}^{d_1}\exp\left(-t\sum_{\mathbf{j}:j_1 = l} |c_{\mathbf{j}}|^2\right)   \ketbra{l}_{1}\ket{m_1}\otimes\cdots\otimes\sum_{l=1}^{d_N}\exp\left(-t\sum_{\mathbf{j}:j_N = l} |c_{\mathbf{j}}|^2\right) \ketbra{l}_{N}\ket{m_N} \\
&=\prod_{i=1}^{N} \exp\left(-t\sum_{\mathbf{j}:j_i = m_i}|c_{\mathbf{j}}|^2\right)\ket{\mathbf{m}}\\
&= \exp\left(-t\sum_{i=1}^N\sum_{\mathbf{j}:j_i = m_i}|c_{\mathbf{j}}|^2\right)\ket{\mathbf{m}} \\
&= \exp\left(-ta_{\mathbf{m}}\right)\ket{\mathbf{m}}
\end{align*}
where we shortened $1 \geq \sum_{i=1}^N\sum_{\mathbf{j}:j_i = m_i}|c_{\mathbf{j}}|^2 =:a_{\mathbf{m}} >0$.
Furthermore, let $\sigma = \ketbra{\mathbf{k}}$.Then we have that $\sigma^{(i)} = \mathrm{diag}(e_{k_i})$ and $\exp(t\sigma^{(i)})=\mathrm{diag}(1,1,\ldots,1,\underbrace{e^t}_{k_i},1,\ldots,1)$. This means that 
$$e^{t\sigma^{(1)}}\otimes\cdots\otimes e^{t\sigma^{(N)}} \ket{\mathbf{m}} =  e^{t\sigma^{(1)}}\otimes\cdots\otimes e^{t\sigma^{(N)}}\ket{m_1}\otimes\cdots\otimes\ket{m_N} = \exp\left(t\sum_{i=1}^N{\delta_{k_im_i}}\right)\ket{\mathbf{m}}$$
Now the action of $H_{\mathbf{k}}$ on $\ket{\psi}$ is, using the above
\begin{align*}
\exp(t\Phi(\sigma)-t\Phi(\ket{\psi}))\cdot \ket{\psi} &= \exp(t\Phi(\sigma)-t\Phi(\ket{\psi}))\cdot \sum_{\mathbf{m}}c_{\mathbf{m}}\ket{\mathbf{m}} \\
&=\sum_{\mathbf{m}}c_{\mathbf{m}}\frac{\left(e^{t\sigma^{(1)}}\otimes\cdots\otimes e^{t\sigma^{(N)}}\right)\left(e^{-t\rho^{(1)}}\otimes\cdots\otimes e^{-t\rho^{(N)}}\right)\ket{\mathbf{m}}}{\|\sum_{\mathbf{m}}c_{\mathbf{m}}\left(e^{t\sigma^{(1)}}\otimes\cdots\otimes e^{t\sigma^{(N)}}\right)\left(e^{-t\rho^{(1)}}\otimes\cdots\otimes e^{-t\rho^{(N)}}\right)\ket{\mathbf{m}}\|} \\
&=\sum_{\mathbf{m}}c_{\mathbf{m}}\frac{\left(e^{t\sigma^{(1)}}\otimes\cdots\otimes e^{t\sigma^{(N)}}\right)\exp(-ta_{\mathbf{m}})\ket{\mathbf{m}}}{\|\sum_{\mathbf{m}}c_{\mathbf{m}}\left(e^{t\sigma^{(1)}}\otimes\cdots\otimes e^{t\sigma^{(N)}}\right)\exp(-ta_{\mathbf{m}})\ket{\mathbf{m}}\|} \\
&= \sum_{\mathbf{m}}c_{\mathbf{m}} \frac {\exp(t(\sum_{i=1}^N \delta_{k_im_i}-a_{\mathbf{m}}))\ket{\mathbf{m}}}{\|\sum_{\mathbf{m}}c_{\mathbf{m}} \exp(t(\sum_{i=1}^N \delta_{k_im_i}-a_{\mathbf{m}}))\ket{\mathbf{m}}\|}
\end{align*}
Now as $t \to \infty$, we see that the denominator is dominated by the coefficient which belongs to $\mathbf{k}$, namely $e^{t(N-a_k)}$, because by our freeness assumption $\sum_{i=1}^N \delta_{k_im_i} \leq N-2$ for all vectors $\ket{\mathbf{m}} \neq \ket{ \mathbf{j}}$. Therefore, in the sum all summands except the one belonging to $\ket{\mathbf{k}}$ will be exponentially suppressed, so that 
$$\lim_{t \to \infty}\exp(t(\Phi(\sigma)-\Phi(\ket{\psi})))\cdot \ket{\psi} = \sigma = \ketbra{\mathbf{k}}$$ 
which means $\ket{\mathbf{k}} \in \overline{G \cdot \ket{\psi}}$ and our claim is proven. 
\end{proof}
\begin{rem} This is nothing but an analytical verification of the more abstract argument given in chapter 7 of \cite{Sjamaar1998}.
\end{rem}
There is a useful application to these results. If we have a free quantum state $\ket{\psi}$, we can try to find a point in its convex set $ \ket{\psi}' \in K(\ket{\psi})$ which minimises distance to the origin in the full entanglement polytope of $\ket{\psi}$. By convexity, this implies that the entanglement polytope of $\ket{\psi}$ must be contained in the halfspace defined by the hyperplane orthogonal (with respect to the scalar product on $Lie(K)$) to the vector distance vector from $\Phi(\ket{\psi}')$ to the origin. This results in the following 
\begin{prop}[Criterion for minimal distance]\label{prop:mindistcrit}
Let $\ket{\psi} = \sum_{\mathbf{j}}c_{\mathbf{j}}\ket{\mathbf{j}}$ be a free quantum state. Then $\Phi(\ketbra{\psi})$ minimises the distance to the origin in $\Delta_{\ket{\psi}}$ if and only if there exists a $\lambda \in \C $ such that for all $\mathbf{j} \in \supp(\ket{\psi})$ we have that
$$\sum_{l=1}^n\sum_{\mathbf{k} \in \supp(\ket{\psi}):k_l = j_l}|c_\mathbf{k}|^2 = \lambda \label{eq:mindistcrit}$$
\end{prop}
\begin{proof}
We have to show that this is equivalent  to $X_{\rho}\ket{\psi} = \lambda\ket{\psi}$. Let $M_i$ be the operator $\mathds{1}\otimes\ldots\otimes\mathds{1}\otimes\rho^{(i)}\otimes\mathds{1}\ldots\otimes\mathds{1}$.  Then, by Proposition \ref{prop:free1}, 
\begin{equation*}M_i \ket{\mathbf{j}} = \left(\mathds{1}\otimes\ldots\otimes\mathds{1}\otimes\left(\sum_{\mathbf{k} \in \supp(\ket{\psi})}|c_{\mathbf{k}}|^2\ketbra{k_i}\right)\otimes\mathds{1}\ldots\otimes\mathds{1}\right)\ket{\mathbf{k}} = \sum_{\mathbf{k}\in \supp(\ket{\psi}):k_i = j_i}|c_{\mathbf{k}}|^2\ket{\mathbf{j}}
\end{equation*}
Since $X_{\rho} = \sum_{l=1}^n M_l$, the proposition follows. 
\end{proof}
One can further simplify this equations if one rewrites them in matrix-vector form. 
\begin{cor}\label{cor:mindist}
For a free state $\ket{\psi}$ let $m= \#\supp(\ket{\psi})$ and  label the product states $\ket{\mathbf{j}}$ in $\supp(\ket{\psi})$ as $\{\mathbf{j}^{(1)},\ldots,\mathbf{j}^{(m)}$. Define an $m \times m$ Matrix $A = A(\ket{\psi})$ by 
$$A = (a_{ik})= \left( \# \{l|\mathbf{j}^{(i)}_l=\mathbf{j}^{(k)}_l\}\right),$$
i.e. the $ik$-th entry is the number of slots in which $\mathbf{j}^{(i)}$ and $\mathbf{j}^{(k)}$ coincide. Then \ref{eq:mindistcrit} is equivalent to $$A \begin{pmatrix} |c_{\mathbf{j}^{(1)}}|^2 \\ \vdots \\ |c_{\mathbf{j}^{(m)}}|^2\end{pmatrix} = \lambda\begin{pmatrix} 1 \\ \vdots \\ 1\end{pmatrix}$$
\end{cor}
\begin{proof}
One simply rewrites the left-hand side of \ref{eq:mindistcrit} as 
$$ \sum_{l=1}^n\sum_{\mathbf{k} \in \supp(\ket{\psi}):k_l = j_l}|c_{\mathbf{k}}|^2 = \sum_{\mathbf{k} \in \supp(\ket{\psi})}\sum_{l=1}^n\delta_{j_lk_l}|c_{\mathbf{k}}|^2 = \sum_{\mathbf{k} \in \supp(\ket{\psi})} \#\{l|\mathbf{j}_l=\mathbf{k}_l\}  |c_{\mathbf{k}}|^2  $$ 
\end{proof}
\begin{expl}
\begin{enumerate}[i)]
\item For the GHZ-state this matrix is simply 
$$A(\ket{GHZ}) = \begin{pmatrix} 3 & 0 \\ 0 & 3 \end{pmatrix}$$ 
\item Consider a non-normalised free state $\ket{\psi} =\ket{100}+\ket{010}+\ket{001}+\ket{121}+\ket{112}$. Then 
$$A(\ket{\psi}) = \begin{pmatrix} 3 & 1 & 1 & 1 & 1 \\ 1 & 3 & 1 & 0 & 1 \\ 1 & 1 & 3 & 1 & 0 \\ 1 & 0 & 1 & 3 & 1 \\ 1 & 1 & 0 & 1 & 3 \end{pmatrix}$$
\end{enumerate}
\end{expl}
Now we need one final result. 
\begin{thm}\label{thm:magiclemma}
Suppose that $\ket{\psi}$ is free and $\supp(\ket{\psi})) \leq \sum_{i=i}^N(d_i-1)$, where $d_i = \mathrm{rank}(\rho^{(i)})$ is the rank of the $i$-th reduced density matrix of $\rho = \ketbra{\psi}$. Then $T$ acts transitively on the set $$M = \left\{ \ketbra{\psi'} \in P(\mathcal{H}): \supp(\ket{\psi'}) = \supp(\ket{\psi}) \right\}.$$
\end{thm}
\begin{proof}
It is clear that the action of $T$ takes $M$ to $M$. For simplicity, denote $S:=\supp(\ket{\psi})$ and $m:= \#S$ Let $\ket{\psi'} = \sum_{\mathbf{j} \in S}c'_{\mathbf{j}}\ket{\mathbf{j}} $ be in the affine cone over $M$. By moving $\ket{\psi}$ and $\ket{\psi'}$ in the orbit of $T_{\R}$, assume that all $c_{\mathbf{j}}$ and $c'_{\mathbf{j}}$ are real and positive, this is possible by our assumption on dimensions. Now let $t=(t^{(1)},\ldots,t^{(N)})\in iT_{\R}$ be a collection of diagonal matrices $t^{(i)}=\mathrm{diag}(t^{(i)}_1,\ldots,t^{(i)}_{d_i})$ with positive real entries. 
Then $t\ket{\mathbf{j}} = \prod_{i=1}^N t^{(i)}_{\mathbf{j_i}}$. Forgetting about the normalisation condition for $\ket{\psi'}$, we have that  $t\cdot \ket{\psi} = \ket{\psi'}$ is equivalent to $$t^{(1)}\otimes\cdots\otimes t^{(n)} \ket{\psi} =\lambda \ket{\psi'}$$ for some nonzero $\lambda$. Writing out the equation for every $\mathbf{j} \in S$ gives 
$$ c_{\mathbf{j}} \prod_{i=1}^N t^{(i)}_{\mathbf{j}_i}=\lambda c'_{\mathbf{j}} \Leftrightarrow \frac{\prod_{i=1}^N t^{(i)}_{\mathbf{j}_i}}{\lambda}=\frac{c'_{\mathbf{j}}}{c_{\mathbf{j}}} \Leftrightarrow -\log\lambda+\sum_{i=1}^N\log  t^{(i)}_{\mathbf{j}_i} = \log\frac{c'_{\mathbf{j}}}{c_{\mathbf{j}}}$$
Now replace $\log t^{(i)}_{d_i}$ by $-\sum_{l=1}^{d_i-1}\log t^{(i)}_{l}$. We can write the equations in a $m \times \sum(d_i-1)+1$ matrix $B$. The row corresponding $r_{\mathbf{j}}$ to $\ket{\mathbf{j}}$ is 
$(-1,v^{\mathbf{j}}_1,\ldots,v^{\mathbf{j}}_N)$ where $v^{\mathbf{j}}_i = \delta_{\mathbf{j}_il}\mathbf{e}^{d_i-1}_l-\delta_{\mathbf{j}_id_i-1}(1,\ldots,1)$.\footnote{$\mathbf{e}^k_l$ is the $l$-th canonical basis vector of length $k$ where $l$ starts from 0.} If we can show that the rows are linearly independent, then this system of linear equations has a solution dependent on a parameter, which we can use to set $\|\lambda\ket{\psi'}\|=1$. Consider therefore the equation $\sum_{\mathbf{j}} a_{\mathbf{j}}r_{\mathbf{j}} = 0$. We must show the only solution is $a_{\mathbf{j}}$=0. The first component of $\sum_{\mathbf{j}} a_{\mathbf{j}}r_{\mathbf{j}}$ is simply $\sum_{\mathbf{j}} a_{\mathbf{j}}$, so we conclude that $\sum_{\mathbf{j}} a_{\mathbf{j}}=0$. Now let us unravel the other components. For every $i=1,\ldots,N$ we get 
$$\sum_{\mathbf{j}}a_{\mathbf{j}}v_i^{\mathbf{j}} = 0 \Leftrightarrow \sum_{\mathbf{j}}a_{\mathbf{j}}\left(\delta_{\mathbf{j}_il}\mathbf{e}^{d_i-1}_l-\delta_{\mathbf{j}_id_i-1}(1,\ldots,1)\right)=0 \Leftrightarrow \sum_{\mathbf{j}:\mathbf{j}_i\neq d_i -1}a_{\mathbf{j}}e^{d_i-1}_{\mathbf{j}_i} = \sum_{\mathbf{j}:\mathbf{j}_i= d_i -1}a_{\mathbf{j}}(1,\ldots,1)$$.Looking at the components of the last equation for all $l=0,\ldots,d_i-2$, and using $\sum_{\mathbf{j}} a_{\mathbf{j}}=0$, we see that
\begin{equation}\sum_{j:\mathbf{j}_i=l}a_{\mathbf{j}} = \sum_{\mathbf{j}:\mathbf{j}= d_i -1}a_{\mathbf{j}} = -\sum_{\mathbf{j}:\mathbf{j}_i\neq d_i -1}a_{\mathbf{j}}\label{eq:magic1} \end{equation} Now we sum this over $l=0,\ldots,d_i-2$ and receive that 
$$\sum_{l=0}^{d_i-2}\sum_{j:\mathbf{j}_i=l}a_{\mathbf{j}} = \sum_{\mathbf{j}:\mathbf{j}_i\neq d_i -1}a_{\mathbf{j}} = -(d_i-1)\sum_{\mathbf{j}:\mathbf{j}_i\neq d_i -1}a_{\mathbf{j}}$$ Which is saying that for all $i$,  $\sum_{\mathbf{j}:\mathbf{j}_i\neq d_i -1}a_{\mathbf{j}}=0$ for all $i=1,\ldots,N$ . But then equation \ref{eq:magic1} implies that $\sum_{\mathbf{j}:\mathbf{j}_i = l}a_{\mathbf{j}}=0$ for all $l = 1,\ldots,d_i-1$.
This is equivalent to saying that the $m$ vectors $\mathbf{e}_{\mathbf{j}}=(\mathbf{e}_{j_1},\ldots,\mathbf{e}_{j_N})$ are linearly independent in $\bigoplus \C^{d_i}$. This, however, follows from the freeness assumption.
 \end{proof}
 By continuity of the action we get the following
 \begin{cor}\label{cor:magiclemma}
 Any state $\ket{\psi'}=\sum_{\mathbf{j} \in \supp(\ket{\psi})}c_{\mathbf{j}}\ket{\mathbf{j}}$ (where the $c_{\mathbf{j}}$ are possibly zero) is contained in the orbit \emph{closure} of $\ket{\psi}$. 
 \end{cor}
 This will prove useful for checking some hierarchy among the orbits.
 \begin{rem}
Note that this theorem directly implies proposition \ref{prop:free2} since we can now find a preimage under the moment map of every point in $K(\ket{\psi})$: Namely, the state $\sum_{\mathbf{j}\in \supp(\psi)} \sqrt{a_{\mathbf{j}}}\ket{\mathbf{j}}$ maps to $\sum_{\mathbf{j}\in\supp(\psi)}a_{\mathbf{j}} \Phi(\ketbra{\mathbf{j}})$.
\end{rem}
Combining the results up to this point, one can compute an inequality for a free state$\ket{\psi}$ on less than $\sum d_i - N $ product states by replacing the coefficients with variables and then solving the equation of Corollary \ref{cor:mindist} for the absolute value squares of the variable (we can therefore assume our coefficients are real and positive and thus call them $\sqrt{a},\sqrt{b},\ldots$). We thus get a state which maps in the orbit of $\ket{\psi}$ which maps to the point in the entanglement polytope of $\ket{\psi}$ closest to the origin. Using the scalar product in $Lie(K)$, we get an inequality for $\Delta_{\ket{\psi}}$. This method can be applied to many orbits in low-dimensional systems\footnote{See the Mathematica documentation on \emph{ClosestPointFinder} for the implementation and the appendix for the calculated inequalities.}. 

\section{Eigenvalue estimates}\label{section:EigValestimates}
We conclude this chapter with a sketch of an algebraic method first used by Christandl and Walter \footnote{Trying to find a simple proof for the inequality defining the three-qubit $W$ polytope while on a plane to Bogota, as I was told.} to prove inequalities for the local spectra. Let $\ket{\psi}\in \mathcal{H}$ be any state, and denote by $\lambda^i_k$ the $k-th$ biggest eigenvalue of the $i$-th reduced density matrix $\rho^{(i)}$ of $\rho = \ketbra{\psi}$. Then, by the min-max-principle from linear algebra and the definition of partial trace, 
$$ \lambda^i_k =  \max_{\substack{U \subseteq \C^{d_i}\\ \dim U = k}}\left(\min_{\substack{x \in U\\ \|x\|=1}} \braket{x|\rho^{(i)}|x}\right) =\left( \max_{\substack{U \subseteq \C^{d_i}\\ \dim U = k}}\min_{\substack{x \in U\\ \|x\|=1}} \braket{\psi|\mathds{1}_{1,\ldots,i-1}\otimes(\ketbra{x})_i\otimes\mathds{1}_{i+1,\ldots,N}|\psi}\right)$$
We can now directly state the following result: 
\begin{prop}\label{prop:eigenvalueestimates}
Let $\lambda_{k_1}^{i_1},\ldots,\lambda_{k_l}^{i_l}$ be a collection of local eigenvalues and $a_1,\ldots,a_l > 0$. Then, for any choice of subspaces $U_1\subseteq \C^{d_1},\ldots,U_l \subseteq \C^{d_n}$ with $\dim U_j = k_j$ we have that 
\begin{equation}
a_1\lambda_{k_1}^{i_1}+\ldots+a_l\lambda_{k_l}^{i_l} \geq \min_{\substack{x_1 \in U_1,\ldots,x_l \in U_l\\\|x_1\| = \ldots = \|x_l\| = 1}} \braket{\psi|\sum_{j=1}^l\left( \mathds{1}\otimes(\ketbra{x_{j}})_i\otimes\mathds{1}\right)|\psi}
\end{equation}
In particular, if $k_j = 1$ for all $j$, this simplifies to 
\begin{equation}
a_1\lambda_{1}^{i_1}+\ldots+a_l\lambda_{l}^{i_l} \geq  \braket{\psi|\sum_{j=1}^l a_j\left(\mathds{1}_{1\ldots i_j-1}\otimes(\ketbra{\Phi_{j}})\otimes\mathds{1}_{i_j+1\ldots n}\right)|\psi}
\end{equation}
for any $\Phi_j \in C^{d_{i_j}}$ with $\|\Phi_j\| = 1.$
\end{prop}
Now, if one has an inequality conjectured for the orbit of a certain state, one can try to apply the above proposition to a general state in the orbit. An example of how this is done can be found in section \ref{section:W3polytope}. 
\chapter{Exact results for low-dimensional systems}
In this chapter we will compute some exact results for low-dimensional systems and many-qubit systems using the results of the last chapter. 
\section{Computation of the entanglement polytopes of $2 \times 2 \times n$ Systems}
For $2 \times 2 \times n$ systems, all polytopes can be computed explicitly. As we will see, there is only one non-trivial case. 
\subsection{Bravyi's Inequalities}
In \cite{Bravyi2004} Bravyi gives a solution in terms of neat inequalities for the one-body quantum marginal problem for mixed 2-qubit states or, equivalently, for $2\times 2 \times 4$ pure states. We can use those to compute the maximal generic polytope for $2\times 2 \times 3$ and $2\times 2 \times 4$ systems. Let us call the the qubit systems $A$ and $B$ and the 4-level system $C$ Bravyi's result can be interpreted as follows: Let $\rho_A$ and $\rho_B$ be qubit density operators and $\rho_C$ a 4-level density operator. Denote by $\lambda_A^{min}$ and $\lambda_B^{min}$ the lower eigenvalues of $\rho_A$ and $\rho_B$ respectively, and by $\mathbf{\lambda}$ the weakly ordered spectrum of $\rho_C$. Then $(\rho_A,\rho_B,\rho_C)$ can be realised as the reduced density matrices of a pure state $\ket{\psi} \in \C^2\otimes\C^2\otimes\C^4$ if and only if 
\begin{align*}
\lambda_A^{\min} &\geq \lambda_3 + \lambda_4  \\
\lambda_B^{\min} &\geq \lambda_3 + \lambda_4  \\
\lambda_A^{\min} +\lambda_B^{\min} & \geq 2\lambda_4+\lambda_3+\lambda_2 \\
|\lambda_A^{\min} - \lambda_B^{\min}| &\leq \min\{\lambda_1-\lambda_3,\lambda_2-\lambda_4\}
\end{align*}
They can be easily be rewritten for the maximal eigenvalues $\lambda_{A,B}^{\max}=1-\lambda_{A,B}^{\min}$ of the $A$ and $B$ subsystems as 
 \begin{align}
 \lambda_A^{\max} &\leq \lambda_1 + \lambda_2 \label{BravyiIneq224_1} \\
 \lambda_B^{\max} &\leq \lambda_1 + \lambda_2  \label{BravyiIneq224_2}\\
 \lambda_A^{\max} +\lambda_B^{\max} & \leq 2\lambda_1+\lambda_2+\lambda_3 \label{BravyiIneq224_3}\\
 |\lambda_A^{\max} - \lambda_B^{\max}| &\leq \min\{\lambda_1-\lambda_3,\lambda_2-\lambda_4\} \label{BravyiIneq224_4}
 \end{align}
From these inequalities one can derive the inequalities for $2\times 2 \times 3$ systems by setting $\lambda_4=0$ (and then $\lambda_3=1-\lambda_1 -\lambda_2$). Then inequalities \eqref{BravyiIneq224_1} and \eqref{BravyiIneq224_2} will stay the same, but \eqref{BravyiIneq224_3} and \eqref{BravyiIneq224_4} simplify so that we get the following four inequalities:
\begin{align}
 \lambda_A^{\max} &\leq \lambda_1 + \lambda_2 \label{BravyiIneq223_1} \\
 \lambda_B^{\max} &\leq \lambda_1 + \lambda_2  \label{BravyiIneq223_2}\\
 \lambda_A^{\max} +\lambda_B^{\max} & \leq 1+\lambda_1 \label{BravyiIneq223_3}\\
 |\lambda_A^{\max} - \lambda_B^{\max}| &\leq \min\{2\lambda_1-1+\lambda_2,\lambda_2 \}\label{BravyiIneq223_4}
 \end{align}
\subsubsection{Entanglement classes in $2\times 2 \times 3$ and $2\times 2 \times 4$ systems}
In \cite{Miyake2004} the entanglement classes of $2\times 2 \times 3$ and $2\times 2 \times 4$ systems are completely classified as follows.
\begin{wrapfigure}[21]{cH}{0.5\textwidth}
\includegraphics[width=0.49\textwidth]{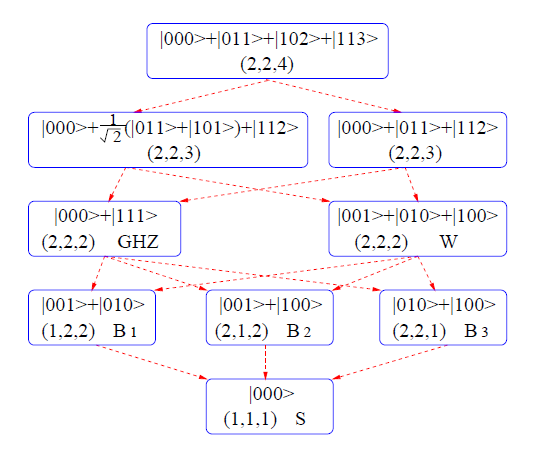}
\caption[Entanglement Hierarchy in $2\times 2\times 3$]{The 9 entanglement classes in a $2\times 2 \times 4$ system together with their hierarchy as presented in \cite{Miyake2004}. The top three entanglement classes are not named here. The class we call $W_3$ is represented by $\ket{000}+\ket{011}+\ket{112}$.}\label{ent_classes_fig}
\end{wrapfigure}
 
There are 9 classes.6 of them are the embedded 2-qubit classes, then there are non-equivalent classe of $2\times 2 \times 3$-entanglement: A generic class and another, non-generic class. There are many analogies of this class to the W class in the three-qubit case, so we will denote this class by $W_3$. The classes are presented in figure \ref{ent_classes_fig} and table \ref{ent_class_table}, both taken from \cite{Miyake2004}.
\begin{table}
\centering
\begin{tabular}{|c|c|c|c|c|}
\hline Class Name & Representative  & Local ranks   \\ 
\hline Generic $2\times 2\times 4$& $\ket{000}+\ket{011}+\ket{102}+\ket{113}$ & $(2,2,4)$    \\ 
\hline Generic $2\times 2\times 3$ & $\ket{000}+\frac{1}{\sqrt{2}}(\ket{011}+\ket{101})+\ket{112}$ & $(2,2,3)$  \\ 
\hline $W_3$  & $\ket{000}+\ket{011}+\ket{112}$ & $(2,2,3)$  \\ 
\hline GHZ  & $\ket{000}+\ket{111}$ &$(2,2,2) $  \\ 
\hline W & $\ket{001}+\ket{010}+\ket{100} $ & $(2,2,2) $  \\ 
\hline B1 & $\ket{000}+\ket{010}$ & $(1,2,2) $ \\ 
\hline B2& $\ket{100}+\ket{001}$ & $(2,1,2)$   \\ 
\hline B3 & $\ket{100}+\ket{010}$ & $(2,2,1)$  \\ 
\hline Seperable & $\ket{000}$ & $(1,1,1)$  \\ 
\hline
\end{tabular} 
\caption{Names and local ranks of the 9 entanglement classes}\label{ent_class_table}
\end{table}
In particular, this 
tells us we just have to find one additional polytope: Namely, the one of the $W_3$ class. The polytopes of the generic classes are given by Bravyi's Inequalities from the previous section and all polytopes of the three-qubit classes have been calculated in \cite{Walter2012}.

\FloatBarrier

\subsection{The entanglement polytopes of a $2 \times 2 \times 3$ system}
Let us now change our viewpoint slightly. We want to look at the polytopes as convex subsets of $\R^4$ by using the ``Most Local Eigenvalues'' coordinates. 
\subsubsection{The generic polytope} 
As discussed above the generic polytope can be in principle read off directly from Bravyi's Inequalities. However, in these the eigenvalues were assumed to be ordered and to belong to density operators, so we will have to implement these assumptions by adding some $\emph{local}$ inequalities to Bravyi's Inequalities \eqref{BravyiIneq223_1} to \eqref{BravyiIneq223_4}.  On the first two subsystems, for example,  this can be achieved by the inequalities
\begin{equation}\frac{1}{2} \leq x_{1,1},x_{2,1} \qquad \mathrm{and} \qquad x_{1,1},x_{2,1} \leq 1 \label{ineq:local223_1}\end{equation}
where we the first inequality is a \emph{Weyl chamber condition} ($x_{1,1} \geq 1-x_{1,1}$) and the second one is a \emph{density operator condtion} ($1-x_{1,1} \geq 0$). On the third subsystem, the density operator conditions are $x_{3,1} \leq 1$ and $x_{3,1}+x_{3,2} \leq 1$ and the Weyl chamber conditions are given by $x_{3,1} \geq x_{3,2} \geq 1 -x_{3,1} - x_{3,2}$. Let us summarise the local and global inequalities:

\begin{align*}
&\mathbf{Global} &&\mathbf{Local}\\
x_{1,1},x_{2,1} &\leq x_{3,1}+x_{3,2}  &\frac{1}{2}&\leq x_{1,1},x_{2,1} \leq 1 \\
x_{1,1} + x_{2,1} &\leq 1+x_{3,1}  &\frac{1}{3}&\leq x_{3,1} \leq 1 \\
|x_{1,1}-x_{2,1}|&\leq\min\{2x_{3,1}+x_{3,2}-1,x_{3,2}\} &\frac{1-x_{3,1}}{2} &\leq x_{3,2} \leq x_{3,1} \\
&   &x_{3,1}+x_{3,2} &\leq 1
\end{align*}
There still is some redundancy among these inequalities. 
\begin{clm}
The following is a defining non-redundant set of inequalities for the entanglement polytope of a generic state on a $2\times 2\times 3$-system:
\begin{align}
\frac{1}{2} \leq x_{1,1}, x_{2,1} &\leq x_{3,1} + x_{3,2} \leq 1 \\
x_{1,1}+x_{2,1} &\leq 1+x_{3,1} \\
|x_{1,1} -x_{2,1}| &\leq x_{3,2} \\
|x_{1,1}-x_{2,1}| &\leq 2x_{3,1}-1 + x_{3,2} \\
\frac{1-x_{3,1}}{2} \leq x_{3,2} &\leq  x_{3,1}\\ 
\end{align}
\end{clm}
\begin{prop}\label{prop:vertices223}
The intersetion of this halfspaces defines a convex polytope with 9 vertices whose coordinates are given by:
\begin{align*} 
v_1 &= (\frac{1}{2}, \frac{1}{2},\frac{1}{3},\frac{1}{3}) \\
v_2 &=(\frac{2}{3}, \frac{2}{3},\frac{1}{3},\frac{1}{3}) \\
v_3 &= (\frac{1}{2},\frac{3}{4},\frac{1}{2},\frac{1}{4}) \\
v_4 &= (\frac{3}{4}, \frac{1}{2},\frac{1}{2},\frac{1}{4}) \\
v_5 &= (\frac{1}{2}, \frac{1}{2},\frac{1}{2},\frac{1}{2}) \\
v_6 &= (1, \frac{1}{2},\frac{1}{2},\frac{1}{2}) \\
v_7 &= (\frac{1}{2},1,\frac{1}{2},\frac{1}{2}) \\
v_{8} &= (\frac{1}{2}, \frac{1}{2},1,0) \\
v_{9} &= (1,1,1,0)\\
\end{align*}
\end{prop}
\begin{proof}[Proof of Claim and Proposition]
This can be done using a standard convex hull routine.
\end{proof}
See also figure \ref{fig:223polytopes}.
\subsubsection{The $W_3$ polytope}\label{section:W3polytope}
To calculate the polytope of the $W_3$ we have to make use of the geometric tricks from the previous section.
The represantative given in \cite{Miyake2004} is 
$\frac{1}{\sqrt{3}}\left(\ket{000}+\ket{011}+\ket{112}\right)$ (after normalisation). For convenience we choose another representative $ \ket{\psi} = \frac{1}{\sqrt{3}}\left(\ket{010}+\ket{101}+\ket{002}\right)$ which can be obtained from the first by exchanging $\ket{0} \leftrightarrow \ket{1}$ on the first two systems. This state is free, and we can use proposition \ref{prop:free1} to compute its reduced density matrices $$\rho^{(1)}= \begin{pmatrix}\frac{2}{3} & 0 \\ 0 &\frac{1}{3} \end{pmatrix}; \qquad \rho^{(2)} = \begin{pmatrix} \frac{2}{3} & 0 \\ 0 & \frac{1}{3} \end{pmatrix};\qquad \rho^{(3)} = \begin{pmatrix} \frac{1}{3} & 0 & 0 \\ 0 & \frac{1}{3} & 0 \\ 0 & 0 & \frac{1}{3} \end{pmatrix} $$
One quickly calculates that $\ket{\psi}$ is not an eigenvector of $X_{\rho}$, however, we can calculate $X_{\tau}$ for an arbitrary state $\tau=\ketbra{\phi}= \ketbra{\phi(a,b,c)}$ on these basis vectors, that is, for $\ket{\phi(a,b,c)}=\sqrt{a}\ket{010} + \sqrt{b}\ket{101}+\sqrt{c}\ket{002}$. By theorem \ref{thm:magiclemma}, all such states lie in the orbit closure of $\ket{\psi}$. 
The reduced density matrices then are
\begin{equation}\label{taurdms} \tau^{(1)} = \begin{pmatrix}a+c & 0 \\ 0 &b \end{pmatrix}; \qquad \tau^{(2)} = \begin{pmatrix} b+c & 0 \\ 0 & a \end{pmatrix};\qquad \tau^{(3)} = \begin{pmatrix} a & 0 & 0 \\ 0 & b & 0 \\ 0 & 0 & c \end{pmatrix} 
\end{equation} 
so that we get 
\begin{align*}
X_{\tau}\ket{\phi} &= \tau^{(1)}\otimes \mathds{1} \otimes \mathds{1}\ket{\phi}+\mathds{1}\otimes\tau^{(2)}\otimes\mathds{1}\ket{\phi}+\mathds{1}\otimes\mathds{1}\otimes\mathds{1}\otimes\tau^{(3)}\ket{\phi} \\
&= (a+c)\sqrt{a}\ket{010} + b\sqrt{b}\ket{101}+(a+c)\sqrt{c}\ket{002} \\
&+a\sqrt{a}\ket{010} + (b+c)\sqrt{b}\ket{101}+(b+c)\sqrt{c}\ket{002} \\
&+a\sqrt{a}\ket{010} + b\sqrt{b}\ket{101}+c\sqrt{c}\ket{002} \\
&=\sqrt{a}(3a+b)\ket{010}+\sqrt{b}(a+3b+c)\ket{101}+\sqrt{c}(b+3c)\ket{002}
\end{align*}
so we are looking for $a,b,c$ with $3a+b=a+3b+c=b+3c$. Note that these are exactly the equations given in Corollary \ref{cor:mindist}, we  One solution is $ a=c=2,b=1$ or after normalisation $a=c=\frac{2}{5},b=\frac{1}{5}$.
The reduced density matrices now are 
\begin{equation}\tau^{(1)}= \begin{pmatrix} \frac{3}{5} & 0 \\ 0 & \frac{2}{5} \end{pmatrix}; \qquad \tau^{(2)}= \begin{pmatrix} \frac{3}{5} & 0 \\ 0 & \frac{2}{5} \end{pmatrix}; \qquad \tau^{(3)}= \begin{pmatrix} \frac{2}{5} & 0 & 0\\ 0 & \frac{2}{5} & 0 \\ 0 & 0 & \frac{1}{5} \end{pmatrix} \label{RDMs1} \end{equation}
I.e. this gives the point in the entanglement polytope of the $W_3$ class closest to the origin. Now, by proposition  \ref{prop:free2}, the intersection of the convex hull of the images of the three basis vectors $\ket{\psi_1}:=\ket{010},\ket{\psi_2}:=\ket{101},\ket{\psi_3}:=\ket{002}$ with the positive Weyl chamber $\mathfrak{t}^+$ will also lie in the entanglement polytope. The point in this convex hull can be realised as $\Phi(\tau)$, with $\tau = \tau(a,b,c)$ as above, and hence is given by the matrices \ref{taurdms}. Therefore, we can directly read off the inequalities which determine whether $a\Phi(\ketbra{\psi_1})+b\Phi(\ketbra{\psi_2})+c\Phi(\ketbra{\psi_3}) \in \mathfrak{t}^+$: 
\begin{align*} 
a &\geq c \geq b = 1 - a - c  \geq 0 \ \Leftrightarrow \ a \geq c \geq \frac{1-a}{2} \\
b+c &\geq a \Leftrightarrow (1-a-c)+c \geq a \ \Leftrightarrow \ a \leq \frac{1}{2} \\
a &\geq \frac{1}{3}
\end{align*}

\begin{wrapfigure}[10]{r}{0.25\textwidth}
\includegraphics[width=0.24\textwidth]{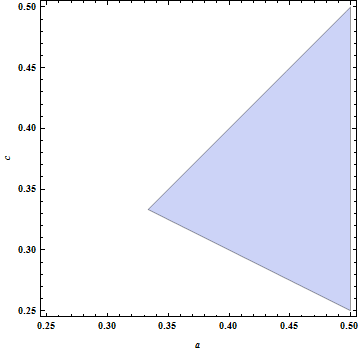}
\caption{The triangle defined by the inequalities \eqref{triangleineq1} and \eqref{triangleineq2}}\label{triangle}
\end{wrapfigure}
 Summarising we get the following two inequalities: 
\begin{align}
\frac{1}{3} \leq a &\leq \frac{1}{2} \label{triangleineq1}\\
\frac{1-a}{2} &\leq c \leq a \label{triangleineq2}
\end{align}
These inequalities define a triangle which can be seen in Figure \ref{triangle}. 

The corners of this triangle correspond to vertices of the generic polytope (and, by the remark above, also to states $\tau$), namely, $a=c=\frac{1}{3}$ corresponds to $v_2=(\frac{2}{3},\frac{2}{3},\frac{1}{3},\frac{1}{3})$, $a=c=\frac{1}{2}$ gives the vertiex corresponding to the GHZ state, and $a=\frac{1}{2}, c=\frac{1}{4}$ to $v_4=(\frac{3}{4},\frac{1}{2},\frac{1}{2},\frac{1}{4})$. Since the first two eigenvalues of $\tau'^{(3)}$ are the same, we can swap the roles of $a$ and $c$ in the inequalities above, which leads to another such triangle. This shares two vertices with the first but the third one is $v_5=(\frac{1}{2},\frac{3}{4},\frac{1}{2},\frac{1}{2})$
These 4 points define a hyperplane in $\R^4$ which is given by the equation \begin{equation}x_{1,1}+x_{2,1}+x_{3,1}+x_{3,2}=2 \ .\label{hyperplane}\end{equation}.
We therefore arrive at the following result: 
\begin{prop}
For the  entanglement polytope of the $W_3$ state the following holds: 
\begin{equation} \Delta_{W_3} \subseteq \Delta \cap \{x_{1,1}+x_{2,1}+x_{3,1}+x_{3,2} \geq 2\}
\end{equation}
Moreover, it contains all vertices of the full polytope except the origin $v_1$. 
\end{prop}
\begin{proof}
The only thing left to show is that the vertices $v_6$ to $v_{11}$ are contained in $\Delta_{W_3}$, the others were treated. However, these are the vertices of the polytope of the GHZ state, from which we know that it lies in the closure of the SLOCC Class of $W_3$ by Theorem \ref{thm:magiclemma} (with the notation above, it is $\ket{\phi(\frac{1}{2},\frac{1}{2},0)}$). Therefore its entanglement polytope, which includes $v_6$ to $v_{11}$, is contained in $\Delta_{W_3}$.
\end{proof}
\begin{figure}[b]
\centering
\begin{subfigure}[h]{0.45\textwidth}
\centering
\includegraphics[width =\textwidth]{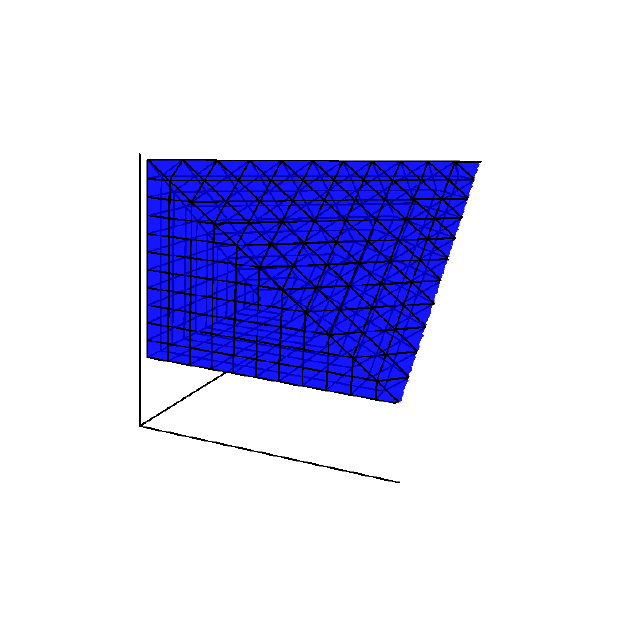}
\subcaption{Full polytope for $x_{3,3}=0$.}
\end{subfigure}
\qquad
\begin{subfigure}[h]{0.45\textwidth}
\centering
\includegraphics[width =\textwidth]{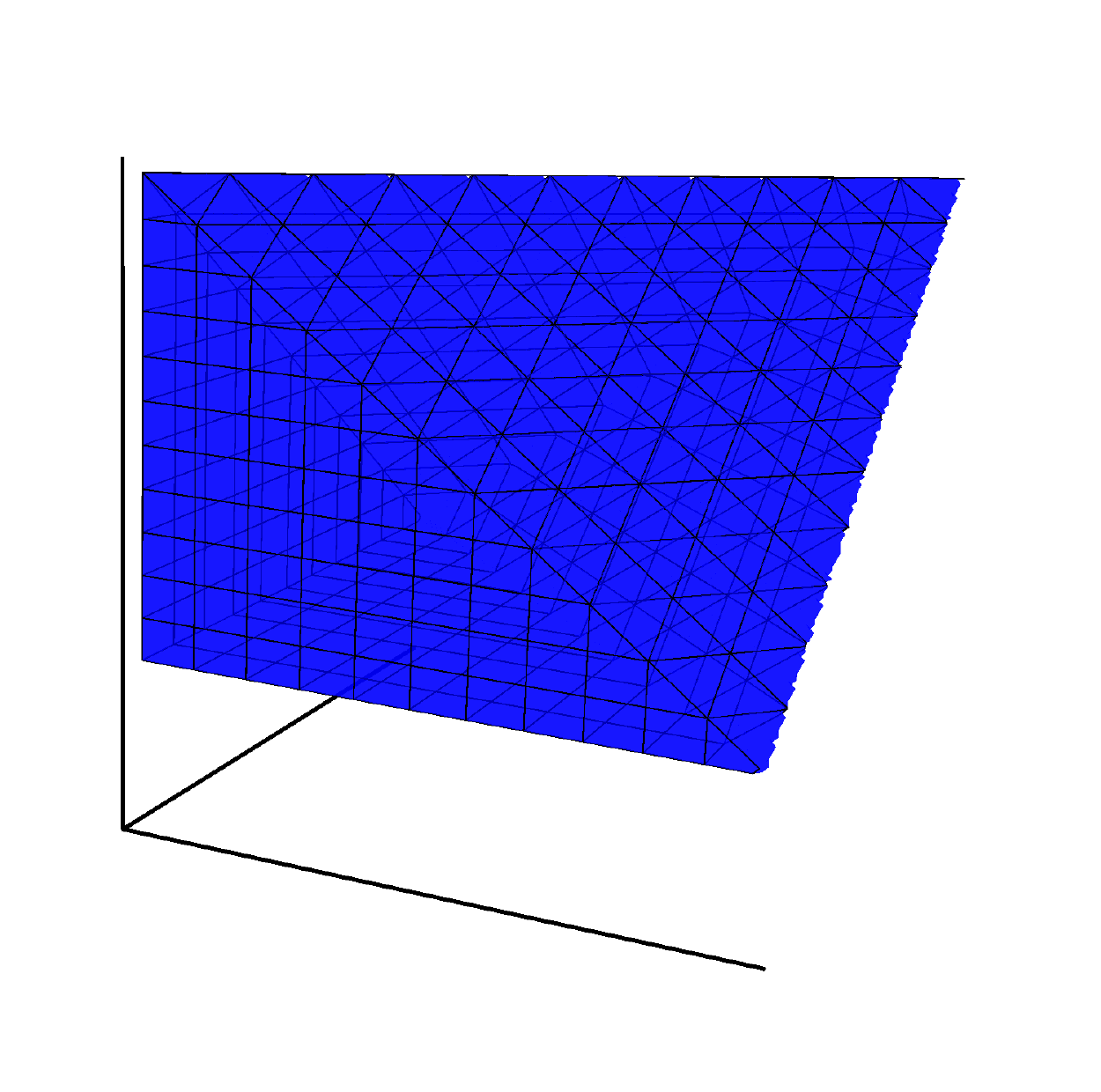}
\subcaption{$W_3$ polytope for $x_{3,3}=0$.}
\end{subfigure}

\begin{subfigure}[h]{0.45\textwidth}
\centering
\includegraphics[width =\textwidth]{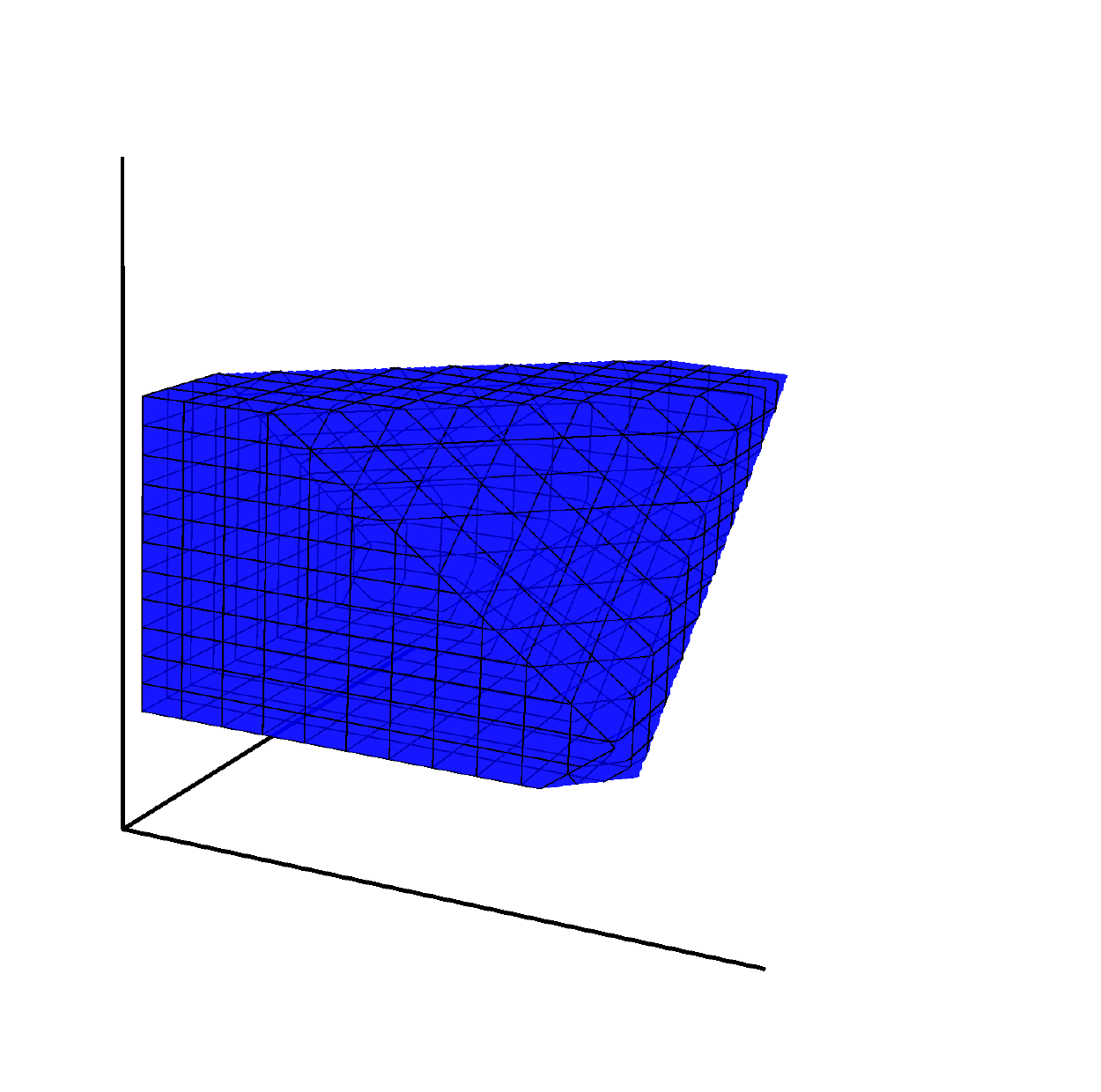}
\subcaption{Full polytope for $x_{3,3}=\frac{2}{9}$.}
\end{subfigure}
\qquad
\begin{subfigure}[h]{0.45\textwidth}
\centering
\includegraphics[width =\textwidth]{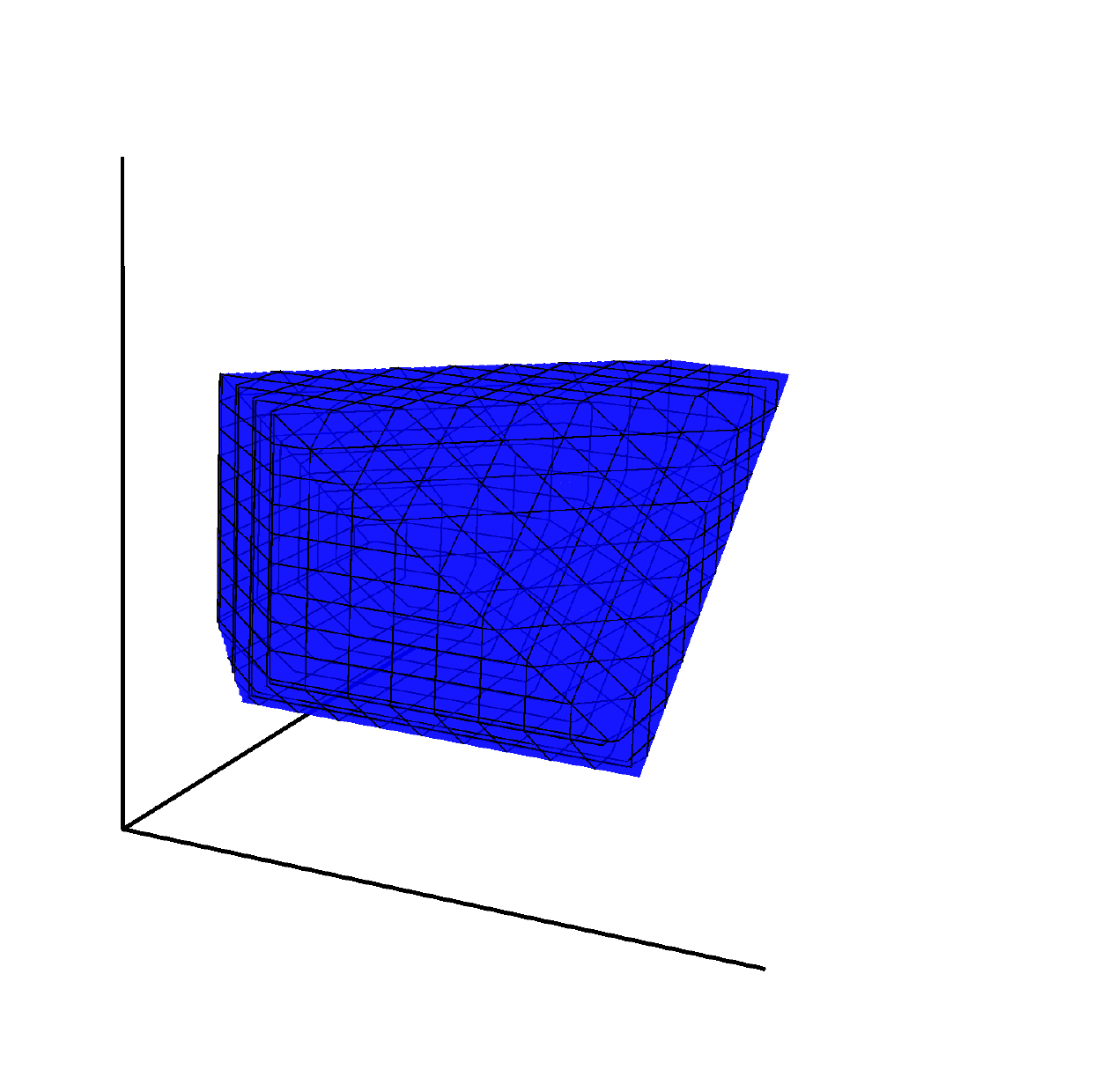}
\subcaption{$W_3$ polytope for$x_{3,3}=\frac{2}{9}$.}
\end{subfigure}

\begin{subfigure}[h]{0.45\textwidth}
\centering
\includegraphics[width =\textwidth]{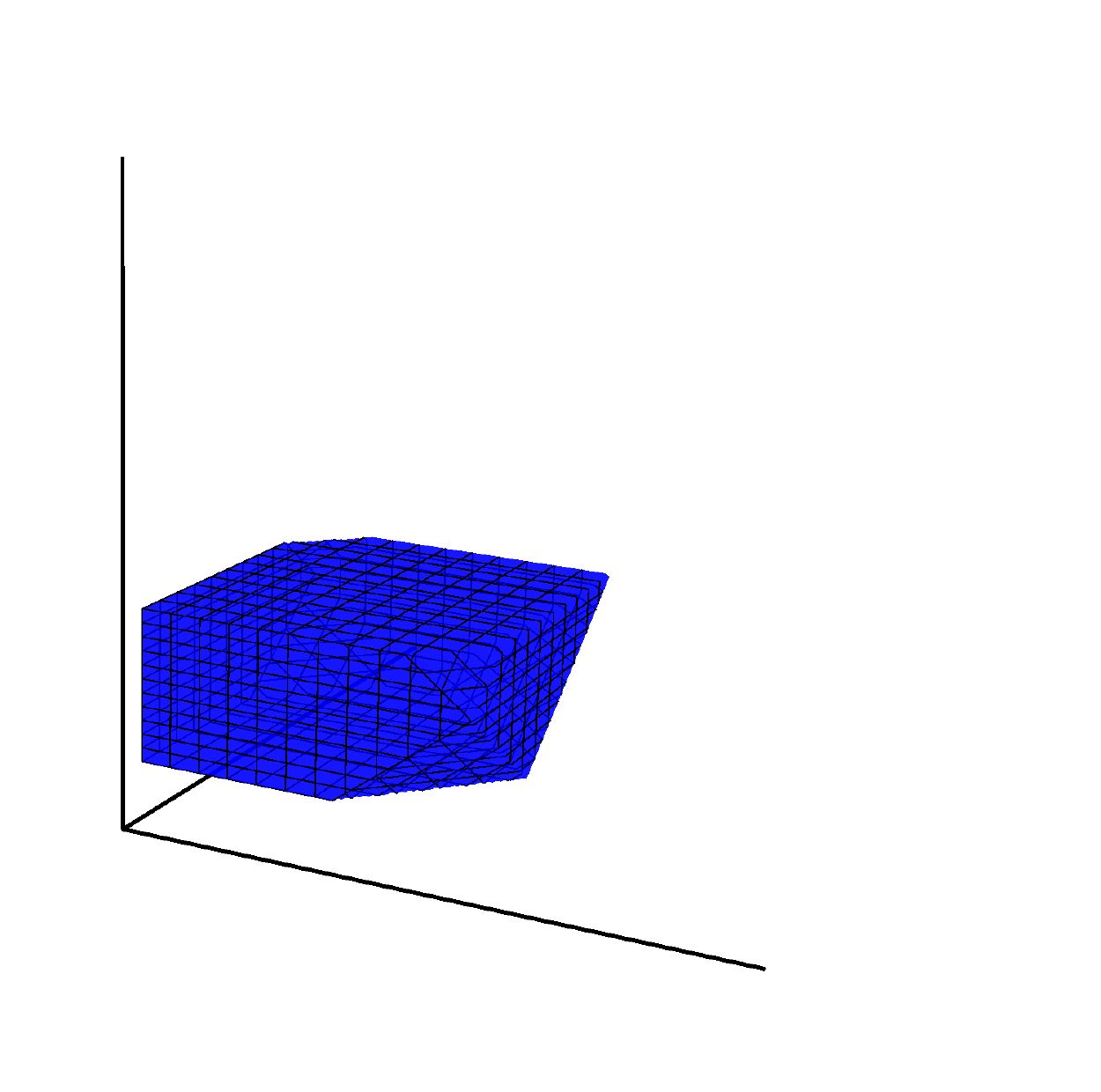}
\subcaption{Full polytope for $x_{3,3}=\frac{4}{9}$.}
\end{subfigure}
\qquad
\begin{subfigure}[h]{0.45\textwidth}
\centering
\includegraphics[width =\textwidth]{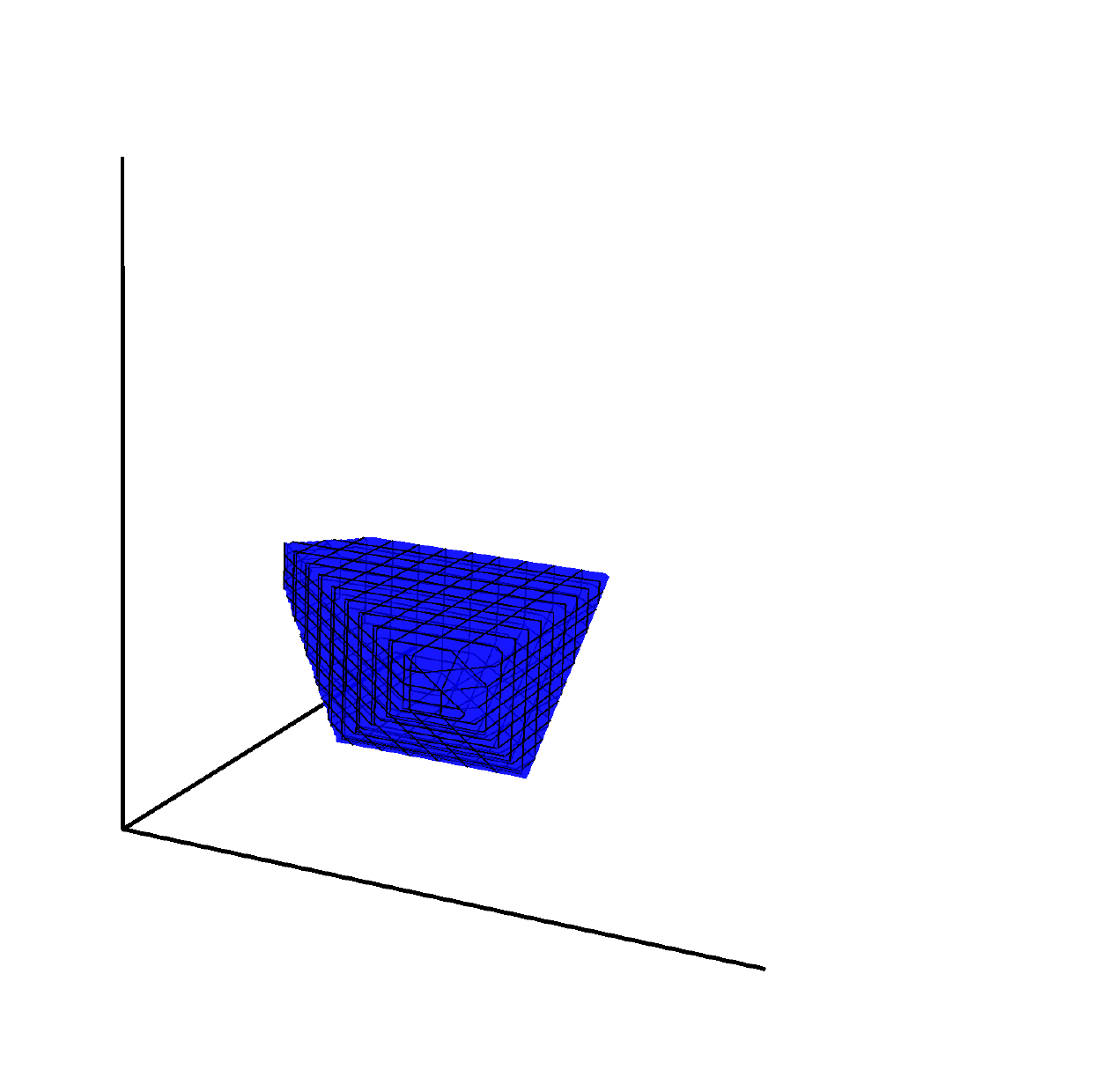}\subcaption{$W_3$ polytope for $x_{3,3}=\frac{4}{9}$.}
\end{subfigure}
\caption[Cuts through $2\times 2 \times 3 $ Entanglement polytopes]{Cuts through the full (left) and $W_3$ polytope for $x_{3,3}=0,1/9,2/9$. Notice that both have the full GHZ polytope for $x_{3,3}=0$.}
\label{fig:223polytopes}
\end{figure}
This gives us both an inner approximation -  the convex hull of $v_2$ through $v_8$ - and an outer approximation - a set of linear inequalities - for $\Delta_{W_3}$. We shall now see that the inner is tight. 

\begin{prop} $\Delta_{W_3}$ is the intersection of the generic polytope with the three inequalities 
\begin{align*}
x_{1,1}+x_{2,1}+x_{3,1}+x_{3,2} &\geq 2 \\
x_{1,1}+x_{3,1} &\geq 1 \\
x_{2,1}+x_{3,1} &\geq 1
\end{align*}
Its vertices are exactly given by the vertices of the generic polytope except $v_1$.
\end{prop} 
See also figure \ref{fig:223polytopes}. The main part in the proof is the following lemma: 
\begin{lem}
Let $\ket{\psi} \in \overline{G\cdot\ket{W_3}}$ with maximal local eigenvalues $\lambda_1,\lambda_2,\lambda_3$. Then 
$\lambda_2 + \lambda_3 \geq 1$ and $ \lambda_1+\lambda_3 \geq 1 $. \end{lem}
\begin{proof}
We use the technique described in section \ref{section:EigValestimates}.
Let $\rho = \ketbra{\psi}$ and denote its reduced density matrices by $(\rho^{(1)},\rho^{(2)},\rho^{(3)})$. 
It is enough to prove the inequality for all states $\ket{\Psi}= g \cdot \ket{W_3}$ for some $g = (g_1,g_2,g_3) \in G$. By QR-decomposition, we can assume that the $g_i$ are upper triangular in the computational basis, so $$\ket{\psi} = g\cdot \sqrt{\frac{2}{5}}\left(\ket{010}+\ket{101}+\frac{1}{\sqrt{2}}\ket{002}\right) = \alpha\ket{000}+\beta\ket{001}+\gamma\ket{010}+\delta\ket{100}+\varepsilon\ket{101}+\eta\ket{002}$$ for some $\alpha,\beta,\gamma,\delta,\varepsilon,\eta$ with $|\alpha|^2+|\beta|^2+|\gamma|^2+|\delta|^2+|\varepsilon|^2+|\eta|^2=1$. Now, by Proposition \ref{prop:eigenvalueestimates}, choosing $\ket{\Phi_2}=\ket{0}_2, \ket{\Phi_3}=\ket{0}_3$, one gets 
$$\lambda_2 + \lambda_3 \geq \braket{\psi|\mathbf{1}_1 \otimes  \ketbra{0} \otimes \mathbf{1}_3 + \mathbf{1}_{12}\otimes\ketbra{0}|\psi} = 2(|\alpha|^2+| \delta |^2)+ |\beta|^2+|\gamma|^2+|\varepsilon|^2 +|\eta|^2 \geq 1 $$
The other inequality follows from symmetry or from the same argument after exchanging $\ket{0} \leftrightarrow \ket{1}$ on the third system.
\end{proof}
This means we have
$$ \mathrm{conv}(\{v_2,\ldots,v_9\})\subseteq \Delta_{W_3} \subseteq \Delta \cap \{x_{1,1}+x_{2,1}+x_{3,1}+x_{3,2} \geq 2 \}\cap \{x_{1,1}+x_{3,1} \geq 1\}\cap \{x_{2,1}+x_{3,1} \geq 1\} $$
One can compute that the extremal points of the right-hand side intersection are exactly equal to $\{v_2,\ldots,v_9\}$. Therefore, the left-hand and right-hand side coincide and the proposition is proven. \\
Of course, these two inequalities do not just fall from the sky. One can for example calculate them from the vertices one hopes are correct. However, for higher-dimensional systems one does not always have such a conjectured set of vertices. This is where our gradient descent method, explained in the next section, comes into play.
\FloatBarrier
\subsection{The entanglement polytopes of a $2\times 2 \times 4$ system}
As proven in \cite{Miyake2004} and shown here in figure \ref{ent_classes_fig} and table \ref{ent_class_table}, there is only one entanglement class whose states have local ranks $(2,2,4)$, namely the generic class with represantative $\ket{000}+\ket{011}+\ket{102}+\ket{113}$. The polytope of this class is given by Bravyi's Inequalities \eqref{BravyiIneq224_1} to \eqref{BravyiIneq224_4} in conjunction with the local inequalities which ensure the density operator and Weyl chamber conditions. Again, let us use ``Most Local Eigenvalues'' coordinates $x=(x_{1,1},x_{2,1},x_{3,1},x_{3,2},x_{3,3})$ corresponding to the element 
$$\left(\begin{pmatrix} x_{1,1} & 0 \\0 & 1-x_{1,1} \end{pmatrix}, \begin{pmatrix} x_{2,1} & 0 \\ 0 & 1-x_{2,1} \end{pmatrix}, \begin{pmatrix} x_{3,1} & 0 & 0 & 0 \\ 0 & x_{3,2} & 0 & 0 \\ 0 & 0 & x_{3,3} & 0 \\ 0 & 0 & 0 & 1-x_{3,1}-x_{3,2}-x_{3,3} \end{pmatrix}\right)$$
of the polytope. We then read off local inequalities $\frac{1}{2} \leq x_{1,1},x_{2,1} \leq 1$ and $0\leq 1-x_{3,1}-x_{3,2}-x_{3,3} \leq x_{3,3} \leq x_{3,2} \leq x_{3,1} \leq 1$ and obtain the following result: 
\begin{prop}
The entanglement polytope of the generic class in the $2\times 2\times 4$-system is given by the inequalities \eqref{BravyiIneq224_1} to \eqref{BravyiIneq224_4} and the set of local inequalities above.
\end{prop}
Using qhull, we can numerically verify the following interesting result: 
\begin{prop}
The set of vertices of this polytope is the union of the vertices of the generic $2\times 2 \times 3$ with the fully mixed vertex $(\frac{1}{2},\frac{1}{2},\frac{1}{4},\frac{1}{4},\frac{1}{4})$. 
\end{prop}
I.e. the generic $2\times 2 \times 4$-polytope is a \emph{cone} over the generic $2\times 2 \times 3$-polytope. 
\section{More analytically obtained results on entanglement polytopes }
Here we present some analytically verifiable results on entanglement polytopes in $2 \times 3 \times n$ systems. Using the Inequalities given by Klyachko in the appendix of \cite{Klyachko2007} for mixed $2 \times 3$ systems we can compute the vertices of the full polytopes using software. A classification of the SLOCC orbits in $2 \times 3 \times n$ systems is given in \cite{Chen2009}, see table \ref{table:23Nentclasses} It turns out that for some orbits the polytope can easily be obtained from the generic one. 

\subsection{$2\times 3 \times 3$ polytopes}\label{section:233polytopes}
There are six orbits with full rank density matrices in the $2 \times 3 \times 3$ system. In \cite{Chen2009}, the following non-normalized representatives are given:  
\begin{align*}
\ket{\psi_0} &=\ket{000}+\ket{111}+\ket{022}+\ket{122} \\
\ket{\psi_1} &=\ket{100}+\ket{010}+\ket{001}+\ket{122}\\
\ket{\psi_2} &=\ket{100}+\ket{010}+\ket{001}+\ket{022}\\
\ket{\psi_3} &=\ket{100}+\ket{010}+\ket{022}\\
\ket{\psi_4} &=\ket{100}+\ket{010}+\ket{001}+\ket{121}+\ket{112}\\
\ket{\psi_5} &=\ket{001}+\ket{010}+\ket{121}+\ket{112}\\
\end{align*}

\begin{prop}[Hierarchy of SLOCC Orbits in $2\times 3 \times 3$\footnote{Both the statement and the proof are unpublished work by P\'eter Vrana}]\label{prop:233hierarchy}
\begin{align}
\overline{G\cdot\ket{\psi_0}} &\supset \overline{G\cdot\ket{\psi_1}} \supset \overline{G\cdot\ket{\psi_2}} \supset \overline{G\cdot\ket{\psi_3}} \\
\overline{G\cdot\ket{\psi_0}} &\supset \overline{G\cdot\ket{\psi_4}} \supset \overline{G\cdot\ket{\psi_5}}
\end{align}
\end{prop}
\begin{proof}
We will give a full proof of the first inclusion and only sketch the proof of the others. Notice that the last inclusion in every line follows directly from \ref{thm:magiclemma}. 
\begin{clm} $$\lim_{\varepsilon \to 0}\frac{1}{\varepsilon}\left( \left(
\begin{array}{cc}
 1 & -1 \\
 0 & \varepsilon  \\
\end{array}
\right), \left(
\begin{array}{ccc}
 1 & -1 & 0 \\
 0 & \varepsilon  & 0 \\
 0 & 0 & 1 \\
\end{array}
\right),\left(
\begin{array}{ccc}
 1 & -1 & 0 \\
 0 & \varepsilon  & 0 \\
 0 & 0 & 1 \\
\end{array}
\right)\right)\cdot \ket{\psi_0} = \ket{\psi_1}$$
\end{clm}
\begin{proof}[Proof of claim]

Let $g(\varepsilon)$ denote the above SLOCC operator. We calculate:
\begin{align*}g(\varepsilon)\cdot \ket{\psi_0} &= \frac{\frac{1}{\varepsilon} \left( \ket{000} + (-\ket{0}+\varepsilon \ket{1}) \otimes (-\ket{0}+\varepsilon \ket{1}) \otimes  (-\ket{0}+\varepsilon \ket{1}) + \ket{022} +  (-\ket{0}+\varepsilon \ket{1})\otimes \ket{22}\right)}{\|(\text{the above})\|} \\ 
&=   \frac{\frac{1}{\varepsilon} \left( \ket{000} -\ket{000}+\varepsilon(\ket{100}+\ket{010}+\ket{001}) - \varepsilon^2 (\ket{011}+\ket{101}+\ket{110})\right)}{\|\cdots\|} \\
&+\frac{\frac{1}{\varepsilon}\left(\varepsilon^3 \ket{111} + \ket{022}  -\ket{022}+\varepsilon \ket{122}\right)}{\|\cdots\|} \\
&= \frac{\ket{100}+\ket{010}+\ket{001}+\ket{122}-\varepsilon(\ket{110}+\ket{101}+\ket{011})+\varepsilon^2\ket{111}}{\sqrt{4+3\varepsilon^2+\varepsilon^4}}\\ &\overset{\varepsilon \to 0}{\to} \ket{100}+\ket{010}+\ket{001}+\ket{122} = \ket{\psi_1} 
\end{align*}
\end{proof}
The other calculations are similar, so we just state their results. Let
\begin{align*}
g_1(\varepsilon)&=\left(\left(
\begin{array}{cc}
 1 & 0 \\
 0 & \varepsilon  \\
\end{array}
\right),\left(
\begin{array}{ccc}
 1 & 0 & 0 \\
 0 & 0 & \frac{1}{\varepsilon } \\
 0 & 1 & 0 \\
\end{array}
\right),\left(
\begin{array}{ccc}
 0 & 1 & 0 \\
 0 & 0 & 1 \\
 1 & 0 & 0 \\
\end{array}
\right)\right) \\
g_2(\varepsilon)&=\left(\frac{1}{2} \left(
\begin{array}{cc}
 \varepsilon  & 0 \\
 0 & 1 \\
\end{array}
\right)\cdot\left(
\begin{array}{cc}
 1 & -1 \\
 1 & 1 \\
\end{array}
\right),\left(
\begin{array}{ccc}
 1 & 0 & 0 \\
 0 & \frac{1}{\varepsilon } & \frac{1}{\varepsilon } \\
 0 & 0 & -\varepsilon  \\
\end{array}
\right)\cdot\left(
\begin{array}{ccc}
 1 & 1 & 0 \\
 1 & -1 & 0 \\
 0 & 0 & 1 \\
\end{array}
\right),\left(
\begin{array}{ccc}
 1 & 0 & 0 \\
 0 & \frac{1}{\varepsilon } & -\frac{1}{\varepsilon } \\
 0 & 0 & \varepsilon  \\
\end{array}
\right)\cdot\left(
\begin{array}{ccc}
 1 & 1 & 0 \\
 1 & -1 & 0 \\
 0 & 0 & 1 \\
\end{array}
\right)\right)
\end{align*}
then \begin{align*}
\lim_{\varepsilon \to 0} g_1(\varepsilon)\cdot \ket{\psi_1} &= \ket{\psi_2} \\
\lim_{\varepsilon \to 0} g_2(\varepsilon)\cdot \ket{\psi_4} &= \ket{\psi_5}
\end{align*}
and the last inclusions follow from \ref{thm:magiclemma}.
\end{proof}
We conclude that (since there are only finitely many orbits) the polytope of $\ket{\psi_0}$ is the generic one. It is therefore given by the Quantum Marginal Inequalities. We state its vertices here for convenience.
\begin{lem}
The entanglement polytope of $\ket{\psi_0}$ is the full $2\times 3 \times 3$ polytope. It has 18 vertices, out of which 13 are vertices of the $2 \times 2 \times 3$ and the   $2 \times 3 \times 2$ polytope\footnote{Notice there are actually several embeddings of these lower-dimensional systems in the $2 \times 3 \times 3$ system, but only the image of one of them will actually lie in our Weyl chamber.} (i.e. the ones given in Proposition \ref{prop:vertices223}, possibly with system 2 and 3 exchanged). The five new vertices are 
\begin{align*}
v_1 &= (\frac{1}{2},\frac{1}{3},\frac{1}{3},\frac{1}{3},\frac{1}{3})\\
v_2 &= (1,\frac{1}{3},\frac{1}{3},\frac{1}{3},\frac{1}{3})\\
v_3 &= (\frac{1}{2},\frac{1}{3},\frac{1}{3},\frac{2}{3},\frac{1}{6})\\
v_4 &= (\frac{1}{2},\frac{2}{3},\frac{1}{6},\frac{1}{3},\frac{1}{3})\\
v_5 &= (\frac{1}{2},\frac{2}{3},\frac{1}{6},\frac{2}{3},\frac{1}{6})\\
\end{align*}
\end{lem}For the other orbits only the "top" polytopes are readily computed. 
\begin{prop}\label{prop:233polytopes}
The entanglement polytope $\Delta_{\ket{\psi_4}}$ is given by the intersection of the generic one with the single halfspace given by the inequality
$$x_{1,1}+2x_{2,1}+x_{2,2}+2x_{3,1}+2x_{3,2} \geq 2.$$
The entanglement polytope $\Delta_{\ket{\psi_1}}$ is given by the intersection of the generic polytope with the halfspaces 
\begin{align}
x_{1,1}+x_{2,1}+x_{2,2}+x_{3,1}+x_{3,2} &\geq 3 \\
x_{1,1}+2x_{2,1}+x_{2,2}+x_{3,1} &\geq 2 \\
x_{1,1}+x_{2,1}+2x_{3,1}+x_{3,2} &\geq 2 \\
\end{align}
\end{prop}
\begin{proof}
We start with the entanglement polytope of $\ket{\psi_4}$. Applying corollary \ref{cor:magiclemma} to the state$\ket{\psi} =\sqrt{a}\ket{100}+\sqrt{b}\ket{010}+ \sqrt{c}\ket{001}+\sqrt{d}\ket{121}+\sqrt{e}\ket{112}$, we have to solve the linear equations 
$$A(\ket{\psi}) = \begin{pmatrix} 3 & 1 & 1 & 1 & 1 \\ 1 & 3 & 1 & 0 & 1 \\ 1 & 1 & 3 & 1 & 0 \\ 1 & 0 & 1 & 3 & 1 \\ 1 & 1 & 0 & 1 & 3 \end{pmatrix}\begin{pmatrix} a \\ b \\ c \\ d \\ e \end{pmatrix} = \lambda \begin{pmatrix} 1 \\ 1 \\ 1 \\ 1 \\ 1 \end{pmatrix}$$
which can be seen to lead to $a=\frac{1}{9},b=c=d=e=\frac{2}{9}$. The corresponding state then has local eigenvalues 
$$\lambda^*=\left(\left(\frac{5}{9},\frac{4}{9}\right),\left(\frac{4}{9},\frac{1}{3},\frac{2}{9}\right),\left(\frac{4}{9},\frac{1}{3},\frac{2}{9}\right)\right)$$
Now we can compute the corresponding inequality $\langle x - O, \lambda^*-O\rangle \geq \langle\lambda^*-O, \lambda^*-O\rangle$. Since $O = ((\frac{1}{2},\frac{1}{2}),(\frac{1}{3},\frac{1}{3},\frac{1}{3}),(\frac{1}{3},\frac{1}{3},\frac{1}{3}))$, $\lambda^* - O = \left(\left(\frac{1}{18},\frac{-1}{18}\right),\left(\frac{1}{9},0,\frac{-1}{9}\right),\left(\frac{1}{9},0,\frac{-1}{9}\right)\right)$ and $\langle\lambda^*-O,\lambda^*-O\rangle = \frac{1}{18}$. The left hand side of the inequality is 
\begin{align*}
&\left\langle  \left(\left(\frac{1}{18},\frac{-1}{18}\right),\left(\frac{1}{9},0,\frac{-1}{9}\right),\left(\frac{1}{9},0,\frac{-1}{9}\right)\right),\right. \\ & \left.\left(\left(x_{1,1}-\frac{1}{2},x_{1,2}-\frac{1}{2}\right),\left(x_{2,1}-\frac{1}{3},x_{2,2}-\frac{1}{3},x_{2,3}-\frac{1}{3}\right),\left(x_{2,1}-\frac{1}{3},x_{2,2}-\frac{1}{3},x_{2,3}-\frac{1}{3} \right)\right)\right\rangle \\ 
&= \frac{x_{1,1}}{18} -\frac{x_{1,2}}{18} + \frac{x_{2,1}}{9} -\frac{x_{2,3}}{9}+\frac{x_{3,1}}{9} -\frac{x_{3,2}}{9} \geq \frac{1}{18}
\end{align*}
which yields the claimed inequality after converting to the greater local eigenvalues, i.e. replacing $x_{i,d_i} = 1 - \sum_{j=1}^{d_i-1}x_{i,j}$. 
This procedure can of course be automatised and has been implemented\footnote{See section \ref{section:closestpointfinder} in the appendix.}.
Thus we get an outer approximation of the polytope. One can calculate now that the the vertices of the generic polytope intersected with this halfspace are exactly the vertices of the generic polytope except the origin $v_1 = (\frac{1}{2},\frac{1}{3},\frac{1}{3},\frac{1}{3},\frac{1}{3})$. This means to show that our claim is correct we have to verify all other vertices are included in $\Delta{\ket{\psi_4}}$. This can be done in two steps: 
\begin{enumerate}[i)]
\item Show the orbit closure of $\ket{\psi_4}$ contains the maximal $2 \times 2 \times 3$ and $2 \times 3 \times 2$ orbits. Then also their entanglement polytopes, along with all their vertices, will be contained in $\Delta_{\psi_4}$
\item Show that the 4 vertices not accounted for in this way also lie in the polytope using Proposition \ref{prop:free2}.
\end{enumerate}
Point i) can be checked again by virtue of Corollary \ref{cor:magiclemma}: Exchange first $\ket{0} \leftrightarrow \ket{1}$ on the third system, leading to a state $\ket{011}+\ket{101}+\ket{000}+\ket{120}+\ket{112}$. By Corollary \ref{cor:magiclemma},
$\ket{011}+\ket{101}+\ket{000}+\ket{112}$ is in  the orbit closure, this however is a representative of the generic orbit in $2\times 2 \times 3$. In exactly the same way one checks that the orbit closure contains a representative of the generic orbit in $2 \times 3 \times 2$. 
For Point ii) we can use Proposition \ref{prop:free2}. Any point in the convex hull of $\{\Phi(\ketbra{010}),\Phi(\ketbra{001}),\Phi(\ketbra{100}),\Phi(\ketbra{121}),\Phi(\ketbra{112})\}$ witb coefficients $a,b,c,d$ and $e$ (satisfying $a+b+c+d+e = 1$) respectively is given by diagonal matrices with entries $$\lambda_1 = (a+b,c+d+e),\lambda_2 = (b+c,a+e,d), \lambda_3 = (a+c,b+d,e)$$ 
We now show that we can get the remaining vertices $v_2,\ldots v_5$ from this (up to ordering, which can of course be done using local unitaries). E.g. to get $v_2 = (1,1/3,1/3,1/3,1/3)$ we need $d = 1/3, e = 1/3$ and then immediately get $c = 1/3, a = 0, b= 0$. Now let $d=1/3, e = 1/6$. It then follows that for $c=0, b=1/6, a =1/3$ we get. $v_3 =  (1/2,1/3,1/3,2/3,1/6)$ We therefore conclude that the entanglement polytope of $\ket{\psi_4}$ is indeed given by the claimed inequality. 
Similarly we get $v_4$ for $a=e=1/3$ and $b=d= 1/6$ and $v_5$ for $a = 1/2, b = 0$ and $c=d=e=1/6$. \\
Now for the second state. The first inequality $x_{1,1}+x_{2,1}+x_{2,2}+x_{3,1}+x_{3,2} \geq 3$ can be proven using the closest point method in exactly the same way as above. The equation in this case is 
$$ \begin{pmatrix} 3 & 1 & 1 & 1 \\ 1 & 3 & 1 & 0 \\ 1 & 1 & 3 & 0 \\ 1 & 0  & 0 &  3 \end{pmatrix}\begin{pmatrix} a \\ b \\ c \\ d \end{pmatrix} = \lambda\begin{pmatrix} 1 \\ 1 \\ 1 \\ 1 \end{pmatrix} \ ,$$  the calculation is not particularly thrilling, and we skip it. The other two, however, must be proven algebraically (at least one, the other then follows by symmetry). By using $x_{2,1}+x_{2,2} = 1 -x_{2,3}$, we see that the first inequality $x_{1,1}+2x_{2,1}+x_{2,2}+x_{3,1} \geq 2$ is equivalent to $x_{1,1}+x_{2,1}- x_{2,3}+x_{3,1} \geq 1$. Now we use an argument similar to that in section \ref{section:EigValestimates} (where we did not consider inequalities with differences). Let $\ket{\psi} \in G\cdot\ket{\psi_1}$ and $\rho = \ketbra{\psi}$ and $\sigma(A)$  denote the spectrum of any oprator $A$.   By the min-max-Principle, we know that 
\begin{equation} x_{2,3}= \min_{\lambda \in \sigma(\rho^{(2)})} \lambda = \min_{\substack{v \in \C^3 \\ \|v\|=1}}\braket{v|\rho^{(2)}|v} =\min_{\substack{v \in \C^3 \\ \|v\|=1}} \braket{\psi_1|(\mathds{1}_1\otimes \ketbra{v} \otimes \mathds{1}_3)|v}.\end{equation}
Therefore, combining this with \ref{prop:eigenvalueestimates}, we see that for any choice of normalised vectors $v_1 \in \C^2, v_2,v_3,v_4 \in \C^3$ and any state $\ket{\psi}$ with ordered local eigenvalues $\lambda_{i,j}, i = 1,\ldots,3, j=1,\ldots d_i$, the following estimate holds:
\begin{equation}\label{proof:1233 estimate}\lambda_{1,1}+\lambda_{2,1}-\lambda_{2,2}+\lambda_{3,1} \geq \braket{\psi|\left(\ketbra{v_1}\otimes\mathds{1}_{23}+\mathds{1}\otimes(\ketbra{v_2}-\ketbra{v_3})\otimes\mathds{1}_3+\mathds{1}_{12}\otimes\ketbra{v_4}\right)|\psi}
\end{equation}
Now we have to choose a smart basis in which to expand $\psi$. For this we first permute $\ket{0} \rightarrow \ket{1}\rightarrow\ket{2}\rightarrow\ket{0}$ on the second system and obtain a new representative of the orbit of $\ket{\psi_1}$: 
$$\ket{\psi_1} =\ket{100}+\ket{010}+\ket{001}+\ket{122}\mapsto \ket{\psi'} = \ket{110}+\ket{020}+\ket{011}+\ket{102}$$
Any state in $G\cdot\ket{\psi_1}$ can thus be written as $g\cdot\ket{\psi'}$ for some $g \in G$. By changing our basis further, if needed, we can assume $g$ is upper triangular. Then $\ket{\phi}:= g\cdot \ket{\psi'}=\sum_{\mathbf{j}\in\supp(\ket{\phi})}c_{\mathbf{j}}\ket{\mathbf{j}}$ will have support contained in
$$\supp(\ket{\phi}) \subseteq \{\ket{000},\ket{100},\ket{010},\ket{001},\ket{011},\ket{110},\ket{101},\ket{002},\ket{102},\ket{020}\}$$
Now we use this in the estimate \eqref{proof:1233 estimate} and plug in $v_1=\ket{0},v_2=\ket{0},v_3=\ket{2},v_4 = \ket{0}$. The estimate then becomes 
\begin{align*}
\lambda_{1,1}+\lambda_{2,1}-\lambda_{2,2}+\lambda_{3,1} &\geq \braket{\phi|\left(\ketbra{0}\otimes\mathds{1}_{23}+\mathds{1}\otimes(\ketbra{0}-\ketbra{2})\otimes\mathds{1}_3+\mathds{1}_{12}\otimes\ketbra{0}\right)|\phi} \\
&= \sum_{\mathbf{j}\in\supp(\ket{\phi})}|c_{\mathbf{j}}|^2 \braket{\mathbf{j}|\left(\ketbra{0}\otimes\mathds{1}_{23}+\mathds{1}\otimes(\ketbra{0}-\ketbra{2})\otimes\mathds{1}_3+\mathds{1}_{12}\otimes\ketbra{0}\right)|\mathbf{j}} \\
&= 3\cdot|c_{000}|^2 \\
&+2\cdot\left(|c_{010}|^2+|c_{001}|^2+|c_{100}|^2+|c_{002}|^2+|c_{020}|^2\right) \\
&+1\cdot\left(|c_{011}|^2+|c_{101}|^2+|c_{110}|^2+|c_{102}|^2\right)\\
&-1\cdot |c_{020}|^2 \\
&\geq \sum_{\mathbf{j}\in\supp(\ket{\phi})}|c_{\mathbf{j}}|^2 = 1.
\end{align*}
The other inequality now follows from symmetry. \\
Now that we have proven the inequalities, we once more have a go at the vertices. Calculating the extremal points of the intersection of the inequalities, they include all vertices of the generic polytope, except the origin, and six additional vertices
\begin{align*}
w_1 &= \left(\frac{1}{2},\frac{1}{3},\frac{1}{3},\frac{1}{2},\frac{1}{3}\right)\\
w_2 &= \left(\frac{1}{2},\frac{3}{8},\frac{3}{8},\frac{3}{8},\frac{3}{8}\right) \\
w_3 &= \left(\frac{1}{2},\frac{2}{5},\frac{3}{10},\frac{2}{5},\frac{2}{5}\right)\\
w_4 &= \left(\frac{1}{2},\frac{2}{5},\frac{2}{5},\frac{2}{5},\frac{2}{5}\right)\\
w_5 &= \left(\frac{1}{2},\frac{1}{2},\frac{1}{3},\frac{1}{3},\frac{1}{3}\right)\\
w_6 &= \left(\frac{2}{3},\frac{1}{3},\frac{1}{3},\frac{1}{3},\frac{1}{3}\right)\\
\end{align*} 
The convex hull of $\Phi(\ket{011}),\Phi(\ket{020}),\Phi(\ket{110}),\Phi(\ket{102})$ is given by diagonal matrices with entries 
$$\lambda(a,b,c,d)=(\lambda_1,\lambda_2,\lambda_3)=((a+b,c+d),(d,a+c,b), (b+c,a,d))$$
with $a+b+c+d =1$. One then finds that (up to ordering) 
\begin{align*}
w_1 &= \lambda(1/6,1/3,1/6,1/3)\\
w_2 &= \lambda(1/4,1/4,1/8,3/8)\\
w_3 &= \lambda(1/5,3/10,1/10,2/5)\\
w_4 &= \lambda(3/10,1/5,1/10,2/5)\\
w_5 &= \lambda(1/3,1/6,1/6,1/3)\\
w_6 &= \lambda(1/3,1/3,0,1/3)\\
\end{align*}
i.e the additional six vertices are included in $\Delta_{\psi_1}$. One now still needs to check that $\ket{\psi_1}$ can degenerate into maximal $2\times 2 \times 3$ and $2 \times 3 \times 2$ orbits. To this end, consider once more the state $\ket{\psi_1}=\ket{100}+\ket{010}+\ket{001}+\ket{122}$. Exchanging $\ket{0}$ and $\ket{1}$ on the third system we arrive at 
$$\ket{\psi''} = \ket{000}+\ket{011}+\ket{101}+\ket{122}$$ which can be converted to a representative of the maximal $2\times 2\times 3$ class by applying a non-invertible transformation 
$$\begin{pmatrix} 1 & 0 & 0 \\ 0 & 1 & 1\\ 0 & 0 & 0 \end{pmatrix}$$
on the second system. This can be achieved by acting on $\ket{\psi''}$ with the SLOCC operator 
$$ \mathds{1}\otimes \begin{pmatrix} \frac{1}{\varepsilon} & 0 & 0 \\ 0 & \frac{1}{\varepsilon} & \frac{1}{\varepsilon}\\ 0 & 0 & \varepsilon^2 \end{pmatrix}\otimes\mathds{1}$$ and letting $\varepsilon \to 0$.
In the same way we can reach the maximal $2\times 3 \times 2$ orbit. We conclude that the polytope contains all the vertices of the intersection of the generic polytope with the halfspaces given by the inequalities \ref{prop:233polytopes}, and therefore the proposition is proven. 
\end{proof}

\subsubsection{A word about the hierarchy of the polytopes}
In chapter \ref{ch:entpoly} we have seen that 
$$ \mathcal{C} \subseteq \mathcal{D} \Rightarrow \Delta_{\mathcal{C}} \subseteq \Delta_{\mathcal{D}},$$
i.e. an inclusion of the orbit closures gives an inclusion of the polytopes. In the $2 \times 3 \times 3$ case we can observe the following. 
\begin{rem}For the states named above,
$$\Delta_{\ket{\psi_4}} \subseteq \Delta_{\ket{\psi_1}}.$$
\end{rem}
\begin{proof}
One simply verifies that every vertex of $\Delta_{\ket{\psi_4}}$ satisfies the inequalities defining $\Delta_{\ket{\psi_1}}$.
\end{proof}
However, it is unclear if this inclusion is dictated by the SLOCC hierarchy or not, i.e at the moment we do not know whether it is also true that $\overline{G\cdot\ket{\psi_4}}$ is included in $G\overline{\cdot\ket{\psi_1}}$. 
\subsection{More closest-point inequalities for $2\times 3 \times N$ systems and $2 \times 4 \times N $ systems} 
The automatised closest-point method can be used for any orbit with a  free representative on not too many product states. For $2 \times 3 \times N$ systems this actually applies to all orbits, see table \ref{table:23Nentclasses}. For each entanglement class we can therefore calculate an inequality in an exact way, save those, of course, that do contain the origin. For $2 \times 3 \times N$ this happens only if the class has the full polytope.  A table of all those inequalities can be found in the appendix in Table \ref{table:23Ncpineqs}. \\
The method also can be applied to some classes in $2 \times 4 \times N$ systems. Its use there is mainly as a criterion a class has the full polytope or not. It is however limitated: There exist classes without free representatives, and classes whose entanglement polytope contains the origin but not the whole polytope (see the next section). Again, the inequalities can be found in the appendix. 
\subsection{A result for N Qubits}
A well-studied example in a system of $N$-Qubits are the so called symmetric Dicke states [REF,REF]
\begin{equation}
\ket{M,N}:=\sum_{\sigma \in P_{M,N}}\sigma\cdot \ket{\underbrace{1\cdots 1}_{M}\underbrace{0\cdots 0}_{N-M}}
\end{equation} 
where $M < N/2$ and $P_{M,N}$ is the set of all nontrivial permutations of the spins, i.e. the sum goes over all basis vectors with exactly $M$ ones. 
\begin{prop}\footnote{A very similar result was arrived at in \cite{Walter2012} in the bosonic case.}
The entanglement polytope of $\ket{M,N}$ is the full $N$-Qubit polytope intersected with the halfspace 
$\sum_{i=1}^N \lambda_i^{max} \geq N-M$.
\end{prop}

\begin{proof} 
To prove the inequality, we use section \ref{section:EigValestimates}. Applying a upper triangular matrix to $\ket{M,N}$ will result in a state $\ket{\psi}=\sum_{\mathbf{j}}c_{\mathbf{j}}\ket{\mathbf{j}}$ supported on product basis vectors where at most $M$ entries are 1. Therefore, according to Proposition \ref{prop:eigenvalueestimates},
$$\sum_{i=1}^N \lambda_i^{max} \geq  \braket{\psi|\left(\sum_{i=1}^N\mathds{1}\otimes(\ket{0}\bra{0}_i)\otimes\mathds{1}\right)|\psi} \geq N-M\sum_{\mathbf{j}}|c_{\mathbf{j}}|^2 = N-M$$
To show that this is the only additional inequality, we once more make use of the vertices. 
The marginal inequalities for qubits are well-known \cite{Higuchi2003}. The vertices (in Most Local Eigenvalues coordinates) are points with entries that are any combination of 1's and $\frac{1}{2}$'s, except those with a single $\frac{1}{2}$.

One can check that the vertices of the intersection of the full N-qubit polytope are just the ones from above which satisfy the additional inequality. Any such vertex has the property that is has at most $2M$ entries which are $1/2$. Now, we can make use of Proposition \ref{prop:free2} once more. Let $$v = \left(\underbrace{\frac{1}{2},\ldots,\frac{1}{2}}_{2M},\underbrace{1\ldots 1}_{N-2M}\right)$$ be such a vertex. The $v$ is exactly equal to $\frac{1}{2}(\Phi(\ketbra{\mathbf{j}_1})+\Phi(\ketbra{\mathbf{j}_2}))$ where (remember $N>2M$) 
$$ \mathbf{j}_1 = \ket{\underbrace{1\cdots 1}_{M}\underbrace{0\cdots 0}_{N-M}}, \mathbf{j}_2 = \ket{\underbrace{0\cdots 0}_{M}\underbrace{1\cdots 1}_{M}\underbrace{0\cdots 0}_{N-2M}}$$ It is clear that we can get all vertices in this way by just choosing  $\mathbf{j}_1$ and  $\mathbf{j}_2$  appropriately. 
\end{proof}
\chapter{Using the gradient flow}\label{chapter:GradientFlow}
In \cite{Walter2012} a gradient flow method to find the closest point to the origin in the entanglement polytope of a certain class is described. Alternatively, one can understand this as a method to (asymptotically) distill ``as much entanglement as possible'' (measured via the linear entropy of entanglement, see \cite{Zurek1993}) from a given state. 

\begin{figure}[h]
\centering
\includegraphics[width = 0.5\textwidth]{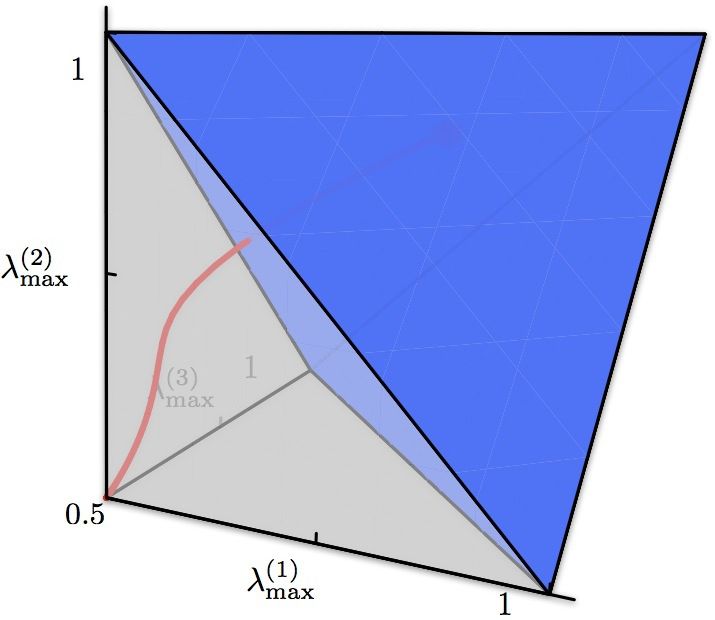}
\caption[Visualistaion of gradient flow to origin in three qubit system]{A visualisation of the gradient flow to the origin method for three qubits, taken from \cite{Walter2012}.}\label{fig:3qubitgradflow}
\end{figure}
In this section we present a generalisation of this by which we can flow not only in the direction of maximal entanglement, but in direction of any rational point in the positive Weyl chamber, thus finding the closest point in the polytope of a given entanglement class to such a point. 
We will first look at the theoretical background and then describe how it can be implemented and used to compute entanglement polytopes.
\section{Theoretical background}
The basic idea for the extended flow is an extension of the action first described by Mumford in \cite[Appendix]{Ness1984}. 
\subsection{An extended G-action}
We first need some preparations.
\subsubsection{A G-Action on coadjoint orbits of integral points}
Let $\lambda$ be an integral point in $\mathfrak{t}^+$. This means $\lambda$ consists of a tuple of diagonal entries  $\lambda^{(i)} \in  \Z^{d_i}$ with weakly decreasing entries for $i=1,\ldots,N$. We identify $\lambda^{(i)}$ with the vector of its entries on the diagonal. Every such vector determines an irreducible representation $V^{d_i}_{\lambda^{(i)}}$ of $SL(d_i)$ as described in \cite{Fulton1991}. By $V_{\lambda}$ we denote the  $G$-Module $\bigotimes_{i=1}^N V^{d_i}_{\lambda^{(i)}}$, it is irreducible with highest weight $\lambda$ and we denote by $\ket{\lambda}$ its highest weight vector. Now, let $O_\lambda = K \cdot \lambda $ be the coadjoint orbit of $\lambda$ in $Lie(K)^*$. Then $O_{\lambda}$ carries a natural $G$-action. To see this, define a $K$-equivariant embedding of $O_{\lambda}$ in $V_{\lambda}$ by sending $U\lambda U^{\dagger} \mapsto U\cdot \ket{\lambda}$. Now we can make use of the fact that $K\cdot\ket{\lambda}= K/T = G/B = G\cdot\ket{\lambda} $ (see e.g \cite{Fulton1991}), i.e. the $G$-orbit and the $K$-orbit of $\ket{\lambda}$ in $V_{\lambda}$ coincide. We can therefore define $g \cdot \lambda := U(g)\cdot = \lambda U(g)\lambda U(g)^{\dagger}$ for some $U(g)\in K$ with $g\cdot\ket{\lambda} = U(g)\cdot\ket{\lambda}$. 
\begin{clm} The above action is well-defined. Moreover, if we take as positive roots the computational basis vectors $\ket{\mathbf{j}}$, such that the Borel subgroup is given by upper triangular matrices and $g = QR$ is the QR-decomposition of $g$, $U(g) = Q$.
\end{clm} 
\begin{proof}
Surely, the well-definedness of this action does not depend on the choice of positive roots, so it suffices to show the second point. Since the Borel subgroup stabilises the highest weight vector, we have $g\cdot\ket{\lambda} = QR\cdot\ket{\lambda} = Q\cdot\ket{\lambda}$ \ . \end{proof}

\subsubsection{The extended action}
Let $X=\overline{G\cdot\rho}$ be an orbit closure in $P(H)$ and assume $\frac{\lambda}{k}$ is a rational point in $\mathfrak{t}^+$.
Let $\lambda^*$ be the highest weight of the dual module $V_{\lambda}^* = V_{\lambda^*}.$ We then define an embedding
\begin{equation}\label{embedding1}X \times O_{\lambda^*} \hookrightarrow P(\mathrm{Sym}^kH)\times P(V_{\lambda^*}) \hookrightarrow P(\mathrm{Sym}^kH \otimes V_{\lambda^*}) \end{equation}
where the first embedding is given by $$ (\ket{\psi},g \cdot \lambda^*) \mapsto (\ketbra{\psi^k} ,(g\ket{\lambda*})(\bra{\lambda^*})g^{\dagger})$$i.e. the product of the Veronese embedding and the embedding discussed above, and the second is the Segre embedding $P(V_1)\times P(V_2) \hookrightarrow P(V_1 \otimes V_2)$. Note that this embedding is by construction $G$-equivariant. 
We can now calculate the moment map of this extended action, at least for elements in the image of this embedding.
\begin{lem}
The moment map on the image of the Veronese embedding $\nu$ of $P(H)$ into $P(\mathrm{Sym}^kH)$ is given by $\Phi_{P(\mathrm{Sym}^kH)}\big|_\nu(H) = k\cdot\Phi_H$.
\end{lem} 
\begin{proof}
This follows from the following calculation: 
\begin{align*}
\Phi_{P(\mathrm{Sym}^kH)}\big|_{\nu(H)}(\ket{\psi}) &= (A \mapsto \braket{\psi^{\otimes k}|\rho_{\mathrm{Sym}^kH}(A)|\psi^{\otimes k}}) \\
&= (A \mapsto \sum_{i=1}^k \braket{\psi|\rho_H(A)|\psi}) \\ 
&=k\cdot\Phi_H(\ket{\psi})
\end{align*}
\end{proof}
\begin{lem}
For any two $G$-Actions on $P(V_1)$ and $P(V_2)$ with moment maps $\Phi_1$ and $\Phi_2$ respectively, the moment map $\Phi$ on the image of the Segre embedding 
$$\iota: P(V_1)\times P(V_2) \hookrightarrow P(V_1 \otimes V_2) $$ 
is given by
$$\Phi(\iota(\ketbra{\phi},\ketbra{\psi})) = \Phi_1(\ketbra{\phi})+\Phi_2(\ketbra{\psi})$$
\end{lem}
\begin{proof}
This follows from the description of the moment map as 
\begin{align*}
\Phi(\iota(\ketbra{\phi},\ketbra{\psi}) & = \Phi(\ketbra{\phi \otimes \psi}) \\
&= (A \mapsto \braket{\phi\otimes\psi|\rho_V(A)|\phi\otimes\psi}) \\
&=(A \mapsto \braket{\phi|\rho_{V_1}(A)|\phi}+\braket{\psi|\rho_{V_2}(A)|\psi} \\
&= \Phi_1(\ketbra{\phi})+\Phi_2(\ketbra{\psi})
\end{align*}
where we have used that $\rho_{V_1 \otimes V_2} = \rho_{V_1} \otimes \mathbb{1}_{V_2} + \mathbb{1}_{V_1} \otimes \rho_{V_2}$.
\end{proof}

\subsection{Using the gradient flow of the extended action}
We can now look at the critical points of the norm square of the moment map $$\tilde{\Phi}: P(\mathrm{Sym}^k \otimes V_{\lambda^*}) \rightarrow Lie(K)^*$$ restricted to the image $\tilde{X}$ of the embedding \ref{embedding1}. 
\begin{prop}
Assume that $$\mathrm{grad}\|\tilde{\Phi}\|^2_{\ketbra{\psi},\mu}=0$$
for some $\mu \in O_{\lambda^*}$.
Then $\Phi(K\cdot\ketbra{\psi}) \cap \mathfrak{t}^+$ is the unique point in $\Delta_{\psi}$ closest to $\lambda /k.$.
\end{prop}
\begin{proof}
Let $\tilde{X}_{\psi,\mu} := \overline{G \cdot (\ketbra{\psi},\mu)}\subseteq \tilde{X} \subseteq P(\mathrm{Sym}^kH \otimes V_{\lambda^*})$.
Now $\lambda^* = -w_0(\lambda)$ where $w_0$ is the longest element in the Weyl group of $G$ (with respect to our maximal torus $T$), as discussed also in \cite{Brion1987}. In our case the Weyl group is isomorphic to the product of the local permutation subgroups and therefore can be considered a subgroup of G.  Therefore $O_{\lambda^*} = O_{-w_0(\lambda)} = O_{-\lambda} = - O_{\lambda}$. Let $V \in K$ such that $V\cdot\lambda^* = \mu$. Now, by the lemmata above, $$\|\tilde{\Phi}(\ketbra{\psi},\mu)\|=\|k\Phi(\ketbra{\psi})+\mu\| = k\|V(\Phi(V^{-1}\cdot\ketbra{\psi})-\frac{\lambda}{k})\| = k\|\Phi(V^{-1}\cdot\ketbra{\psi})-\frac{\lambda}{k}\| $$
Other the other hand, \ref{critpts} and \ref{mindistprop} tell us that 
\begin{align*} \|\tilde{\Phi}(\ketbra{\psi},\mu)\| & = \min_{\tilde{X}_{\psi,\mu}}\|\tilde{\Phi}(\ketbra{\psi'},\mu')\| = \inf_{g \in G} \| \tilde{\Phi}(g\cdot \ketbra{\psi},g\cdot\mu)\|\\
&= \inf_{g \in G}\|k\Phi(g\cdot\ketbra{\psi})+ (gV)\cdot\lambda^*\| \\
&= \inf_{g \in G}\|k\Phi(g\cdot\ketbra{\psi})-U(gV)\cdot\lambda\| \\
&= k\inf_{g\in G}\|U(gV)(\Phi(U(gV)^{-1}\cdot(g\cdot\ketbra{\psi}))-\frac{\lambda}{k})\| \\
&= k\inf_{g'\in G}\|\Phi(g'\cdot\ketbra{\psi})-\frac{\lambda}{k}\|
\end{align*}
where we have used equivariance of the moment map and invariance of the norm under $K$. Thus 
$$
k\|\Phi(V^{-1}\cdot\ketbra{\psi})-\frac{\lambda}{k}\| = k\inf_{g'\in G}\|\Phi(g'\cdot\ketbra{\psi})-\frac{\lambda}{k}\|$$
which is exactly what we wanted.
\end{proof}

\section{Implementation of the Gradient Flow}
Unfortunately, there is no such easy recipe to look for a critical point of the norm square of the extended moment map as for the usual one, where we can exploit our discussion about free states. As developed in \cite{Kirwan1984}, the gradient flow will always converge in the polytope, i.e. we have that the limit 
$$\lim_{t \to \infty}\Phi(\phi^t_{\mathrm{grad}\|\Phi\|^2}(\rho))$$ exists for all $\rho \in P(\mathcal{H})$, is contained in the entanglement polytope of $\rho$, and minimises distance to the origin. This is analytically intractable, especially for the extended action. However, one can compute a limit point numerically, by basically using a gradient descent method. In the following we explain how this can be achieved. 
\subsection{Discretising the flow}
Even though one theoretically needs the extended action, it is enough to compute in the image of the embedding \ref{embedding1}, since the gradient flow always stays in the $G$-orbit and thus also in the image of $X \times O_{\lambda^*}$ (since X was defined to be a $G$-orbit closure). Moreover, rather than computing the gradient directly and performing finite size steps in $P(\mathcal{H})$, we will approximate the gradient by approximating the infinitesimal action of $Lie(K)$ via \eqref{gradienteq}. Thus we can ensure the flow always stays in the orbit. Let $\ket{\psi} \in \mathcal{H}$ be a state in whose entanglement polytope we wish to find the closest point to $\frac{\lambda}{k}$. The procedure works as follows (we switch back now to the mathematical convention $\tr\xi = 0$ for $\xi \in Lie(K)$): 
\begin{enumerate}[1.]
\item Define $\psi_0 = \psi, G_0 = id,\lambda^* = w_0(-\lambda) + (\frac{k}{d_1},\ldots,\frac{k}{d_N})$ (to ensure $\tr\lambda^*=0$) and choose a "random" tuple of unitary matrices $U_0$, corresponding to a starting point $(\ketbra{\psi},U\cdot\lambda^*)$ for the flow. 
\item Calculate the image of the extended moment map $$\xi_0 = \tilde{\Phi}(\ketbra{\psi_0},U_0\cdot\lambda) = k\cdot\Phi(\ketbra{\psi})+U_0\lambda^* U_0^{\dagger} = \left(k\cdot(\rho^{(1)}-\frac{\mathds{1}}{d_1})+U_0^{(1)}\mathrm{diag}((\lambda^*)^{(1)})(U_0^{(1)})^{\dagger},\ldots\right)$$ 
\item Compute the new group element by approximating the infinitesimal action of $\xi_0$ with a finite stepsize $h$,$G_1 = G_0(\exp(-h\cdot\xi_0))$, then the new state $\psi_1 = G_1 \cdot \psi_0$, and the new point in the coadjoint orbit $U_1 = U(G_1U_0)$ (with $U(G)$ as above). 
\item Now compute $\xi_1 = \tilde{\Phi}(\ketbra{\psi_1},U_1\cdot\lambda)$. Since we are minimising the norm of $\tilde{\Phi}(\ketbra{\psi},U\cdot\lambda)$ we check whether $\|x1\|-\|x0\| < -c$, where $c>0$ is some constant.
If this is the case, we continue with step 3 replacing $G_0$ by $G_1$, $\psi_0$ by $\psi_1$, and $U_0$ by $U_1$, and increase the stepsize by a factor of 1.1. If we have not made progress, we redo step 3 with half the stepsize. 
\item Repeat 3 and 4 until we meet some exit criterion (stepsize too small, no progress for a certain number of tries, maximal number of steps reached, etc.
\end{enumerate}

An example of how the flow works for three qubits can be seen in the figure below where we flow to the point  $(\frac{1}{2},\frac{1}{3},\frac{1}{3})$ (corresponding to $k$ = 6 and $\lambda = ((3,3),(4,2),(4,2))$) from various starting states.

\begin{figure}[b]
\centering
\begin{subfigure}[h]{0.45\textwidth}
\centering
\includegraphics[width =\textwidth]{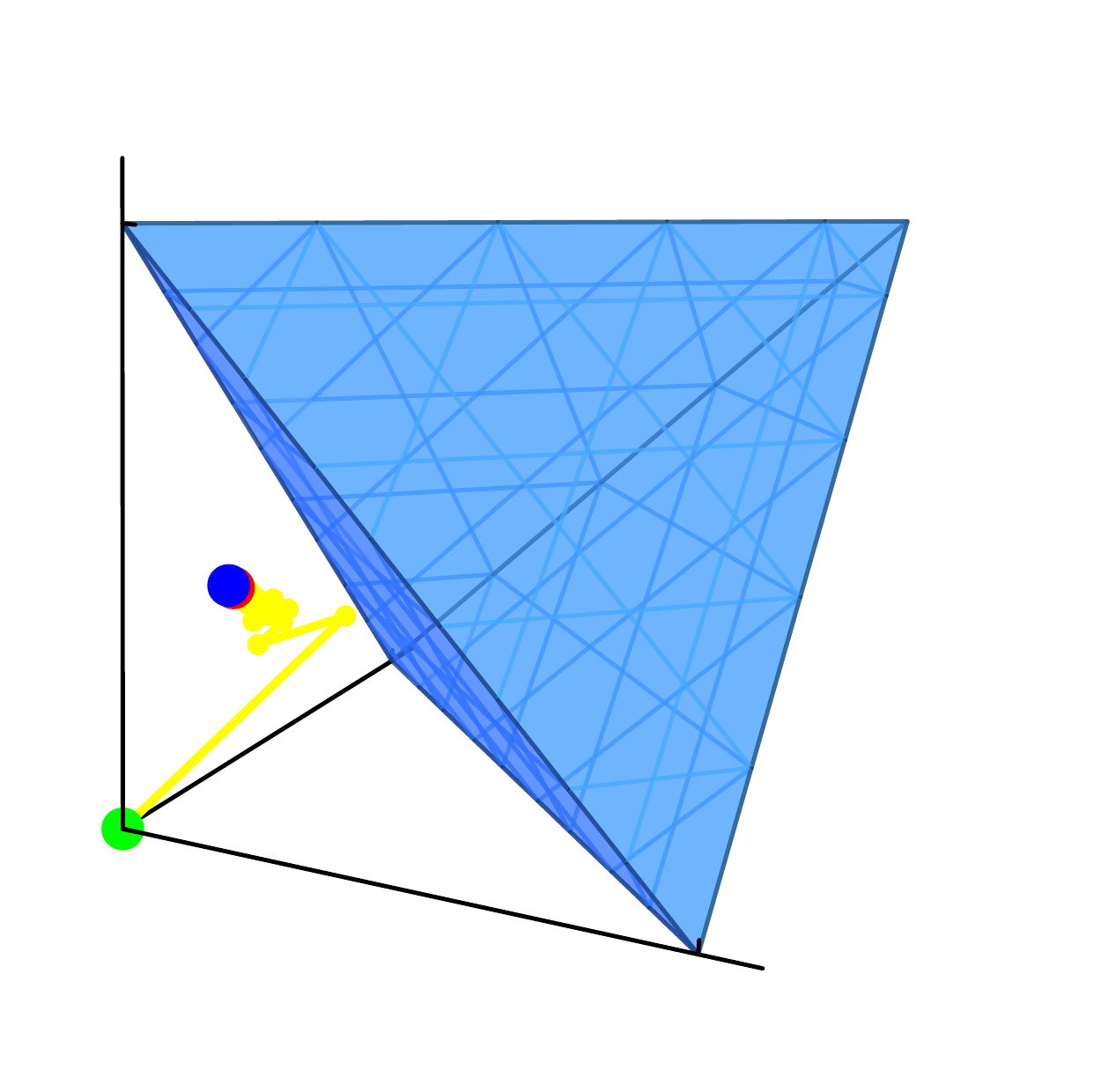}
\caption{Flowing to $p$ starting from the GHZ state}
\end{subfigure}
\begin{subfigure}[h]{0.45\textwidth}
\centering
\includegraphics[width =\textwidth]{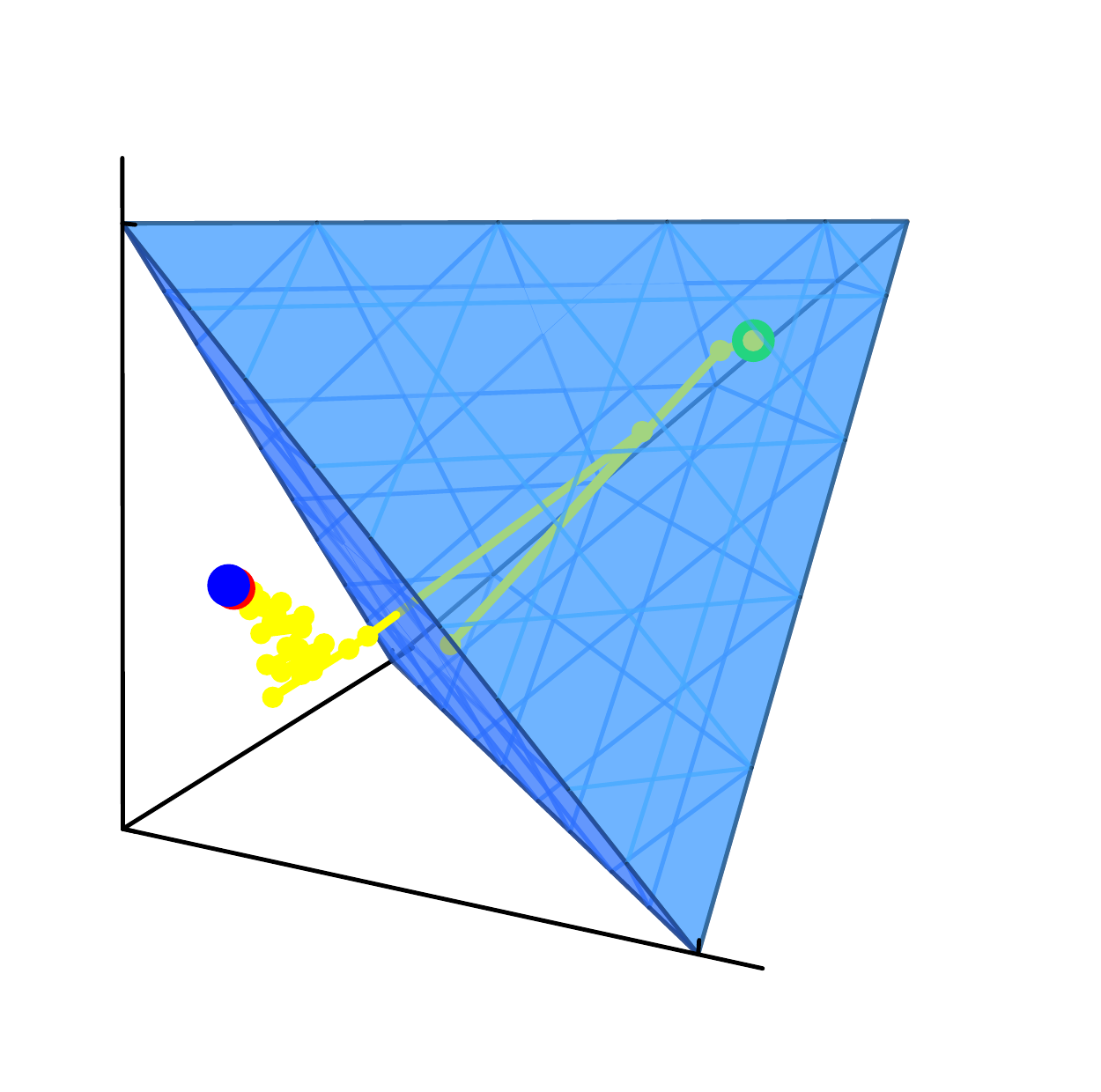}
\caption{Flowing to $p$ from a different point in the SLOCC orbit of the GHZ state}
\end{subfigure}

\begin{subfigure}[h]{0.45\textwidth}
\centering
\includegraphics[width =\textwidth]{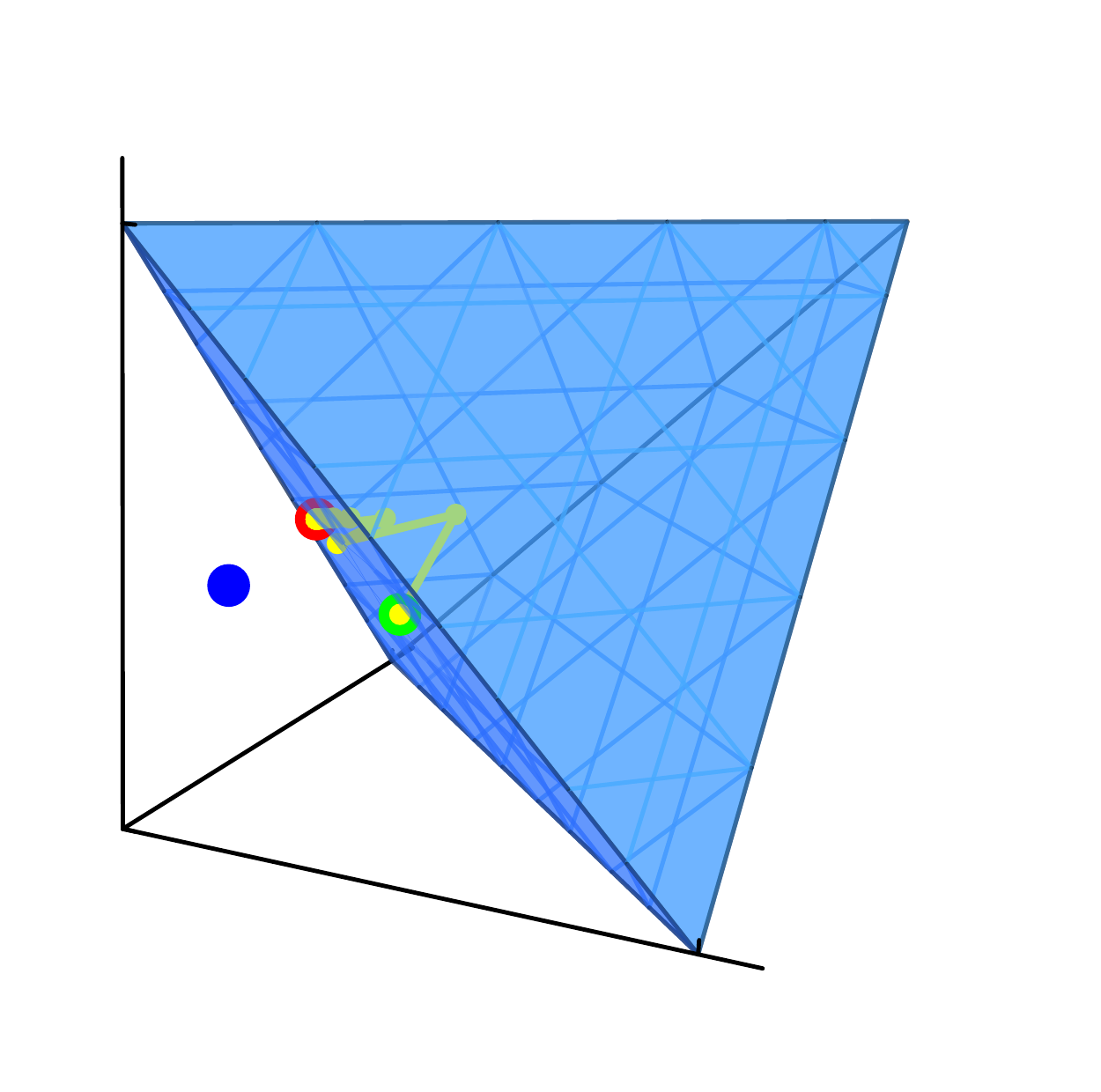}
\caption{Flowing to $p$ starting from the W state. The trajectory always stays in the W polytope and finally arrives at the point in the W polytope closest to $p$}
\end{subfigure}
\qquad\qquad
\begin{subfigure}[h]{0.45\textwidth}
\centering
\includegraphics[width =\textwidth]{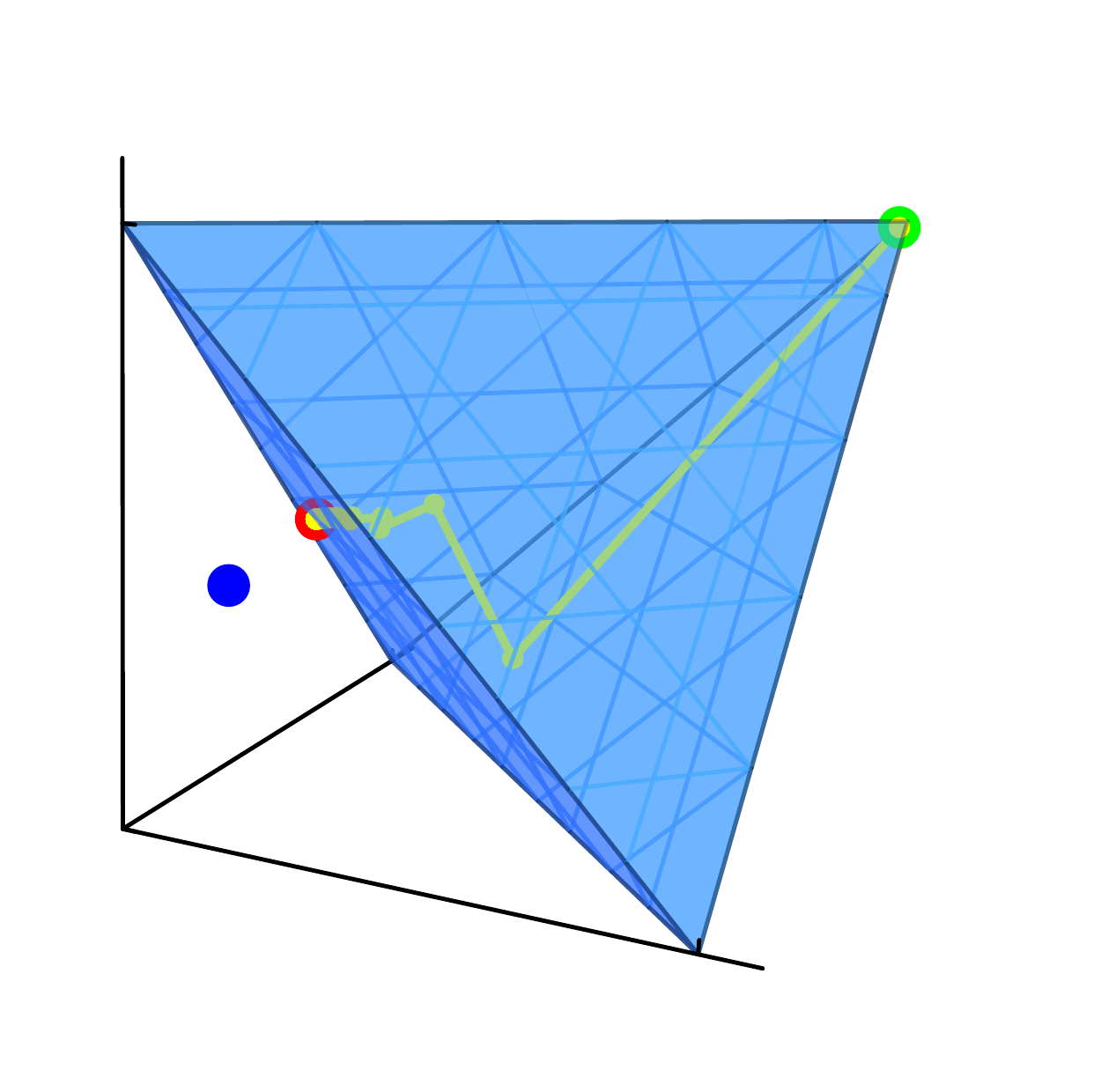}
\caption{Flowing to $p$ starting from the  a different state in the W orbit. The orbit always stays in the W polytope and finally arrives at the point in the W polytope closest to $p$}
\end{subfigure}
\caption[Illustration of the gradient flow to a different point for three qubits]{Applying the gradient flow to the point $p=(0.55,0.66,0.66)$ (blue) starting at states in the GHZ and W orbits. The green point denotes the starting point and the red point denotes the final point. Notice how the flow always stays inside the W polytope if we start at a point in the W orbit.}
\end{figure}
\FloatBarrier
\section{Semi-Interactive polytope computation(SIC)}\label{section:SIC}
The idea for this algorithm and most of its implementation comes from Michael Walter. \\ 
The main use of the gradient flow to calculate the polytopes is the following easy observation made already in chapter \ref{ch:entpoly}. Suppose $p \in \Delta_{\psi}$ is the closest point to $p' \notin \Delta_{\psi}$. Then, by convexity of the entanglement polytopes, we get a halfspace which contains $\Delta_{\psi}$ in form of a linear inequality: Namely, for all points $x \in \Delta_{\psi}$, we have $\langle x - p',p - p' \rangle \geq \langle p - p', p - p'\rangle$. 
\subsection{The algorithm}\label{section:SICalgo}
We now want to exploit this to compute polytopes. For this, we suppose for a moment that our gradient flow is exact, and that we can exactly solve ``convexity problems''  i.e. that we can compute the extremal points $E(\bigcap_i H_i)$ of an intersection of halfspaces and the linear inequalities describing the convex hull of a set of points With these assumptions, we now present an algorithm that, given any state, finds its entanglement polytope in finite time. We make use of the fact that moment polytopes always have rational vertices (see \cite{Brion1987}). Therefore, also all inequalities can only have rational coefficients, which means we can make them integral by multiplying the inequality with a suitable integer. 
\begin{enumerate}
\item Let $\mathcal{I}$ be the set  of local constraints $l(x)\leq a$ for the eigenvalues (i.e the eigenvalues are ordered, positive, and sum to 1). By multiplying those with a suitable integer, we can assume $\mathcal{I}$ consists of integral inequalities. 
\item Define the sets of found vertices $\mathcal{V}_{\mathrm{found}} := \emptyset$ and expected vertices $\mathcal{V}_{\mathrm{expected}} := E(\bigcap\mathcal{I})$ Notice that since the inequalities are integral, $\mathcal{V}_{\mathrm{expected}}$ consist of rational points.
\item Choose any $v \in \mathcal{V}_{\mathrm{expected}}$ and remove it from $\mathcal{V}_{\mathrm{expected}}$. Apply the gradient flow in direction of $v$. on $\ketbra{\psi}$ (which we can do because $v$ is rational). There are now two possibilities: We either find that $v \in \Delta_{\psi}$ or $v \notin \Delta_{\psi}$.
\item \begin{itemize}
\item If $v \notin \Delta_{\psi}$, we get a linear inequality $l'(x) \leq a'$ with integral coefficients. We then add this inequality to $\mathcal{I}$ and redefine $\mathcal{V}_{\mathrm{expected}} := E(\bigcap\mathcal{I})$
\item If $v \in \Delta_{\psi}$, we add $v$ to $\mathcal{V}_{\mathrm{found}}$.
\end{itemize}
\item Repeat steps 3 and 4 until $\mathcal{V}_{\mathrm{expected}}=\emptyset$.
\end{enumerate}
It is clear that this algorithm terminates after finitely many steps and is correct if all our results are exact. 

\subsection{Implementation}
Implementation of this algorithm would be easy if we actually could compute everything. But as it is, we must make do with our numerical approximations\footnote{Of course one could at least solve the convexity problems exactly in theory, but this gets too inefficient as we go for higher dimension.}. It turns out that the most problematic thing is the inequality we get when we did not find a vertex. Since we need to get rational inequalities, they must be rounded. Until now, attempts at automatising the rounding procedure have failed, i.e. after every vertex which we did not find we must round and enter the new inequality by hand. This is why the implemented version is called Semi-Interactive Polytope Computation - the vertices are computed automatically but the inequalities have to be guessed. 
One can decrease the amount of work significantly by adding the inequalities for the generic polytope, i.e the inequalities solving the quantum marginal problem, which are known thanks to Klyachko \cite{Klyachko2007}, to the set $\mathcal{I}$ in step 1 of the algorithm.  \\

Note that since our gradient flow always stays in the orbit polytope, also our approximation to the orbit polytope will always be contained in the orbit polytope (if the algorithm terminates). 

\subsection{Application}
SIC has been succesfully used to compute all polytopes of $2 \times 3 \times N$ systems. It has also been applied to systems of at most 6 qubits and $2\times 4 \times N$ systems. In theory, all that is needed to actually compute more polytopes is more patience and a bit of practice in how to use SIC. However, the numerics failed completely for systems of size $2\times 4 \times 7$ and $2\times 4 \times 8$. In the $2\times 4 \times 7$  case it seems that the polytopes are too ``thin'' to find an interior point for stable computation of the convex hull, while the $2\times 4 \times 8$ case caused Mathematica to exit on the computers used in the course of this thesis\footnote{Admittedly,  not the most high-end machines.}. Also, some attempts were made at applying it to three qutrits and $4\times 4 \times 4$ systems, and no problems were found so far other than increasing runtime. 
\begin{figure}
\centering
\begin{subfigure}[h]{0.45\textwidth}
\centering
\includegraphics[width =\textwidth]{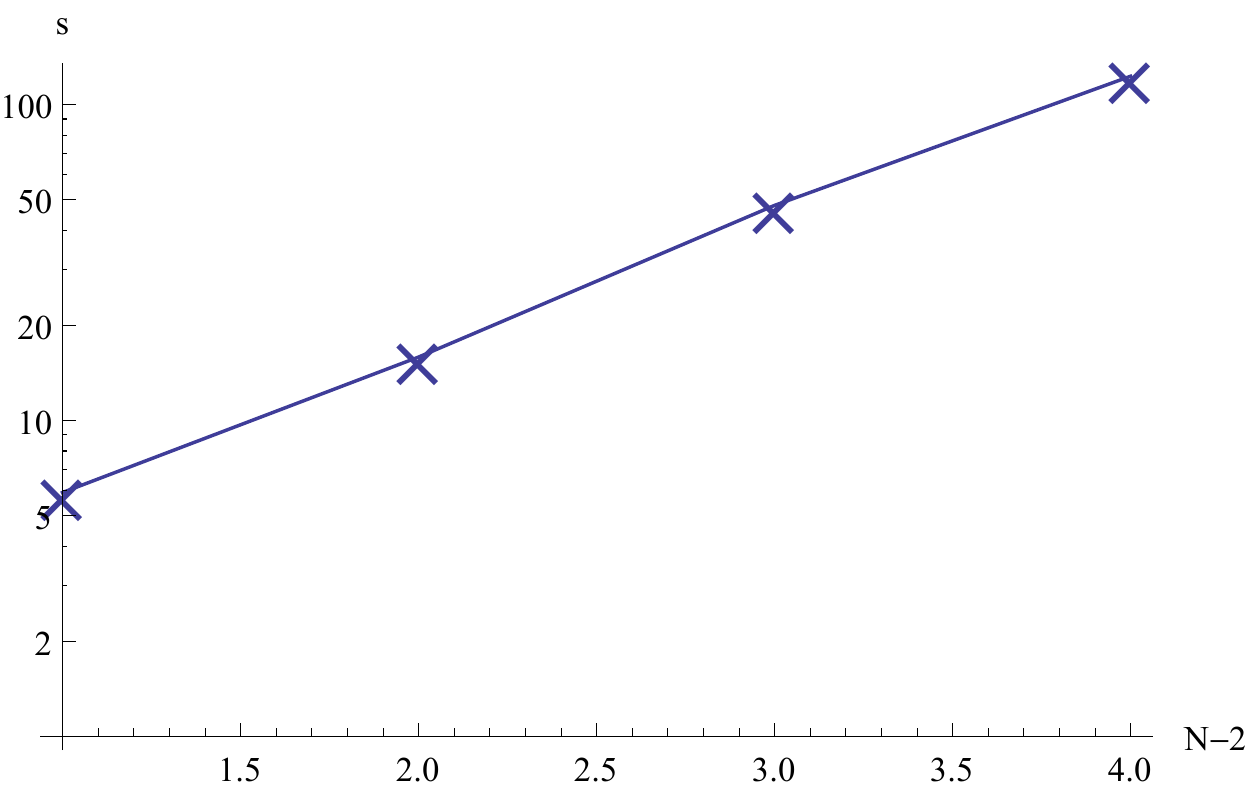}
\caption{Logarithmic Plot of SIC runtime for generic polytope in $2 \times 3 \times N$ systems}
\end{subfigure}
\begin{subfigure}[h]{0.45\textwidth}
\centering
\includegraphics[width =\textwidth]{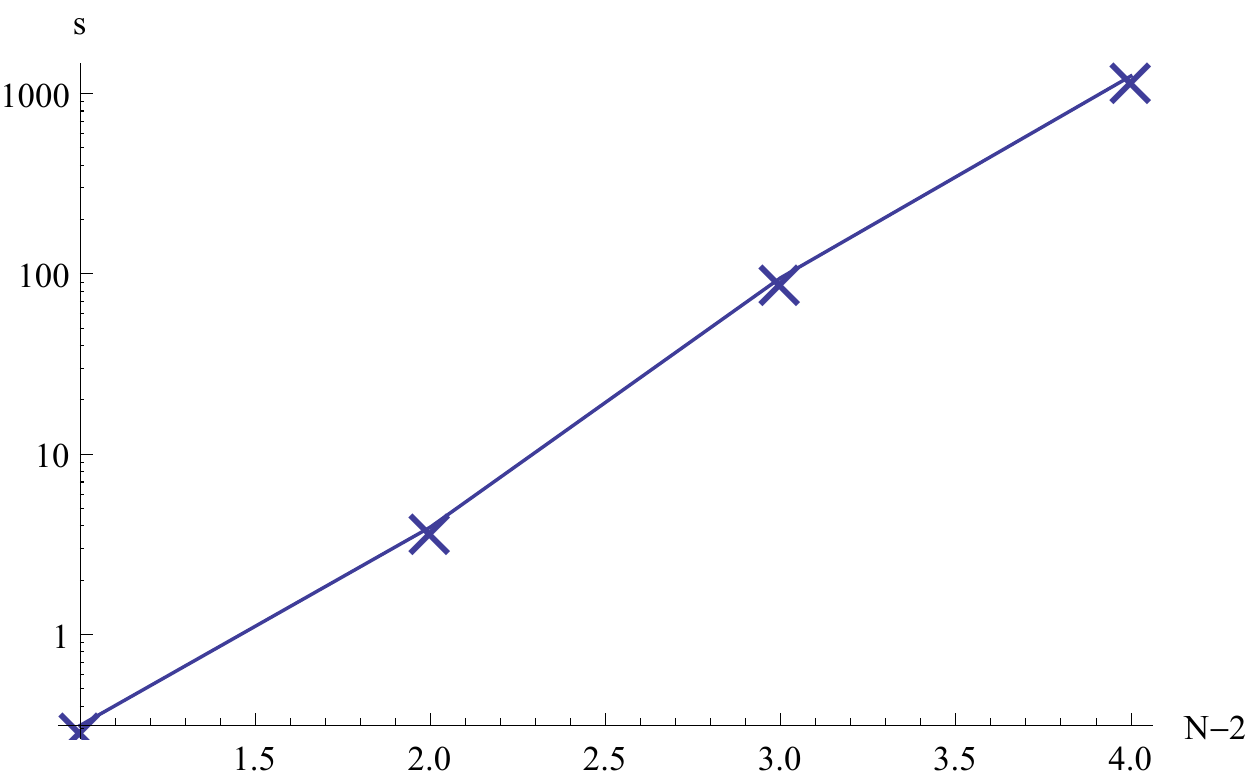}
\caption{Logarithmic Plot of SIC runtime for generic polytope in  $N$ Qubit systems}
\end{subfigure}
\caption[SIC Runtimes]{SIC runtime for generic polytope verification for $2\times 3 \times N$ systems (a) and $N$ Qubit systems (b) on a standard home computer. Notice runtime is exponential both in system size and local dimensions.}
\end{figure}
\chapter{Overview of the numerical Results}
In this chapter we give a brief overview over the results which were obtained using SIC. The full results can be found in the appendix. For $2\times 2 \times n $ systems, we have already computed all polytopes by hand.
\section{$2 \times 3 \times n$}
\begin{table}
\renewcommand{\arraystretch}{1.2}
\begin{equation*}
\begin{array}{c|c}
\text{Local Dimensions} & \text{Representative of SLOCC class} \\ \hline
2\times 3\times N(N \geq 6) & \ket{000} + \ket{011} + \ket{022} + \ket{103}+\ket{114} + \ket{125} \\ \hline
\multirow{2}{*}{$2\times 3\times 5$} & \ket{024} + \ket{000} + \ket{011}+\ket{102} + \ket{113} \\
&\ket{024} + \ket{121} + \ket{000} + \ket{011} +\ket{102} + \ket{113} \\ \hline
\multirow{5}{*}{$2\times 3\times 4$} & \ket{123} + \ket{012} + \ket{000} + \ket{101} \\ 
&\ket{023} + \ket{012} + \ket{000} + \ket{101} \\
&\ket{123} + \ket{012} + \ket{110} + \ket{000} + \ket{101}\\
&\ket{023} + \ket{122} + \ket{012} + \ket{000} + \ket{101}\\
&\ket{023} + \ket{122} + \ket{012} + \ket{110} + \ket{000} + \ket{101} \\ \hline
\multirow{6}{*}{$2\times 3\times 3$} &\ket{000} + \ket{111} + \ket{022} \\
&\ket{000} + \ket{111} + \ket{022} + \ket{122} \\
&\ket{010} + \ket{001} + \ket{112} + \ket{121} \\
&\ket{100} + \ket{010} + \ket{001} + \ket{112} + \ket{121}\\
&\ket{100} + \ket{010} + \ket{001} + \ket{022} \\
&\ket{100} + \ket{010} + \ket{001} + \ket{122} \\ \hline
\multirow{2}{*}{$2\times 3\times 2$}&\ket{000} + \ket{011} + \ket{121} \\
&\ket{000} + \ket{011} + \ket{110} + \ket{121} \\ \hline
2\times 2\times 4 &\ket{000} + \ket{011} + \ket{102} + \ket{113} \\ \hline
\multirow{2}{*}{$2\times 2\times 3$}&\ket{000} + \ket{011} + \ket{112} \\ 
&\ket{000} + \ket{011} + \ket{101} + \ket{112} \\ \hline
\multirow{2}{*}{$2\times 2\times 2$} &\ket{000} + \ket{111}\\
&\ket{001} + \ket{010} + \ket{100} \\ \hline
1\times 3\times 3 &\ket{000} + \ket{011} + \ket{022} \\ \hline
1\times 2\times 2 &\ket{000} + \ket{011}\\ \hline
2\times 1\times 2 &\ket{000} + \ket{101}\\ \hline
2\times 2\times 1 &\ket{000} + \ket{110}\\ \hline
1\times 1\times 1 &\ket{000}
\end{array}
\end{equation*}
\caption{Representatives of all SLOCC entanglement classes on systems of size $2\times 3 \times N$, as shown in \cite{Chen2009}.\label{table:23Nentclasses}}
\end{table}
Table \ref{table:23Nentclasses} shows representatives for all SLOCC classes of systems of local dimensions $ 2 \times 3 \times N$. Notice that for such dimensions we always have finitely many orbits. 
For each of those we have computed the entanglement polytope in terms of inequalities separating it from the generic one. In this section we give a brief overview on the structure of the results. 
\subsection{$ 2\times 3 \times 3$}
In section \ref{section:233polytopes} we have calculated the top two entanglement polytopes for this system. It turns out that already for this system, the polytopes of the lesser entangled top-rank classes are increasingly complicated. This shows that one cannot hope for an "easy" method to compute these polytopes, i.e. it is almost impossible to guess an exhaustive set of inequalities (as we did for the $2\times 2 \times 3$ case) in general: The most complicated case already includes 13 inequalities. 
\subsection{$2 \times 3 \times 4$}
This is the lowest-dimensional case where all systems have different dimensions. This is mirrored in an asymmetry of the inclusion of the polytopes in the generic. By Theorem \ref{thm:magiclemma}, the hierarchy of the classes can be read directly read off from the representatives (it is just given by inclusion of their supports). The polytopes  of $\ket{123} + \ket{012} + \ket{110} + \ket{000} + \ket{101}$ and $\ket{123} + \ket{012} + \ket{000} + \ket{101}$ have a nice inclusion in the generic polytope, given by three inequalities or even a single one respectively, compare \ref{table:234orbitpol}. The inclusion of the other polytopes is far more complicated, requiring 9 respectively 16 inequalities.
\subsection{$2 \times 3 \times 5$}
Here, one thing that can be observed is that the inequalities calculated by SIC are considerably "easier" than the closest point inequality, here $3 x_{1,1}+4 x_{2,1}+4 x_{2,2}+4 x_{3,1}+3 x_{3,2}+3 x_{3,3}\geq 7$, i.e. the coefficients are lower. This can be observed in various systems and is probably due to the fact that the closest point to the origin almost never lies in the interior of a facet of the polytope but on the intersection of several of them. Thus it can happen that the closest point inequality is not itself a facet of the polytope, but just any seperating hyperplane between the origin and the polytope. In this case, for example, the closest point inequality is implied by the two inequalities
\begin{align*}
&-x_{1,1}-x_{2,1}-2 x_{2,2}-2 x_{3,1}-2 x_{3,2}-2 x_{3,3}+3\leq 0 \\
& -x_{1,1}-2 x_{2,1}-2 x_{2,2}-2 x_{3,1}-x_{3,2}-x_{3,3}+3\leq 0 
\end{align*}
One would assume that those are actually facets of the polytope. 
\section{$2 \times 4 \times N$}
Although the numerical exploration of the $2 \times 4 \times N$ polytopes is yet at the beginning, several new patterns can be observed. For finitely many orbits the hierarchy of the SLOCC classes is mirrored by that of the polytopes, i.e every class has its own polytope. 
As we leave the realm of finitely many orbits, this is no longer true (as was also witnessed in \cite{Walter2012} for the polytopes of four qubits). In the $2 \times 4 \times 4$ case, we have computed at least two classes which map to the full polytope. Moreover, while for finitely many orbits it is always true that only class (namely the generic one) does include the origin, the $2 \times 4 \times 4$ case contains an example of a class which includes the origin but does not have the full polytope.
Also some examples of higher-dimensional polytopes were computed, but only up $2 \times 4 \times 6$. For $2 \times 4 \times 7$ the numerics failed. 
\section{Higher dimensional systems} 
We also tried to apply the method to $3\times 3 \times 3$ and $4 \times 4 \times 4 $ systems. In the three Qutrits case, we were interested in two states, namely the three-qubit GHZ state
$\ket{000}+\ket{111}+\ket{222}$, also known as unit tensor [REF?] and the state corresponding to multiplication of upper triangular $2 \times 2$ matrices. One version of this tensor is\footnote{Which can intuitively be explained as follows: Denote $\ket{0} = \begin{pmatrix}
1 & 0 \\ 0 & 0 \end{pmatrix}, \ket{1} = \begin{pmatrix} 0 & 1 \\ 0 & 0 \end{pmatrix},\ket{2} = \begin{pmatrix} 0 & 0 \\ 0 & 1 \end{pmatrix}$ then we get a basis vector $\ket{ijk}$ in the tensor if $\ket{i}\cdot\ket{j}=\ket{k}$.} $\ket{000}+\ket{222}+\ket{011}+\ket{121}$. It turns out that the unit tensor has the full polytope while there is a number of inequalities for the polytope of the upper triangular matrix multiplication tensor. \\
Actually we would be interested in the tensor in the $4 \times 4 \times 4$ system corresponding to multiplication of arbitrary $2 \times 2$ matrices. The problem in this case is that in contrast to the three qutrit case, where marginal inequalities were given by Franz et al. in \cite{Franz2000} we do not know the full polytope. This leaves "too much in the dark" to actually use SIC, especially since we're already dealing with a 64-dimensional space, and SIC is accordingly slow and inaccurate. 
\chapter{Conclusions}
\section{Results}
We can now answer the questions asked in the introduction in more detail. We can say that some methods have been found to analytically compute entanglement polytopes in low dimensions. The closest point method can in theory be used for states satisfying the conditions in any dimensions. However, we also have seen that the value of these inequalities decreases in higher dimensions due to the geometry of the problem, i.e the closest point in the origin will typically lie in the intersection of many facets of the polytope, and the inequality we receive is mostly not among them and thus does not give a tight bound for the polytope. \\ 

More results for larger dimensions were obtained using numerical methods. These methods have shown that in general there can be a lot of inequalities for a single polytope, i.e completely classifyfing the polytopes without numerical help either requires a significant theoretical advancement or knowledge of a generating set of covariants. \\

We have computed all polytopes for $2\times 3 \times N$ systems. The $2\times 4 \times N $ and three qutrits cases could theoretically also be treated in more detail (even though, as already stated, our numerics fail completely for $2\times 4 \times 7$ and $2\times 4 \times 8$). The only thing prevented us from doing so was the deadline of the thesis, since computing the polytopes usually takes patience. \\

What is actually remarkable is the simplicity (very low integral coefficients for the most part). This resembles Klyachko's findings in \cite{Klyachko2007}, where the inequalities given for the full polytope are presented in the appendix. 

\section{Outlook}
Given from what we have learned from this thesis, it should be possible to compute polytopes numerically also for larger dimensions. However at some point it then becomes inevitable to automatise SIC. To this end, the quality of the numerics would have to be improved considerably. Of course, the drawback of numerics is that it is not a priori clear how to treat infinite families of SLOCC classes. Also this would have to be subject to further investigation. \\
Moreover, the linear inequalities for the spectra derived from the numerics can be subjected to further theoretical research. For example, if one could calculate the preimage of points under the momemt map, it should be possible to generalise the closest point method discussed here to arbitrary target points, thus giving a way of explicitly proving the inequalities obtained from the numerics. 
\chapter{Acknowledgements} 
I would to thank  
\begin{itemize}
\item Prof. Renato Renner for introducing and elating me to the Topic of Quantum Information Theory in an excellent lecture in 2012, 
\item My supervisor Prof. Matthias Christandl for introducing me to entanglement polytopes and making this thesis possible, 
\item Michael Walter for countless explanations, ideas, lines of mathematica code, reading and correcting and rereading my arguments,  but above all for his infinite amount of patience and helpfulness, 
\item Dr. P\'eter Vrana for sharing his experience and results with me, counterchecking my ideas and pointing me in interesting directions, 
\item Prof. Brent Doran for an occasional chat on invariant theory,
\item my parents for their support and motivation and 
\item last but not least my girlfriend, Ramona Ochsner, for her patience, understanding, care and support. 
\end{itemize}

Without those people this thesis would never have been possible. 

\bibliographystyle{habbrv}
\bibliography{bibliography}

\begin{thebibliography}{10}

\bibitem{Atiyah1982}
M.~F. Atiyah.
\newblock Convexity and commuting hamiltonians.
\newblock {\em Bulletin of the London Mathematical Society}, 14(1):1--15, Jan
  1982.

\bibitem{Ballmann2006}
W.~Ballmann.
\newblock {\em Lectures on K\"ahler Manifolds}.
\newblock European Mathematical Society Publishing House, Jul 2006.

\bibitem{Barber1996}
C.~B. Barber, D.~P. Dobkin, and H.~Huhdanpaa.
\newblock The quickhull algorithm for convex hulls.
\newblock {\em ACM Transactions on Mathematical Software}, 22(4):469--483, Dec
  1996.

\bibitem{Barvinok2002}
A.~Barvinok.
\newblock {\em A Course in Convexity (Graduate Studies in Mathematics, V. 54)}.
\newblock American Mathematical Society, 2002.

\bibitem{Bell1964}
J.~S. Bell.
\newblock {On the Einstein-Podolsky-Rosen paradox}.
\newblock {\em Physics}, 1:195--200, 1964.

\bibitem{Bennett1993}
C.~Bennett, G.~Brassard, C.~Crepeau, R.~Jozsa, A.~Peres, and W.~Wootters.
\newblock Teleporting an unknown quantum state via dual classical and epr
  channels.
\newblock {\em Phys Rev Lett}, pages 1895--1899, March 1993.

\bibitem{Bennett1992}
C.~Bennett and S.~Wiesner.
\newblock Communication via one- and two-particle operators on
  {Einstein}-podolsky-rosen states.
\newblock {\em Physical Review Letters}, 69(20):2881--2884, Nov 1992.

\bibitem{Bennett1984}
C.~H. Bennett and G.~Brassard.
\newblock {Quantum cryptography: Public key distribution and coin tossing}.
\newblock In {\em Proceedings of IEEE International Conference on Computers,
  Systems, and Signal Processing}, pages 175--179, Bangalore, 1984.

\bibitem{Bravyi2004}
S.~Bravyi.
\newblock Requirements for compatibility between local and multipartite quantum
  states.
\newblock {\em Quantum Inf. and Comp., Vol.}, 4,:No.1pp.012--026, 2004,
  quant-ph/0301014.

\bibitem{Brion1987}
M.~Brion.
\newblock {\em Sur l'image de l'application moment}.
\newblock Springer-Verlag, 1987.

\bibitem{Chen2006}
L.~Chen and Y.-X. Chen.
\newblock Range criterion and classification of true entanglement in
  $2\times{M}\times{N}$ system.
\newblock {\em Phys. Rev. A}, 73(2006)052310:73052310, 2006, quant-ph/0506152.

\bibitem{Chen2009}
L.~Chen, Y.-X. Chen, and Y.-X. Mei.
\newblock Classification of multipartite entanglement containing infinitely
  many kinds of states, 2009, quant-ph/0604184.

\bibitem{Duer2000}
W.~D\"ur, G.~Vidal, and J.~I. Cirac.
\newblock Three qubits can be entangled in two inequivalent ways.
\newblock {\em Phys. Rev. A}, 62,:062314, 2000, quant-ph/0005115.

\bibitem{Einstein1935}
A.~Einstein, B.~Podolsky, and N.~Rosen.
\newblock Can quantum-mechanical description of physical reality be considered
  complete?
\newblock {\em Physical Review}, 47(10):777--780, May 1935.

\bibitem{Franz2000}
M.~Franz.
\newblock Moment polytopes of projective g-varieties and tensor products of
  symmetric group representations, 2000.

\bibitem{Fulton1997}
W.~Fulton.
\newblock {\em Young Tableaux: With Applications to Representation Theory and
  Geometry (London Mathematical Society Student Texts, Vol. 35)}.
\newblock Cambridge University Press, 1997.

\bibitem{Fulton1991}
W.~Fulton and J.~Harris.
\newblock {\em Representation Theory: A First Course (Graduate Texts in
  Mathematics / Readings in Mathematics)}.
\newblock Springer, 1991.

\bibitem{Greenberger1990}
D.~M. Greenberger, M.~A. Horne, A.~Shimony, and A.~Zeilinger.
\newblock Bell's theorem without inequalities.
\newblock {\em American Journal of Physics}, 58(12):1131, 1990.

\bibitem{Higuchi2003}
A.~Higuchi, A.~Sudbery, and J.~Szulc.
\newblock One-qubit reduced states of a pure many-qubit state: polygon
  inequalities.
\newblock {\em Phys. Rev. Lett.}, 90,107902, 2003.

\bibitem{Kirwan1984}
F.~C. Kirwan.
\newblock {\em Cohomology of Quotients in Symplectic and Algebraic Geometry
  (Mathematical Notes, Vol. 31)}.
\newblock Princeton University Press, 1984.

\bibitem{Klyachko2007}
A.~Klyachko.
\newblock Quantum marginal problem and representations of the symmetric group.
\newblock 2007, quant-ph/0409113.

\bibitem{Miyake2004}
A.~Miyake and F.~Verstraete.
\newblock Multipartite entanglement in 2 x 2 x n quantum systems.
\newblock {\em Phys. Rev. A}, 69,:012101, 2004, quant-ph/0307067.

\bibitem{Ness1984}
L.~Ness and D.~Mumford.
\newblock A stratification of the null cone via the moment map.
\newblock {\em American Journal of Mathematics}, 106(6):pp. 1281--1329, 1984.

\bibitem{Nielsen2000}
M.~A. Nielsen and I.~L. Chuang.
\newblock {\em Quantum Computation and Quantum Information (Cambridge Series on
  Information and the Natural Sciences)}.
\newblock Cambridge University Press, 2000.

\bibitem{Raussendorf2001}
R.~Raussendorf and H.~J. Briegel.
\newblock A one-way quantum computer.
\newblock {\em Physical Review Letters}, 86(22):5188--5191, May 2001.

\bibitem{Sjamaar1998}
R.~Sjamaar.
\newblock Convexity properties of the moment mapping re-examined.
\newblock {\em Adv. Math.}, 138:,46--91, 1998, dg-ga/9408001.

\bibitem{Smirnov2004}
A.~V. Smirnov.
\newblock Decomposition of symmetric powers of irreducible representations of
  semisimple lie algebras and the brion polytope.
\newblock {\em TRANSACTIONS- MOSCOW MATHEMATICAL SOCIETY C/C OF TRUDY-
  MOSKOVSKOE MATEMATICHESKOE OBSHCHESTVO; 65;}, 65:213--234, 2004.

\bibitem{Walter2012}
M.~Walter, B.~Doran, D.~Gross, and M.~Christandl.
\newblock Entanglement polytopes.
\newblock {\em Science, vol.}, 340,:no.6137,pp.1205--1208, Aug. 2012,
  1208.0365.

\bibitem{Zurek1993}
W.~Zurek, S.~Habib, and J.~Paz.
\newblock Coherent states via decoherence.
\newblock {\em Physical Review Letters}, 70(9):1187--1190, Mar 1993.

\end{thebibliography}

\begin{appendix}
\chapter{Appendix: A collection of all computed inequalities}
In this section we give detailed results: For each system on dimensions smaller than $2\times 3 \times N$, we give the inequalities (both global and local) for the generic polytope\footnote{As always we here use the inequalities given by Klyachko in \cite{Klyachko2007} adjusted to our system.}, all inequalities obtained using the closest point method and exhaustive sets of inequalities for the polytopes of the SLOCC classes computed using SIC. For systems of larger dimensions, we computed only a limited number of cases, which are also included. Still we found at least one inequality for almost every entanglement of systems of sizes between $2  \times 4 \times 4$ and $2\times 4 \times 6$.
\section{$2 \times 2 \times n$}
We repeat the results for $2 \times 2 \times n$ systems. 
\subsection{Three Qubits}
Although the three-qubit case was done already in \cite{Walter2012}, we include it for the sake of completeness.
\subsubsection{Generic Polytope}

\begin{align*}
 x_{1,1}&\geq 0 \\
 x_{2,1}&\geq 0 \\
 x_{3,1}&\geq 0 \\
 x_{1,1}&\leq 1 \\
 x_{2,1}&\leq 1 \\
 x_{3,1}&\leq 1 \\
 2 x_{1,1}&\geq 1 \\
 2 x_{2,1}&\geq 1 \\
 2 x_{3,1}&\geq 1 \\
 x_{1,1}+x_{2,1}&\leq x_{3,1}+1 \\
 x_{1,1}+x_{3,1}&\leq x_{2,1}+1 \\
 x_{2,1}+x_{3,1}&\leq x_{1,1}+1 \\
\end{align*}

\subsubsection{Orbit Polytopes}
\begin{table}[h]
\renewcommand{\arraystretch}{1.2}
\begin{equation*}
\begin{array}{|c|l|}
\multicolumn{2}{c}{\textbf{Orbit polytopes for $2 \times 2 \times 2$}} \\ \hline
\text{Representative of Class} & \text{Additional Inequalities} \\ \hline
\ket{000}+\ket{111} & \text{None} \\ \hline 
\ket{100}+\ket{010}+\ket{001} & x_{1,1}+x_{2,1}+x_{3,1} \geq 2 \\ \hline
\end{array}
\end{equation*}
\caption{Entanglement Polytope Inequalities for three Qubits}
\end{table}
\FloatBarrier
\subsection{$2 \times 2 \times 3$}
\subsubsection{Generic Polytope}

\begin{align*}
 x_{1,1}&\geq 0 \\
 x_{2,1}&\geq 0 \\
 x_{3,2}&\geq 0 \\
 x_{1,1}&\leq 1 \\
 x_{2,1}&\leq 1 \\
 2 x_{1,1}&\geq 1 \\
 2 x_{2,1}&\geq 1 \\
 x_{3,2}&\leq x_{3,1} \\
 x_{3,1}+x_{3,2}&\leq 1 \\
 x_{3,1}+2 x_{3,2}&\geq 1 \\
 x_{1,1}&\leq x_{2,1}+x_{3,2} \\
 x_{1,1}&\leq x_{3,1}+x_{3,2} \\
 x_{2,1}&\leq x_{1,1}+x_{3,2} \\
 x_{2,1}&\leq x_{3,1}+x_{3,2} \\
 x_{1,1}+x_{2,1}&\leq x_{3,1}+1 \\
 x_{1,1}+1&\leq x_{2,1}+2 x_{3,1}+x_{3,2} \\
 x_{2,1}+1&\leq x_{1,1}+2 x_{3,1}+x_{3,2} \\
\end{align*}
\subsubsection{Orbit Polytopes}
\begin{table}[h]
\renewcommand{\arraystretch}{1.2}
\begin{equation*}
\begin{array}{|c|l|}
\multicolumn{2}{c}{\textbf{Orbit polytopes for $2 \times \times 3$}} \\ \hline
\text{Representative of Class} & \text{Additional Inequalities} \\ \hline
\ket{000}+\ket{101}+\ket{011}+\ket{112} & \text{None} \\ \hline 
\multirow{3}{*}{$\ket{000}+\ket{011}+\ket{112}$} & x_{1,1}+x_{2,1}+x_{3,1}+x_{3,2} \geq 2 \\
& x_{1,1}+x_{3,1} \geq 1 \\
& x_{2,1}+x_{3,1} \geq 1 \\ \hline
\end{array}
\end{equation*}
\caption{Entanglement Polytope Inequalities for $2 \times 2 \times 3$}
\end{table}
\subsection{$2\times 2\times 4$}
\subsubsection{Generic Polytope}
\begin{align*}
 x_{1,1}&\geq 0 \\
 x_{2,1}&\geq 0 \\
 x_{3,3}&\geq 0 \\
 x_{1,1}&\leq 1 \\
 x_{2,1}&\leq 1 \\
 2 x_{1,1}&\geq 1 \\
 2 x_{2,1}&\geq 1 \\
 x_{3,2}&\leq x_{3,1} \\
 x_{3,3}&\leq x_{3,2} \\
 x_{1,1}&\leq x_{3,1}+x_{3,2} \\
 x_{2,1}&\leq x_{3,1}+x_{3,2} \\
 x_{3,1}+x_{3,2}+x_{3,3}&\leq 1 \\
 x_{3,1}+x_{3,2}+2 x_{3,3}&\geq 1 \\
 x_{1,1}+x_{3,3}&\leq x_{2,1}+x_{3,1} \\
 x_{2,1}+x_{3,3}&\leq x_{1,1}+x_{3,1} \\
 x_{1,1}+x_{2,1}&\leq 2 x_{3,1}+x_{3,2}+x_{3,3} \\
 x_{1,1}+1&\leq x_{2,1}+x_{3,1}+2 x_{3,2}+x_{3,3} \\
 x_{2,1}+1&\leq x_{1,1}+x_{3,1}+2 x_{3,2}+x_{3,3} \\
\end{align*}
\subsubsection{Orbit Polytopes}
\begin{table}[h]
\renewcommand{\arraystretch}{1.2}
\begin{equation*}
\begin{array}{|c|l|}
\multicolumn{2}{c}{\textbf{Orbit polytopes for $2 \times \times 4$}} \\ \hline
\text{Representative of Class} & \text{Additional Inequalities} \\ \hline
\ket{000}+\ket{101}+\ket{012}+\ket{113} & \text{None} \\ \hline 
\end{array}
\end{equation*}
\caption{Entanglement Polytope Inequalities for $2 \times 2 \times 4$}
\end{table}
\pagebreak
\section{$2 \times 3 \times N$}
\FloatBarrier
We give first the inequalities obtained from the closest point method for every $2 \times 3 \times N $ orbit. 
\subsection{Closest Point Inequalities}
\begin{table}[h]
\renewcommand{\arraystretch}{1.2}
\begin{math}
\begin{array}{|p{3cm}|p{5cm}|c|}
\multicolumn{3}{c}{\textbf{Closest Point inequalities for } 2 \times 3 \times N} \\ \hline 
\text{Local Dimensions} & \text{Representative of SLOCC class} & \text{Inequality} \\ \hline
$2 \times 3\times N(N \geq 6) $&$ \ket{000} + \ket{011} + \ket{022} + \ket{103}+\ket{114} + \ket{125} $& \text{n/a} \\ \hline
\multirow{2}{*}{$2 \times 3 \times 5$} & $\ket{024} + \ket{000} + \ket{011}+\ket{102} + \ket{113}$ & 3 x_{1,1}+4 x_{2,1}+4 x_{2,2}+4 x_{3,1}+3 x_{3,2}+3 x_{3,3}\geq 7 \\
&$\ket{024} + \ket{121} + \ket{000} + \ket{011} +\ket{102} + \ket{113}$ & \text{n/a} \\ \hline
\multirow{5}{*}{$2 \times 3 \times 4$} &$ \ket{123} + \ket{012} + \ket{000} + \ket{101} $ & x_{2,1}+x_{3,1}+x_{3,2}\geq 1 \\ 
&$\ket{023} + \ket{012} + \ket{000} + \ket{101}$ & x_{1,1}+2 x_{2,1}+3 x_{3,1}+2 x_{3,2}+2 x_{3,3}\geq 5 \\
&$\ket{123} + \ket{012} + \ket{110} + \ket{000} + \ket{101}$ & x_{1,1}+2 x_{2,1}+x_{2,2}+2 x_{3,1}+2 x_{3,2}+x_{3,3}\geq 3\\
&$\ket{023} + \ket{122} + \ket{012} + \ket{000} + \ket{101}$ & 2 x_{1,1}+3 x_{2,1}+2 x_{2,2}+3 x_{3,1}+2 x_{3,2}+x_{3,3}\geq 5 \\
&$\ket{023} + \ket{122} + \ket{012} + \ket{110} + \ket{000} + \ket{101} $ & \text{n/a} \\ \hline
\multirow{6}{*}{$2 \times 3 \times 3$} &$\ket{000} + \ket{111} + \ket{022}$ & 2 x_{1,1}+x_{2,1}+x_{3,1}\geq 2 \\
&$\ket{000} + \ket{111} + \ket{022} + \ket{122}$ & \text{n/a} \\
&$\ket{010} + \ket{001} + \ket{112} + \ket{121}$ & 2 x_{1,1}+2 x_{2,1}+x_{2,2}+2 x_{3,1}+x_{3,2}\geq 4 \\
&$\ket{100} + \ket{010} + \ket{001} + \ket{112} + \ket{121}$ & x_{1,1}+2 x_{2,1}+x_{2,2}+2 x_{3,1}+x_{3,2}\geq 3\\
&$\ket{100} + \ket{010} + \ket{001} + \ket{022}$ & 2 x_{1,1}+2 x_{2,1}+x_{2,2}+2 x_{3,1}+x_{3,2}\geq 4 \\
&$\ket{100} + \ket{010} + \ket{001} + \ket{122}$ & x_{1,1}+x_{2,1}+x_{2,2}+x_{3,1}+x_{3,2}\geq 2 \\ \hline
\end{array}
\end{math}
\caption{Inequalities calculated using the closest point method on the states in table \ref{table:23Nentclasses}.\label{table:23Ncpineqs}}
\end{table}
\FloatBarrier
\pagebreak
\subsection{$2 \times 3 \times 3$}
\subsubsection{Generic Polytope}
\begin{multicols}{2}
\begin{fleqn}
\begin{align*}
 x_{1,1}&\geq 0 \\
 x_{2,1}&\geq 0 \\
 x_{2,2}&\geq 0 \\
 x_{3,2}&\geq 0 \\
 x_{1,1}&\leq 1 \\
 2 x_{1,1}&\geq 1 \\
 x_{2,2}&\leq x_{2,1} \\
 x_{3,2}&\leq x_{3,1} \\
 x_{2,1}+x_{2,2}&\leq 1 \\
 x_{3,1}+x_{3,2}&\leq 1 \\
 x_{2,1}+2 x_{2,2}&\geq 1 \\
 x_{3,1}+2 x_{3,2}&\geq 1 \\
 x_{2,1}&\leq x_{1,1}+x_{3,2} \\
 x_{2,1}&\leq x_{3,1}+x_{3,2} \\
 x_{2,2}&\leq x_{1,1}+x_{3,1} \\
 x_{3,1}&\leq x_{1,1}+x_{2,2} \\
 x_{3,1}&\leq x_{2,1}+x_{2,2} \\
 x_{1,1}+x_{2,1}&\leq x_{3,1}+1 \\
 x_{1,1}+x_{2,2}&\leq x_{3,1}+1 \\
 x_{1,1}+x_{2,2}&\leq x_{3,2}+1 \\
 x_{1,1}+x_{3,1}&\leq x_{2,1}+1 \\
 x_{1,1}+x_{3,2}&\leq x_{2,1}+1 \\
 x_{1,1}+x_{3,2}&\leq x_{2,2}+1 \\
 x_{1,1}&\leq x_{2,1}+x_{2,2}+x_{3,1} \\
 x_{1,1}&\leq x_{2,1}+x_{2,2}+x_{3,2} \\
 x_{1,1}&\leq x_{2,1}+x_{3,1}+x_{3,2} \\
 \end{align*}
 \begin{align*}
 x_{1,1}&\leq x_{2,2}+x_{3,1}+x_{3,2} \\
 x_{2,1}&\leq x_{1,1}+x_{2,2}+x_{3,2} \\
 x_{1,1}+x_{2,1}&\leq x_{2,2}+x_{3,1}+1 \\
 2 x_{2,1}+x_{2,2}&\leq x_{1,1}+x_{3,1}+1 \\
 2 x_{2,1}+x_{2,2}&\leq x_{1,1}+x_{3,2}+1 \\
 x_{1,1}+2 x_{2,1}+x_{2,2}&\leq x_{3,1}+2 \\
 x_{2,1}+2 x_{2,2}&\leq x_{1,1}+x_{3,2}+1 \\
 x_{1,1}+x_{2,1}+2 x_{2,2}&\leq x_{3,2}+2 \\
 x_{2,1}+x_{2,2}&\leq x_{1,1}+x_{3,1}+x_{3,2} \\
 x_{1,1}+x_{2,1}+x_{2,2}&\leq x_{3,1}+x_{3,2}+1 \\
 x_{1,1}+x_{3,1}+x_{3,2}&\leq x_{2,1}+x_{2,2}+1 \\
 x_{1,1}+x_{2,1}&\leq x_{2,2}+2 x_{3,1}+x_{3,2} \\
 x_{1,1}+x_{3,1}&\leq 2 x_{2,1}+x_{2,2}+x_{3,2} \\
 x_{2,1}+1&\leq x_{1,1}+x_{2,2}+x_{3,1}+2 x_{3,2} \\
 x_{2,2}+1&\leq x_{1,1}+x_{2,1}+2 x_{3,1}+x_{3,2} \\
 x_{3,1}+1&\leq x_{1,1}+x_{2,1}+2 x_{2,2}+x_{3,2} \\
 2 &\leq x_{1,1}+2 x_{2,1}+x_{2,2}+2 x_{3,1}+x_{3,2} \\
 2 x_{2,1}+x_{2,2}&\leq x_{1,1}+2 x_{3,1}+x_{3,2} \\
 x_{2,1}+2 x_{2,2}&\leq x_{1,1}+2 x_{3,1}+x_{3,2} \\
 2 x_{3,1}+x_{3,2}&\leq x_{1,1}+2 x_{2,1}+x_{2,2} \\
 x_{1,1}+1&\leq 2 x_{2,1}+x_{2,2}+2 x_{3,1}+x_{3,2} \\
 x_{1,1}+x_{2,1}+2 x_{2,2}&\leq 2 x_{3,1}+x_{3,2}+1 \\
 x_{1,1}+x_{3,1}+2 x_{3,2}&\leq 2 x_{2,1}+x_{2,2}+1 \\
 \end{align*}
 \end{fleqn}
 \end{multicols}
 \FloatBarrier
 \pagebreak
 
\subsubsection{Orbit Polytopes}
\begin{table}[h]
\begin{equation*}
\begin{array}{|c|l|}
\multicolumn{2}{c}{\textbf{Orbit polytopes for $2 \times 3 \times 3$}} \\ \hline
\text{Representative of Class} & \text{Additional Inequalities} \\ \hline
\ket{000}+\ket{111}+\ket{022}+\ket{122} & \text{None} \\ \hline 
\multirow{3}{*}{$\ket{100}+\ket{010}+\ket{010}+\ket{122}$} & x_{1,1}+2 x_{2,1}+x_{2,2}+x_{3,1}\geq 2 \\
& x_{1,1}+x_{2,1}+2 x_{3,1}+x_{3,2}\geq 2 \\
&x_{1,1}+x_{2,1}+x_{2,2}+x_{3,1}+x_{3,2}\geq 2 \\ \hline
\ket{010}+\ket{100}+\ket{001}+\ket{112}+\ket{121} & x_{1,1}+2 x_{2,1}+x_{2,2}+2 x_{3,1}+x_{3,2}\geq 3 \\ \hline 
\multirow{8}{*}{$\ket{000}+\ket{111}+\ket{022}$} & x_{2,2}\leq x_{3,1} \\
& x_{3,2}\leq x_{2,1} \\
& x_{1,1}+x_{2,1}\geq 1 \\
& x_{1,1}+x_{3,1}\geq 1 \\
& x_{2,1}+x_{2,2}+x_{3,2}\geq 1 \\
& x_{2,2}+x_{3,1}+x_{3,2}\geq 1 \\
& x_{2,2}+1\leq x_{1,1}+x_{3,1}+x_{3,2} \\
& x_{3,2}+1\leq x_{1,1}+x_{2,1}+x_{2,2} \\ \hline
\multirow{13}{*}{$\ket{100}+\ket{010}+\ket{001}+\ket{022}$} &x_{2,2}\leq x_{3,1} \\
 &x_{3,2}\leq x_{2,1} \\
 &x_{1,1}+x_{2,1}\geq 1 \\
 &x_{1,1}+x_{3,1}\geq 1 \\
 &x_{2,1}+x_{2,2}+x_{3,2}\geq 1 \\
 &x_{2,2}+x_{3,1}+x_{3,2}\geq 1 \\
 &x_{1,1}+x_{2,1}+x_{2,2}+x_{3,1}\geq 2 \\
 &x_{1,1}+x_{2,1}+x_{3,1}+x_{3,2}\geq 2 \\
 &2 x_{1,1}+2 x_{2,1}+x_{2,2}+2 x_{3,1}+x_{3,2}\geq 4 \\
 &x_{2,2}+3\leq 2 x_{1,1}+x_{2,1}+2 \left(x_{3,1}+x_{3,2}\right) \\
 &x_{3,2}+3\leq 2 x_{1,1}+2 x_{2,1}+2 x_{2,2}+x_{3,1} \\
 &4 x_{1,1}+4 x_{2,1}+5 x_{2,2}+5 x_{3,1}+x_{3,2}\geq 9 \\
 &4 x_{1,1}+5 x_{2,1}+x_{2,2}+4 x_{3,1}+5 x_{3,2}\geq 9 \\ \hline
\multirow{14}{*}{$\ket{100}+\ket{010}+\ket{112}+\ket{121}$} & x_{2,2}\leq x_{3,1} \\
& x_{3,2}\leq x_{2,1} \\
& x_{1,1}+x_{2,1}\geq 1 \\
& x_{1,1}+x_{3,1}\geq 1 \\
& x_{2,1}+x_{2,2}+x_{3,2}\geq 1 \\
& x_{2,2}+x_{3,1}+x_{3,2}\geq 1 \\
& x_{1,1}+x_{2,1}+x_{2,2}+x_{3,1}\geq 2 \\
& x_{1,1}+x_{2,1}+x_{3,1}+x_{3,2}\geq 2 \\
& 2 x_{1,1}+2 x_{2,1}+x_{2,2}+2 x_{3,1}+x_{3,2}\geq 4 \\
& x_{2,2}+3\leq 2 x_{1,1}+x_{2,1}+2 \left(x_{3,1}+x_{3,2}\right) \\
& x_{3,2}+3\leq 2 x_{1,1}+2 x_{2,1}+2 x_{2,2}+x_{3,1} \\
& 4 x_{1,1}+4 x_{2,1}+5 x_{2,2}+5 x_{3,1}+x_{3,2}\geq 9 \\
& 4 x_{1,1}+5 x_{2,1}+x_{2,2}+4 x_{3,1}+5 x_{3,2}\geq 9 \\ \hline
\end{array}
\end{equation*}
\caption{Entanglement Polytope Inequalities for $2 \times 3 \times 3$}
\end{table}
\FloatBarrier
\pagebreak
\subsection{$2 \times 3 \times 4$}
From now on we will write all inequalities in the form $l(\mathbf{x})+d \leq 0$ which makes them less readable but more presentable. 
\subsubsection{Generic Polytope}
\begin{fleqn}
\begin{multicols}{2}
\begin{align*}
 -x_{1,1}&\leq 0 \\
 -x_{2,2}&\leq 0 \\
 -x_{3,3}&\leq 0 \\
 x_{1,1}-1&\leq 0 \\
 1-2 x_{1,1}&\leq 0 \\
 x_{2,2}-x_{2,1}&\leq 0 \\
 x_{3,2}-x_{3,1}&\leq 0 \\
 x_{3,3}-x_{3,2}&\leq 0 \\
 x_{2,1}+x_{2,2}-1&\leq 0 \\
 -x_{2,1}-2 x_{2,2}+1&\leq 0 \\
 -x_{1,1}+x_{2,1}-x_{3,2}&\leq 0 \\
 x_{2,1}-x_{3,1}-x_{3,2}&\leq 0 \\
 x_{3,1}+x_{3,2}+x_{3,3}-1&\leq 0 \\
 x_{1,1}+x_{2,1}-x_{3,1}-1&\leq 0 \\
 x_{1,1}+x_{2,2}-x_{3,2}-1&\leq 0 \\
 -x_{3,1}-x_{3,2}-2 x_{3,3}+1&\leq 0 \\
 x_{1,1}-x_{2,2}-x_{3,1}-x_{3,2}&\leq 0 \\
 x_{1,1}-x_{2,1}-x_{2,2}-x_{3,3}&\leq 0 \\
 x_{1,1}-x_{2,1}-x_{3,1}-x_{3,3}&\leq 0 \\
 x_{1,1}-x_{2,2}-x_{3,1}-x_{3,3}&\leq 0 \\
 x_{1,1}-x_{2,1}-x_{3,2}-x_{3,3}&\leq 0 \\
 x_{1,1}-x_{3,1}-x_{3,2}-x_{3,3}&\leq 0 \\
 -x_{1,1}-x_{2,2}-x_{3,2}-x_{3,3}+1&\leq 0 \\
 -x_{2,1}-x_{2,2}-x_{3,2}-x_{3,3}+1&\leq 0 \\
 -x_{2,1}-x_{3,1}-x_{3,2}-x_{3,3}+1&\leq 0 \\
 -x_{1,1}+x_{2,1}+x_{2,2}-x_{3,1}-x_{3,2}&\leq 0 \\
 \end{align*}
 \begin{align*}
 -x_{1,1}+x_{2,1}-x_{2,2}-x_{3,2}+x_{3,3}&\leq 0 \\
 x_{1,1}+x_{2,1}+x_{2,2}-x_{3,1}-x_{3,2}-1&\leq 0 \\
 x_{1,1}+x_{2,1}-x_{2,2}-x_{3,1}+x_{3,3}-1&\leq 0 \\
 x_{1,1}+x_{2,2}-2 x_{3,1}-x_{3,2}-x_{3,3}&\leq 0 \\
 x_{1,1}-2 x_{2,1}-x_{2,2}+x_{3,1}-x_{3,2}&\leq 0 \\
 x_{1,1}-2 x_{2,1}-x_{2,2}+x_{3,2}-x_{3,3}&\leq 0 \\
 -x_{1,1}-x_{2,1}+x_{2,2}-2 x_{3,1}-x_{3,2}+1&\leq 0 \\
 -x_{1,1}+x_{2,2}-2 x_{3,1}-x_{3,2}-x_{3,3}+1&\leq 0 \\
 -x_{1,1}-x_{2,1}-2 x_{2,2}+x_{3,1}-x_{3,2}+1&\leq 0 \\
 -x_{1,1}-2 x_{2,1}-x_{2,2}+x_{3,1}-x_{3,3}+1&\leq 0 \\
 -x_{1,1}+2 x_{2,1}+x_{2,2}-x_{3,1}+x_{3,3}-1&\leq 0 \\
 -x_{1,1}+x_{2,1}+2 x_{2,2}-x_{3,1}+x_{3,3}-1&\leq 0 \\
 x_{1,1}+x_{2,1}+2 x_{2,2}-x_{3,1}+x_{3,3}-2&\leq 0 \\
 x_{1,1}-2 x_{2,1}-x_{2,2}-2 x_{3,1}-x_{3,2}+1&\leq 0 \\
 x_{1,1}-x_{2,1}-2 x_{3,1}-2 x_{3,2}-x_{3,3}+1&\leq 0 \\
 x_{1,1}+x_{2,1}-x_{2,2}-2 x_{3,1}-x_{3,2}-x_{3,3}&\leq 0 \\
 x_{1,1}-x_{2,1}-x_{2,2}-2 x_{3,1}-x_{3,2}-x_{3,3}+1&\leq 0 \\
 x_{1,1}-x_{2,1}-x_{2,2}-x_{3,1}-2 x_{3,2}-x_{3,3}+1&\leq 0 \\
 -x_{1,1}+x_{2,1}-x_{2,2}-x_{3,1}-2 x_{3,2}-x_{3,3}+1&\leq 0 \\
 -x_{1,1}+2 x_{2,1}+x_{2,2}-2 x_{3,1}-x_{3,2}-x_{3,3}&\leq 0 \\
 -x_{1,1}+2 x_{2,1}+x_{2,2}-x_{3,1}-2 x_{3,2}-x_{3,3}&\leq 0 \\
 -x_{1,1}+x_{2,1}+2 x_{2,2}-x_{3,1}-2 x_{3,2}-x_{3,3}&\leq 0 \\
 x_{1,1}+2 x_{2,1}+x_{2,2}-2 x_{3,1}-x_{3,2}-x_{3,3}-1&\leq 0 \\
 x_{1,1}+x_{2,1}+2 x_{2,2}-x_{3,1}-2 x_{3,2}-x_{3,3}-1&\leq 0 \\
 -x_{1,1}-2 x_{2,1}-x_{2,2}-3 x_{3,1}-2 x_{3,2}-x_{3,3}+3&\leq 0 \\
 x_{1,1}-2 x_{2,1}-x_{2,2}-3 x_{3,1}-2 x_{3,2}-x_{3,3}+2&\leq 0 \\
\end{align*}
\end{multicols}
\end{fleqn}
\FloatBarrier
\pagebreak
\subsubsection{Orbit Polytopes}
\begin{table}[h]
\renewcommand{\arraystretch}{1.2}
\begin{equation*}
\begin{array}{|c|l|}
\multicolumn{2}{c}{\textbf{Orbit polytopes for $2 \times 3 \times 4$}} \\ \hline
\text{Representative of Class} & \text{Additional Inequalities} \\ \hline
\ket{000}+\ket{101}+\ket{110}+\ket{012}+\ket{023}+\ket{122} & \text{None} \\ \hline 
\multirow{3}{*}{$\ket{000}+\ket{110}+\ket{012}+\ket{123}+\ket{101}$} & -x_{2,1}-x_{2,2}-x_{3,1}+1\leq 0 \\
& -x_{1,1}-x_{2,1}-2 x_{3,1}-x_{3,2}-x_{3,3}+2\leq 0 \\
& -x_{1,1}-2 x_{2,1}-x_{2,2}-2 x_{3,1}-2 x_{3,2}-x_{3,3}+3\leq 0 \\ \hline 
\ket{000}+\ket{101}+\ket{012}+\ket{123} & -x_{2,1}-x_{3,1}-x_{3,2} \leq 0 \\ \hline
\multirow{9}{*}{$\ket{000}+\ket{101}+\ket{012}+\ket{023}+\ket{122}$} &-x_{2,1}-x_{3,1}-x_{3,2}+1\leq 0 \\
& -x_{1,1}-x_{2,1}-x_{3,1}+x_{3,3}+1\leq 0 \\
& -x_{1,1}-x_{2,1}-x_{2,2}-x_{3,1}-x_{3,2}+2\leq 0 \\
& -x_{1,1}-x_{2,1}-x_{3,1}-2 x_{3,2}-x_{3,3}+2\leq 0 \\
& -x_{1,1}-x_{2,2}-2 x_{3,1}-x_{3,2}-x_{3,3}+2\leq 0 \\

& -3 x_{1,1}-6 x_{2,1}-3 x_{2,2}-2 x_{3,1}+x_{3,2}+x_{3,3}+5\leq 0 \\
& -2 x_{1,1}-3 x_{2,1}-2 x_{2,2}-3 x_{3,1}-2 x_{3,2}-x_{3,3}+5\leq 0 \\
& -x_{1,1}-2 x_{2,1}-x_{2,2}-2 x_{3,1}-x_{3,2}-x_{3,3}+3\leq 0 \\
& -3 x_{1,1}-3 x_{2,1}-3 x_{2,2}-x_{3,1}-x_{3,2}+2 x_{3,3}+4\leq 0 \\ \hline
\multirow{16}{*}{$\ket{000}+\ket{101}+\ket{012}+\ket{123}$} & x_{2,2}-x_{3,1}\leq 0 \\
& x_{3,2}-x_{2,1}\leq 0 \\
& x_{3,3}-x_{2,2}\leq 0 \\
& -x_{1,1}-x_{2,1}+1\leq 0 \\
& -x_{1,1}-x_{3,1}+1\leq 0 \\
& -x_{2,1}-x_{2,2}-x_{3,2}+1\leq 0 \\
& -3 x_{1,1}-2 x_{3,1}+x_{3,2}+x_{3,3}+2\leq 0 \\
& 3 x_{2,2}-2 x_{3,1}+x_{3,2}+x_{3,3}-1\leq 0 \\
& -x_{1,1}+x_{2,2}-x_{3,1}-x_{3,2}+1\leq 0 \\
& x_{2,1}+x_{2,2}-x_{3,1}-x_{3,2}-x_{3,3}\leq 0 \\
& -3 x_{1,1}-2 x_{2,1}-3 x_{3,1}-2 x_{3,2}-2 x_{3,3}+5\leq 0 \\
& -x_{1,1}-x_{2,1}-x_{2,2}-x_{3,1}-x_{3,3}+2\leq 0 \\
& -x_{1,1}-x_{2,1}-x_{3,1}-x_{3,2}-x_{3,3}+2\leq 0 \\
& -3 x_{1,1}+3 x_{2,2}-x_{3,1}-x_{3,2}+2 x_{3,3}+1\leq 0 \\
& -4 x_{1,1}+x_{2,1}+5 x_{2,2}-5 x_{3,1}-5 x_{3,2}-x_{3,3}+4\leq 0 \\
& -4 x_{1,1}+x_{2,1}+5 x_{2,2}-5 x_{3,1}-5 x_{3,2}-x_{3,3}+4\leq 0 \\
& -3 x_{1,1}-3 x_{2,1}-3 x_{2,2}-x_{3,1}+2 x_{3,2}-x_{3,3}+4\leq 0 \\ \hline
\end{array}
\end{equation*}
\caption{Entanglement Polytope Inequalities for $2 \times 3 \times 4$\label{table:234orbitpol}}
\end{table}
\FloatBarrier
\pagebreak
\subsection{$2 \times 3 \times 5$}
\subsubsection{Generic Polytope}
\begin{fleqn}
\begin{multicols}{2}
\begin{align*}
 x_{1,1}-1&\leq 0 \\
 -x_{1,1}&\leq 0 \\
 -x_{2,2}&\leq 0 \\
 -x_{3,4}&\leq 0 \\
 1-2 x_{1,1}&\leq 0 \\
 x_{2,1}+x_{2,2}-1&\leq 0 \\
 x_{2,2}-x_{2,1}&\leq 0 \\
 x_{3,2}-x_{3,1}&\leq 0 \\
 x_{3,3}-x_{3,2}&\leq 0 \\
 x_{3,4}-x_{3,3}&\leq 0 \\
 -x_{2,1}-2 x_{2,2}+1&\leq 0 \\
 x_{2,1}-x_{3,1}-x_{3,2}&\leq 0 \\
 x_{3,1}+x_{3,2}+x_{3,3}+x_{3,4}-1&\leq 0 \\
 x_{1,1}+x_{2,2}-x_{3,1}+x_{3,4}-1&\leq 0 \\
 x_{1,1}-x_{2,2}-x_{3,1}-x_{3,3}&\leq 0 \\
 x_{1,1}-x_{2,1}-x_{3,2}-x_{3,3}&\leq 0 \\
 x_{1,1}-x_{3,1}-x_{3,2}-x_{3,3}&\leq 0 \\
 x_{1,1}+x_{2,1}+x_{2,2}-x_{3,1}-x_{3,2}-1&\leq 0 \\
 -x_{1,1}-x_{2,2}-x_{3,2}-x_{3,3}+1&\leq 0 \\
 -x_{2,1}-x_{2,2}-x_{3,2}-x_{3,3}+1&\leq 0 \\
 -x_{2,1}-x_{3,1}-x_{3,2}-x_{3,3}+1&\leq 0 \\
 -x_{3,1}-x_{3,2}-x_{3,3}-2 x_{3,4}+1&\leq 0 \\
 x_{1,1}-x_{2,1}-x_{2,2}-x_{3,2}+x_{3,4}&\leq 0 \\
 x_{1,1}-x_{2,1}-x_{3,1}-x_{3,2}+x_{3,4}&\leq 0 \\
 -x_{1,1}+x_{2,2}-2 x_{3,1}-x_{3,2}-x_{3,3}+1&\leq 0 \\
 \end{align*}
 \begin{align*}
 x_{1,1}+x_{2,1}-x_{2,2}-2 x_{3,1}-x_{3,2}-x_{3,4}&\leq 0 \\
 x_{1,1}+x_{2,2}-x_{3,1}-2 x_{3,2}-x_{3,3}-x_{3,4}&\leq 0 \\
 x_{1,1}+x_{2,1}-2 x_{3,1}-x_{3,2}-x_{3,3}-x_{3,4}&\leq 0 \\
 x_{2,1}+x_{2,2}-x_{3,1}-x_{3,2}-x_{3,3}-x_{3,4}&\leq 0 \\
 x_{1,1}-2 x_{2,1}-x_{2,2}-x_{3,1}+x_{3,3}+2 x_{3,4}&\leq 0 \\
 x_{1,1}+x_{2,1}+2 x_{2,2}-x_{3,1}-2 x_{3,2}-x_{3,3}-1&\leq 0 \\
 x_{1,1}+2 x_{2,1}+x_{2,2}-2 x_{3,1}-x_{3,2}-x_{3,3}-1&\leq 0 \\
 x_{1,1}+x_{2,1}+2 x_{2,2}-2 x_{3,1}-x_{3,2}-x_{3,4}-1&\leq 0 \\
 -x_{1,1}+2 x_{2,1}+x_{2,2}-x_{3,1}-2 x_{3,2}-x_{3,3}&\leq 0 \\
 -x_{1,1}+2 x_{2,1}+x_{2,2}-2 x_{3,1}-x_{3,2}-x_{3,4}&\leq 0 \\
 x_{1,1}-x_{2,1}-x_{2,2}-2 x_{3,1}-x_{3,2}-x_{3,3}+1&\leq 0 \\
 -x_{1,1}+x_{2,1}-x_{2,2}-x_{3,1}-2 x_{3,2}-x_{3,4}+1&\leq 0 \\
 x_{1,1}-2 x_{2,1}-x_{2,2}-x_{3,1}-2 x_{3,3}-x_{3,4}+1&\leq 0 \\
 x_{1,1}-x_{2,1}-2 x_{3,1}-x_{3,2}-2 x_{3,3}-x_{3,4}+1&\leq 0 \\
 x_{1,1}-2 x_{2,1}-x_{2,2}-2 x_{3,2}-x_{3,3}-x_{3,4}+1&\leq 0 \\
 x_{1,1}-x_{2,2}-2 x_{3,1}-2 x_{3,2}-x_{3,3}-x_{3,4}+1&\leq 0 \\
 -x_{1,1}+x_{2,1}-x_{3,1}-2 x_{3,2}-x_{3,3}-x_{3,4}+1&\leq 0 \\
 -x_{1,1}-2 x_{2,1}-x_{2,2}-2 x_{3,1}-x_{3,2}+x_{3,4}+2&\leq 0 \\
 -x_{1,1}-2 x_{2,1}-x_{2,2}-x_{3,2}-2 x_{3,3}-x_{3,4}+2&\leq 0 \\
 -x_{1,1}-x_{2,1}-2 x_{2,2}-2 x_{3,2}-x_{3,3}-x_{3,4}+2&\leq 0 \\
 x_{1,1}+x_{2,1}-x_{2,2}-3 x_{3,1}-2 x_{3,2}-2 x_{3,3}-x_{3,4}+1&\leq 0 \\
 -x_{1,1}+x_{2,1}+x_{2,2}-2 x_{3,1}-2 x_{3,2}-x_{3,3}-x_{3,4}+1&\leq 0 \\
 -x_{1,1}+x_{2,1}+2 x_{2,2}-3 x_{3,1}-2 x_{3,2}-x_{3,3}-2 x_{3,4}+1&\leq 0 \\
 -x_{1,1}+x_{2,1}-x_{2,2}-2 x_{3,1}-3 x_{3,2}-2 x_{3,3}-x_{3,4}+2&\leq 0 \\
 -x_{1,1}+x_{2,1}+2 x_{2,2}-2 x_{3,1}-3 x_{3,2}-2 x_{3,3}-x_{3,4}+1&\leq 0 \\
 -x_{1,1}+2 x_{2,1}+x_{2,2}-3 x_{3,1}-2 x_{3,2}-2 x_{3,3}-x_{3,4}+1&\leq 0 \\
 x_{1,1}-x_{2,1}-x_{2,2}-x_{3,1}-x_{3,2}-2 x_{3,3}-x_{3,4}+1&\leq 0 \\
 x_{1,1}-2 x_{2,1}-x_{2,2}-3 x_{3,1}-2 x_{3,2}-x_{3,3}-x_{3,4}+2&\leq 0 \\
 -x_{1,1}-x_{2,1}+x_{2,2}-3 x_{3,1}-2 x_{3,2}-x_{3,3}-x_{3,4}+2&\leq 0 \\
\end{align*}
\end{multicols}
\end{fleqn}
\pagebreak
\subsubsection{Orbit Polytopes}
\begin{table}[h]
\renewcommand{\arraystretch}{1.2}
\begin{equation*}
\begin{array}{|c|l|}

\multicolumn{2}{c}{\textbf{Orbit polytopes for $2 \times 3 \times 5$}} \\ \hline
\text{Representative of Class} & \text{Additional Inequalities} \\ \hline
\ket{000}+\ket{011}+\ket{102}+\ket{113}+\ket{121}+\ket{024} & \text{None} \\ \hline 
\multirow{15}{*}{$\ket{000}+\ket{011}+\ket{102}+\ket{113}+\ket{024}$} & -x_{2,1}-x_{2,2}-x_{3,1}+1\leq 0 \\
& -x_{1,1}-x_{3,1}-x_{3,2}+1\leq 0 \\
& -x_{1,1}-x_{2,1}-x_{3,1}+x_{3,4}+1\leq 0 \\
& -x_{2,2}-x_{3,1}-x_{3,2}-x_{3,3}+1\leq 0 \\
& -x_{1,1}+x_{2,2}-x_{3,1}-2 x_{3,2}-x_{3,3}+1\leq 0 \\
& -x_{1,1}+x_{2,1}-2 x_{3,1}-x_{3,2}-x_{3,3}+1\leq 0 \\
& -x_{1,1}+x_{2,2}-2 x_{3,1}-x_{3,2}-x_{3,4}+1\leq 0 \\
& -x_{1,1}-2 x_{2,1}-x_{2,2}-x_{3,1}-x_{3,2}+x_{3,4}+2\leq 0 \\
&-x_{1,1}-x_{2,1}-2 x_{2,2}-2 x_{3,1}-2 x_{3,2}-2 x_{3,3}+3\leq 0 \\
& -x_{1,1}-2 x_{2,1}-2 x_{2,2}-2 x_{3,1}-x_{3,2}-x_{3,3}+3\leq 0 \\
& -x_{1,1}-x_{2,1}-x_{2,2}-x_{3,1}-x_{3,2}-x_{3,3}+2\leq 0 \\
& -x_{1,1}-x_{2,1}-x_{3,1}-2 x_{3,2}-x_{3,3}-x_{3,4}+2\leq 0 \\
& -x_{1,1}-x_{2,2}-2 x_{3,1}-x_{3,2}-x_{3,3}-x_{3,4}+2\leq 0 \\
& -x_{1,1}-2 x_{2,1}-x_{2,2}-2 x_{3,1}-x_{3,2}-2 x_{3,3}-x_{3,4}+3\leq 0 \\
& -x_{1,1}-x_{2,1}-2 x_{2,2}-2 x_{3,1}-2 x_{3,2}-x_{3,3}-x_{3,4}+3\leq 0 \\ \hline
\end{array}
\end{equation*}
\caption{Entanglement Polytope Inequalities for $2 \times 3 \times 5$}
\end{table}
\pagebreak
\subsection{$2 \times 3 \times 6$}

\subsubsection{Generic Polytope}
\begin{fleqn}
\begin{multicols}{2}
\begin{align*}
 x_{1,1}-1&\leq 0 \\
 -x_{1,1}&\leq 0 \\
 -x_{2,2}&\leq 0 \\
 -x_{3,5}&\leq 0 \\
 1-2 x_{1,1}&\leq 0 \\
 x_{2,1}+x_{2,2}-1&\leq 0 \\
 x_{2,2}-x_{2,1}&\leq 0 \\
 x_{3,2}-x_{3,1}&\leq 0 \\
 x_{3,3}-x_{3,2}&\leq 0 \\
 x_{3,4}-x_{3,3}&\leq 0 \\
 x_{3,5}-x_{3,4}&\leq 0 \\
 -x_{2,1}-2 x_{2,2}+1&\leq 0 \\
 -x_{2,1}+x_{3,4}+x_{3,5}&\leq 0 \\
 x_{2,1}-x_{3,1}-x_{3,2}&\leq 0 \\
 x_{3,1}+x_{3,2}+x_{3,3}+x_{3,4}+x_{3,5}-1&\leq 0 \\
 x_{1,1}-x_{3,1}-x_{3,2}-x_{3,3}&\leq 0 \\
 -x_{1,1}+x_{2,2}-x_{3,1}+x_{3,4}+x_{3,5}&\leq 0 \\
 -x_{2,1}-x_{2,2}-x_{3,2}-x_{3,3}+1&\leq 0 \\
 x_{1,1}-x_{2,1}-x_{3,1}-x_{3,2}+x_{3,4}&\leq 0 \\
 x_{1,1}-x_{2,2}-x_{3,1}-x_{3,2}+x_{3,5}&\leq 0 \\
 x_{1,1}-x_{2,1}-x_{3,1}-x_{3,3}+x_{3,5}&\leq 0 \\
 -x_{1,1}+x_{2,1}+x_{2,2}-x_{3,1}-x_{3,2}+x_{3,5}&\leq 0 \\
 x_{1,1}-x_{2,1}-x_{2,2}-x_{3,1}+x_{3,4}+x_{3,5}&\leq 0 \\
 -x_{3,1}-x_{3,2}-x_{3,3}-x_{3,4}-2 x_{3,5}+1&\leq 0 \\
 x_{1,1}+x_{2,2}-x_{3,1}-2 x_{3,2}-x_{3,3}-x_{3,4}&\leq 0 \\
 x_{1,1}+x_{2,1}-2 x_{3,1}-x_{3,2}-x_{3,3}-x_{3,4}&\leq 0 \\
 x_{2,1}+x_{2,2}-x_{3,1}-x_{3,2}-x_{3,3}-x_{3,4}&\leq 0 \\
 x_{1,1}+x_{2,2}-2 x_{3,1}-x_{3,2}-x_{3,3}-x_{3,5}&\leq 0 \\
 \end{align*}
 \begin{align*}
 -x_{1,1}+x_{2,1}-x_{3,1}-2 x_{3,2}-x_{3,3}-x_{3,4}+1&\leq 0 \\
 x_{1,1}+x_{2,1}-x_{2,2}-2 x_{3,1}-x_{3,2}-x_{3,3}+x_{3,5}&\leq 0 \\
 x_{1,1}-2 x_{2,1}-x_{2,2}-x_{3,1}+x_{3,3}+x_{3,4}+2 x_{3,5}&\leq 0 \\
 -x_{1,1}-x_{2,1}+x_{2,2}-x_{3,1}+x_{3,3}+x_{3,4}+2 x_{3,5}&\leq 0 \\
 -x_{1,1}+x_{2,1}+2 x_{2,2}-x_{3,1}-2 x_{3,2}-x_{3,3}+x_{3,5}&\leq 0 \\
 -x_{1,1}+2 x_{2,1}+x_{2,2}-2 x_{3,1}-x_{3,2}-x_{3,3}+x_{3,5}&\leq 0 \\
 -x_{1,1}+x_{2,1}+2 x_{2,2}-2 x_{3,1}-x_{3,2}-x_{3,4}+x_{3,5}&\leq 0 \\
 x_{1,1}-2 x_{2,1}-x_{2,2}-2 x_{3,1}-x_{3,2}+x_{3,4}-x_{3,5}+1&\leq 0 \\
 -x_{1,1}+x_{2,1}-x_{2,2}-x_{3,1}-2 x_{3,2}-x_{3,3}+x_{3,5}+1&\leq 0 \\
 -x_{1,1}-2 x_{2,1}-x_{2,2}-x_{3,1}+x_{3,3}+2 x_{3,4}+x_{3,5}+1&\leq 0 \\
 x_{1,1}-x_{2,1}-x_{2,2}-x_{3,1}-x_{3,2}-2 x_{3,3}-x_{3,4}+1&\leq 0 \\
 x_{1,1}-x_{2,1}-x_{2,2}-x_{3,1}-2 x_{3,2}-x_{3,3}-x_{3,5}+1&\leq 0 \\
 x_{1,1}-x_{2,1}-x_{3,1}-2 x_{3,2}-2 x_{3,3}-x_{3,4}-x_{3,5}+1&\leq 0 \\
 x_{1,1}-x_{2,2}-2 x_{3,1}-x_{3,2}-2 x_{3,3}-x_{3,4}-x_{3,5}+1&\leq 0 \\
 x_{1,1}+x_{2,1}+x_{2,2}-2 x_{3,1}-2 x_{3,2}-x_{3,3}-x_{3,4}-x_{3,5}&\leq 0 \\
 -x_{1,1}-x_{2,2}-x_{3,1}-2 x_{3,2}-2 x_{3,3}-x_{3,4}-x_{3,5}+2&\leq 0 \\
 x_{1,1}+x_{2,1}+2 x_{2,2}-3 x_{3,1}-2 x_{3,2}-x_{3,3}-2 x_{3,4}-x_{3,5}&\leq 0 \\
 x_{1,1}+x_{2,1}+2 x_{2,2}-2 x_{3,1}-3 x_{3,2}-2 x_{3,3}-x_{3,4}-x_{3,5}&\leq 0 \\
 x_{1,1}+2 x_{2,1}+x_{2,2}-3 x_{3,1}-2 x_{3,2}-2 x_{3,3}-x_{3,4}-x_{3,5}&\leq 0 \\
 x_{1,1}+x_{2,1}-x_{2,2}-3 x_{3,1}-2 x_{3,2}-x_{3,3}-2 x_{3,4}-x_{3,5}+1&\leq 0 \\
 x_{1,1}-2 x_{2,1}-x_{2,2}-2 x_{3,1}-x_{3,2}-3 x_{3,3}-2 x_{3,4}-x_{3,5}+2&\leq 0 \\
 x_{1,1}-2 x_{2,1}-x_{2,2}-x_{3,1}-3 x_{3,2}-2 x_{3,3}-2 x_{3,4}-x_{3,5}+2&\leq 0 \\
 -x_{1,1}+x_{2,1}-x_{2,2}-2 x_{3,1}-3 x_{3,2}-x_{3,3}-2 x_{3,4}-x_{3,5}+2&\leq 0 \\
 -x_{1,1}+2 x_{2,1}+x_{2,2}-3 x_{3,1}-2 x_{3,2}-x_{3,3}-2 x_{3,4}-x_{3,5}+1&\leq 0 \\
 -x_{1,1}+2 x_{2,1}+x_{2,2}-2 x_{3,1}-3 x_{3,2}-2 x_{3,3}-x_{3,4}-x_{3,5}+1&\leq 0 \\
 -x_{1,1}-2 x_{2,1}-x_{2,2}-x_{3,1}-2 x_{3,2}-3 x_{3,3}-2 x_{3,4}-x_{3,5}+3&\leq 0 \\
 -x_{1,1}-x_{2,1}-2 x_{2,2}-x_{3,1}-3 x_{3,2}-2 x_{3,3}-2 x_{3,4}-x_{3,5}+3&\leq 0 \\
\end{align*}
\end{multicols}
\end{fleqn}
\subsubsection{Orbit Polytopes}
\begin{table}[h]
\renewcommand{\arraystretch}{1.2}
\begin{equation*}
\begin{array}{|c|l|}
\multicolumn{2}{c}{\textbf{Orbit polytopes for $2 \times 3 \times 6$}} \\ \hline
\text{Representative of Class} & \text{Additional Inequalities} \\ \hline
\ket{000}+\ket{011}+\ket{022}+\ket{103}+\ket{114}+\ket{125} & \text{None} \\ \hline
\end{array}
\end{equation*}
\caption{Entanglement Polytope Inequalities for $2 \times 3 \times 6$}
\end{table}
 \newpage
\section{$2\times 4 \times N$}
The results arrived at for this cases are (as of yet) very incomplete, but we include the ones we obtained. We also no longer include the inequalities for the generic polytopes as they take more and more space (the inequalities for the $2 \times 4 \times 8$ case provided in \cite{Klyachko2007} extend over 7 pages). They can be retrieved from the 
 Mathematica package. 
\subsection{$2 \times 4 \times 4$}
We use notation from \cite{Chen2009}: 
\begin{align*}
\ket{\phi_0} &= \ket{000}+\ket{111}+\ket{022} \\
\ket{\phi_1} &= \ket{000}+\ket{111}+\ket{022}+\ket{122} \\
\ket{\phi_2} &= \ket{010}+\ket{001}+\ket{112}+\ket{121} \\
\ket{\phi_3} &= \ket{100}+\ket{010}+\ket{001}+\ket{112}+\ket{121} \\
\ket{\phi_4} &= \ket{100}+\ket{010}+\ket{001}+\ket{022} \\
\ket{\phi_5} &= \ket{100}+\ket{010}+\ket{001}+\ket{122} \\
\ket{\phi_6} &= \ket{000}+\ket{011}+\ket{110}+\ket{121} \\
\end{align*}
\subsubsection{Closest Point Inequalities}
\FloatBarrier
\begin{table}[h]
\renewcommand{\arraystretch}{1.2}
\begin{equation*}
\begin{array}{|c|c|}
\multicolumn{2}{c}{\textbf{Closest Point Inequalities for } 2 \times 4 \times 4} \\ \hline 
\text{Representative of orbit} & \text{Inequality} \\ \hline 
\ket{133}+\ket{\varphi_{0}} &\text{n/a} \\ 
\ket{033}+\ket{\varphi_{0}} & 2 x_{1,1}+x_{2,1}+x_{3,1}\geq 2 \\
\ket{133}+\ket{\varphi_{1}} & \text{n/a}\\
\ket{033}+x\ket{133}+\ket{\varphi_{0}} & \text{n/a} \\
\ket{133}+ \ket{\varphi_{2}} & 3 x_{1,1}+8 x_{2,1}+4 x_{2,2}+3 x_{2,3}+8x_{3,1}+4 x_{3,2}+3 x_{3,3}\geq 11 \\
\ket{133} + \ket{\varphi_{3}} & 3 x_{1,1}+8 x_{2,1}+4 x_{2,2}+3 x_{2,3}+8x_{3,1}+4 x_{3,2}+3 x_{3,3}\geq 11 \\
\ket{033}+ \ket{\varphi_{3}} & 3x_{1,1}+8x_{2,1}+4x_{2,2}+3 x_{2,3}+8 x_{3,1}+4 x_{3,2}+3 x_{3,3}\geq 11\\
\ket{133}+ \ket{\varphi_{4}} & 2 x_{1,1}+2x_{2,1}+2 x_{2,2}+x_{2,3}+2 x_{3,1}+2 x_{3,2}+x_{3,3}\geq 4 \\
\ket{033}+ \ket{\varphi_4 }& 2 x_{1,1}+2 x_{2,1}+x_{2,2}+x_{2,3}+2 x_{3,1}+x_{3,2}+x_{3,3}\geq 4 \\
\ket{133} + \ket{\varphi_5} & \text{n/a} \\
\ket{133} + \ket{033} + \ket{\varphi_5} & \text{n/a} \\
\ket{133}+\ket{032}+\ket{\varphi_2} & x_{2,1}+x_{2,2}+x_{3,1}\geq 1 \\
\ket{133}+\ket{032}+\ket{\varphi_3} &x_{2,1}+x_{2,2}+x_{3,1}\geq 1\\
\ket{033}+\ket{132}+\ket{\varphi_4} &x_{1,1}+x_{2,1}+x_{2,2}+x_{3,1}+x_{3,2}\geq 2 \\
\ket{133}+\ket{032}+\ket{\varphi_5} & \text{n/a}\\
\ket{033}+\ket{132}+\ket{\varphi_6} & x_{2,1}+x_{3,1}+x_{3,2}\geq 1 \\ \hline

\end{array}
\end{equation*}
\caption{Inequalities calculated using the closest point method for $2 \times 4 \times 4$ systems}
\end{table}
\FloatBarrier
\subsubsection{Orbit polytopes}
\begin{table}[h]
\renewcommand{\arraystretch}{1.2}
\begin{equation*}
\begin{array}{|c|l|}
\multicolumn{2}{c}{\textbf{Some Orbit polytopes for $2 \times 4 \times 4$}} \\ \hline
\text{Representative of Class} & \text{Additional Inequalities} \\ \hline
\multirow{5}{*}{$\ket{133}+\ket{\varphi_{0}}$} & x_{2,2}+x_{2,3}-x_{3,1}-x_{3,2}\leq 0 \\
& -x_{2,1}-x_{2,2}+x_{3,2}+x_{3,3}\leq 0 \\
& -x_{2,2}-x_{2,3}-x_{3,1}-x_{3,2}+1\leq 0 \\
& -x_{2,1}-x_{2,2}-x_{3,2}-x_{3,3}+1\leq 0 \\ \hline
\ket{133} + \ket{033} + \ket{\varphi_5}  & \text{None} \\ \hline 
\ket{133}+\ket{032}+\ket{\varphi_5}    & \text{None} \\ \hline 
\end{array}
\end{equation*}
\caption{Entanglement Polytope Inequalities for the three entanglement classes in $2 \times 4 \times 4$}
\end{table}
\newpage
\subsection{$2 \times 4 \times 5$}
For a definition of the representatives of the orbits in $2 \times 4 \times 5$ see Appendix B. 
\subsubsection{Closest Point Inequalities}
\FloatBarrier
\begin{table}[h]
\renewcommand{\arraystretch}{1.2}
\begin{equation*}
\begin{array}{|c|c|}
\multicolumn{2}{c}{\text{Closest Point Inequalities for } 2 \times 4 \times 5} \\ \hline 
\text{Representative of orbit} & \text{Inequality} \\ \hline 
 \ket{\Gamma_0} & 6 x_{1,1}+7 x_{2,1}+3 x_{2,2}+10 x_{3,1}+7 x_{3,2}+7 x_{3,3}+6
  x_{3,4}\geq 13 \\
 \ket{\Gamma_1} & \text{n/a} \\
 \ket{\Gamma_2} & 9 x_{1,1}+5 x_{2,1}+9 x_{3,1}+5 x_{3,2}+5 x_{3,3}+5 x_{3,4}\geq 14
  \\
 \ket{\Gamma_3} & \text{n/a} \\
 \ket{\Gamma_4} & x_{2,1}+x_{2,2}+x_{3,1}+x_{3,2}\geq 1 \\
 \ket{\Gamma_5} & \text{n/a} \\

 \ket{\Gamma_6} & \text{n/a} \\
 \ket{\Gamma_8} & 7 x_{1,1}+24 x_{2,1}+17 x_{2,2}+10 x_{2,3}+24 x_{3,1}+21 x_{3,2}+14
  x_{3,3}+7 x_{3,4}\geq 31 \\
 \ket{\Gamma_9} & 10 x_{1,1}+21 x_{2,1}+11 x_{2,2}+10 x_{2,3}+21 x_{3,1}+20
  x_{3,2}+11 x_{3,3}+10 x_{3,4}\geq 31 \\

 \ket{\Gamma_{10}} & \text{n/a} \\

 \ket{\Gamma_{12}} & \text{n/a} \\
 \ket{\Gamma_{14}} & 3 x_{1,1}+3 x_{2,1}+2 x_{2,2}+2 x_{2,3}+5 x_{3,1}+3 x_{3,2}+3
  x_{3,3}\geq 8 \\ \hline
\end{array}
\end{equation*}
\caption{Inequalities calculated using the closest point method for $2 \times 4 \times 5$ systems}
\end{table}
\newpage
\FloatBarrier
\subsubsection{Orbit Polytopes}
\begin{table}[h]
\renewcommand{\arraystretch}{1.2}
\begin{equation*}
\begin{array}{|c|l|}
\multicolumn{2}{c}{\textbf{Some Orbit polytopes for $2 \times 4 \times 5$}} \\ \hline
\text{Representative of Class} & \text{Additional Inequalities} \\ \hline
\multirow{5}{*}{$\ket{\Gamma_1}$} &  -x_{2,1}-x_{3,1}-x_{3,2}-x_{3,3}+1\leq 0 \\ \hline
& x_{1,1}+x_{2,2}+2 x_{2,3}-3 x_{3,1}-2 x_{3,2}-x_{3,3}-x_{3,4}\leq 0 \\
& x_{1,1}-x_{2,1}+x_{2,2}-x_{2,3}-3 x_{3,1}-2 x_{3,2}-x_{3,3}-x_{3,4}+1\leq 0 \\
& x_{1,1}-2 x_{2,1}-2 x_{2,2}-x_{2,3}-2 x_{3,1}-x_{3,2}-3 x_{3,3}-x_{3,4}+2\leq 0 \\
& x_{1,1}-2 x_{2,1}-2 x_{2,2}-x_{2,3}-x_{3,1}-3 x_{3,2}-2 x_{3,3}-x_{3,4}+2\leq 0 \\
\ket{\Gamma_{12}} & \text{None} \\ \hline 
\ket{\Gamma_{13}} & \text{None} \\ \hline 
\end{array}
\end{equation*}
\caption{Entanglement Polytope Inequalities for the $\ket{\Gamma_1},\ket{\Gamma_{12}}$ and  $\ket{\Gamma_{13}}$ classes in $2 \times 4 \times 5$}
\end{table}
\begin{table}[h]
\renewcommand{\arraystretch}{1.1}
\begin{equation*}
\begin{array}{c}
\textbf{Orbit polytope of $\ket{\Gamma_0}$ for $2 \times 4 \times 5$} \\ \hline
\text{Additional Inequalities} \\ \hline
 x_{2,3}-x_{3,1} \leq 0  \\     x_{3,3}-x_{2,1}\leq 0 \\
 x_{3,4}-x_{2,2}\leq 0   \\    x_{3,4}-x_{2,2}\leq 0 \\
 -x_{1,1}-x_{2,2}-x_{3,1}+1\leq 0   \\   
 -x_{1,1}-x_{2,1}-x_{3,2}+1\leq 0 \\
 -x_{1,1}-x_{3,1}-x_{3,2}+1\leq 0   \\   
 -x_{1,1}+x_{2,2}+x_{2,3}-x_{3,1}+x_{3,4}\leq 0 \\
 -x_{2,1}-x_{2,2}-x_{2,3}-x_{3,2}+1\leq 0   \\   
 -x_{2,1}-x_{2,3}-x_{3,1}-x_{3,2}+1\leq 0 \\
 -x_{2,1}-x_{2,2}-x_{3,1}-x_{3,3}+1\leq 0   \\   
 -x_{2,1}-x_{3,1}-x_{3,2}-x_{3,3}+1\leq 0 \\
 x_{2,1}+x_{2,3}-x_{3,1}-x_{3,2}-x_{3,3}\leq 0   \\   
 x_{2,1}+x_{2,3}-x_{3,1}-x_{3,2}-x_{3,3}\leq 0 \\
 -x_{1,1}-x_{2,1}+x_{2,3}-x_{3,1}-x_{3,3}+1\leq 0   \\   
 -x_{1,1}+x_{2,2}-x_{3,1}-2 x_{3,2}-x_{3,3}+1\leq 0 \\
 -x_{1,1}+x_{2,1}-2 x_{3,1}-x_{3,2}-x_{3,3}+1\leq 0   \\   
 -x_{1,1}+x_{2,1}-2 x_{3,1}-x_{3,2}-x_{3,3}+1\leq 0 \\
 -x_{1,1}+x_{2,2}-2 x_{3,1}-x_{3,2}-x_{3,3}+1\leq 0   \\   
 -x_{1,1}+x_{2,2}-2 x_{3,1}-x_{3,2}-x_{3,4}+1\leq 0  \\
 -x_{1,1}-x_{2,2}-x_{2,3}-x_{3,1}+x_{3,4}+1\leq 0   \\   
 -9 x_{2,1}-4 x_{2,2}-9 x_{3,1}-5 x_{3,2}-9 x_{3,3}+9\leq 0 \\
 -x_{2,2}-x_{2,3}-x_{3,1}-x_{3,2}-x_{3,4}+1\leq 0   \\   
 -x_{2,2}-x_{2,3}-x_{3,1}-x_{3,2}-x_{3,4}+1\leq 0 \\
 -x_{2,1}-x_{2,2}-x_{3,2}-x_{3,3}-x_{3,4}+1\leq 0   \\   
 -x_{2,3}-x_{3,1}-x_{3,2}-x_{3,3}-x_{3,4}+1\leq 0 \\
 -3 x_{1,1}-4 x_{2,2}-3 x_{2,3}-3 x_{3,1}+4 x_{3,4}+3\leq 0   \\   
 -x_{1,1}+x_{2,2}+2 x_{2,3}-x_{3,1}-x_{3,2}+x_{3,4}\leq 0 \\
 -x_{1,1}-x_{2,1}-2 x_{2,2}-x_{2,3}-x_{3,1}-x_{3,3}+2\leq 0   \\   
 -x_{1,1}-2 x_{2,1}-x_{2,2}-x_{2,3}-x_{3,2}-x_{3,3}+2\leq 0 \\
 -x_{1,1}-x_{2,1}-x_{2,3}-2 x_{3,1}-x_{3,2}-x_{3,3}+2\leq 0   \\   
 -x_{1,1}-2 x_{2,1}-x_{2,2}-x_{2,3}-x_{3,1}-x_{3,4}+2\leq 0 \\
 -x_{1,1}-x_{2,1}-2 x_{3,1}-x_{3,2}-x_{3,3}-x_{3,4}+2\leq 0   \\   
 -x_{1,1}+x_{2,2}+x_{2,3}-x_{3,1}-2 x_{3,2}-x_{3,3}-x_{3,4}+1\leq 0 \\
 -x_{1,1}+x_{2,1}+x_{2,3}-2 x_{3,1}-x_{3,2}-x_{3,3}-x_{3,4}+1\leq 0   \\   
 -4 x_{1,1}+x_{2,2}+5 x_{2,3}-5 x_{3,1}-4 x_{3,2}-5 x_{3,3}+x_{3,4}+4\leq 0 \\
 -x_{1,1}+x_{2,2}+2 x_{2,3}-2 x_{3,1}-x_{3,2}-2 x_{3,3}-x_{3,4}+1\leq 0   \\   
 -x_{1,1}-x_{2,1}+x_{2,3}-2 x_{3,1}-2 x_{3,2}-x_{3,3}-x_{3,4}+2\leq 0 \\
 -x_{1,1}+x_{2,1}+2 x_{2,3}-2 x_{3,1}-2 x_{3,2}-x_{3,3}-x_{3,4}+1\leq 0   \\   
 -x_{1,1}-x_{2,2}-6 x_{2,3}-7 x_{3,1}-8 x_{3,2}-7 x_{3,3}-6 x_{3,4}+7\leq 0 \\
 -x_{1,1}-x_{2,1}-x_{2,2}-x_{2,3}-x_{3,1}-x_{3,2}-x_{3,4}+2\leq 0   \\   
 -x_{1,1}-x_{2,1}-x_{2,2}-x_{3,1}-x_{3,2}-x_{3,3}-x_{3,4}+2\leq 0 \\
 -2 x_{1,1}-4 x_{2,1}-2 x_{2,2}-2 x_{2,3}+x_{3,1}-x_{3,2}-x_{3,3}+x_{3,4}+3\leq 0   \\   
 -4 x_{1,1}-8 x_{2,1}-8 x_{2,2}-4 x_{2,3}-3 x_{3,1}+x_{3,2}+5 x_{3,3}-3 x_{3,4}+7\leq 0 \\
 -3 x_{1,1}+x_{2,1}-2 x_{2,2}-2 x_{2,3}-4 x_{3,1}-x_{3,2}-x_{3,3}+2 x_{3,4}+3\leq 0   \\   
 -2 x_{1,1}-2 x_{2,1}-2 x_{2,2}-x_{2,3}-2 x_{3,1}-2 x_{3,2}-x_{3,3}-2 x_{3,4}+4\leq 0 \\
 -x_{1,1}-2 x_{2,1}-x_{2,2}-x_{2,3}-2 x_{3,1}-x_{3,2}-2 x_{3,3}-x_{3,4}+3\leq 0 \\  \hline
\end{array}
\end{equation*}
\caption{Entanglement Polytope Inequalities for the $\ket{\Gamma_0}$ orbit $2 \times 4 \times 5$}
\end{table}
\FloatBarrier
\subsection{$2 \times 4 \times 6$}
\subsubsection{Closest Point Inequalities}
\begin{table}[h]
\renewcommand{\arraystretch}{1.2}
\begin{equation*}
\begin{array}{|c|c|}
\multicolumn{2}{c}{\textbf{Closest Point Inequalities for } 2 \times 4 \times 6} \\ \hline 
\text{Representative of orbit} & \text{Inequality} \\ \hline 
\ket{\Theta_0} & x_{2,1}+x_{2,2}+x_{3,1}+x_{3,2}\geq 1 \\
\ket{\Theta_1} & 6 x_{1,1}+5 x_{2,1}+5 x_{2,2}+6 x_{3,1}+6 x_{3,2}+5 x_{3,3}+5 x_{3,4}\geq 11 \\
\ket{\Theta_2} & 15 x_{1,1}+40 x_{2,1}+33 x_{2,2}+18 x_{2,3}+40 x_{3,1}+37 x_{3,2}+22 x_{3,3}+15 x_{3,4}+7 x_{3,5}\geq 55 \\
\ket{\Theta_3} & \text{n/a} \\
\ket{\Theta_4} & \text{n/a} \\
\ket{\Theta_5} & x_{2,1}+x_{3,1}+x_{3,2}+x_{3,3}+x_{3,4}\geq 1 \\ \hline
\end{array}
\end{equation*}
\caption{Inequalities calculated using the closest point method for $2 \times 4 \times 6$ systems}
\end{table}
\FloatBarrier
\subsubsection{Orbit polytopes}
\begin{table}[h]
\renewcommand{\arraystretch}{1.2}
\begin{equation*}
\begin{array}{|c|l|}
\multicolumn{2}{c}{\textbf{Some Orbit polytopes for $2 \times 4 \times 6$}} \\ \hline
\text{Representative of Class} & \text{Additional Inequalities} \\ \hline
\ket{\Theta_0} &  -x_{2,1}-x_{2,2}-x_{3,1}-x_{3,2}+1\leq 0 \\
\ket{\Theta_4} & \text{None} \\ \hline 
\end{array}
\end{equation*}
\caption{Entanglement Polytope Inequalities for the $\ket{\Gamma_1},\ket{\Gamma_{12}}$ and  $\ket{\Gamma_{13}}$ classes in $2 \times 4 \times 6$}
\end{table}

\begin{table}[h]
\renewcommand{\arraystretch}{1.1}
\begin{equation*}
\begin{array}{c}
\textbf{Orbit polytope of $\ket{\Theta_1}$ for $2 \times 4 \times 6$} \\ \hline
\text{Additional Inequalities} \\ \hline
 x_{2,3}-x_{3,1}\leq 0 \\
 x_{3,3}-x_{2,1}\leq 0 \\
 x_{3,4}-x_{2,2}\leq 0 \\
 x_{3,5}-x_{2,3}\leq 0 \\
 -x_{1,1}+x_{2,1}+x_{2,3}-x_{3,1}+x_{3,5}\leq 0 \\
 -x_{2,1}-x_{2,2}-x_{2,3}-x_{3,2}+1\leq 0 \\
 4 x_{2,3}-3 x_{3,1}+x_{3,2}+x_{3,3}+x_{3,4}-1\leq 0 \\
 x_{2,1}+x_{2,3}-x_{3,1}-x_{3,2}-x_{3,3}\leq 0 \\
 x_{2,1}+x_{2,3}-x_{3,1}-x_{3,2}-x_{3,3}\leq 0 \\
 -x_{1,1}+x_{2,1}-2 x_{3,1}-x_{3,2}-x_{3,3}+1\leq 0 \\
 -x_{1,1}-x_{2,2}-x_{2,3}-x_{3,1}+x_{3,4}+1\leq 0 \\
 -x_{2,2}-x_{2,3}-x_{3,1}-x_{3,2}-x_{3,4}+1\leq 0 \\
 -x_{2,2}-x_{2,3}-x_{3,1}-x_{3,2}-x_{3,4}+1\leq 0 \\
 -x_{2,1}-x_{2,2}-x_{3,2}-x_{3,3}-x_{3,4}+1\leq 0 \\
 -x_{2,2}-x_{3,1}-x_{3,2}-x_{3,3}-x_{3,4}+1\leq 0 \\
 -x_{1,1}+x_{2,1}+x_{2,3}-2 x_{3,1}-x_{3,2}-x_{3,3}-x_{3,4}+1\leq 0 \\
 2 x_{2,1}+2 x_{2,3}-x_{3,1}-x_{3,2}-x_{3,3}+x_{3,4}+x_{3,5}-1\leq 0 \\
 x_{1,1}-2 x_{2,1}-x_{2,2}-2 x_{3,2}-2 x_{3,3}-x_{3,4}-x_{3,5}+1\leq 0 \\
 -3 x_{1,1}-x_{2,1}-2 x_{2,3}-3 x_{3,1}-2 x_{3,2}+x_{3,4}+2 x_{3,5}+3\leq 0 \\
 x_{2,1}+x_{2,2}+x_{2,3}-x_{3,1}-x_{3,2}-x_{3,3}-x_{3,4}-x_{3,5}\leq 0 \\
 -5 x_{1,1}-4 x_{2,1}-4 x_{2,2}-5 x_{3,1}-5 x_{3,2}-4 x_{3,3}-4 x_{3,4}+9\leq 0 \\
 -x_{1,1}-x_{2,2}-x_{2,3}-2 x_{3,1}-x_{3,2}-x_{3,3}-x_{3,5}+2\leq 0 \\
 -x_{1,1}-x_{2,1}-2 x_{2,2}-2 x_{3,2}-2 x_{3,3}-x_{3,4}-x_{3,5}+2\leq 0 \\
 -x_{1,1}-x_{2,2}-2 x_{3,1}-x_{3,2}-x_{3,3}-x_{3,4}-x_{3,5}+2\leq 0 \\
 -x_{1,1}-x_{2,2}-2 x_{3,1}-x_{3,2}-x_{3,3}-x_{3,4}-x_{3,5}+2\leq 0 \\
 -5 x_{2,2}-5 x_{2,3}-2 x_{3,1}-2 x_{3,2}+3 x_{3,3}-2 x_{3,4}+3 x_{3,5}+2\leq 0 \\
 -5 x_{1,1}+5 x_{2,1}-6 x_{3,1}-x_{3,2}-x_{3,3}+4 x_{3,4}+4 x_{3,5}+1\leq 0 \\
 x_{1,1}+x_{2,1}+2 x_{2,2}+2 x_{2,3}-2 x_{3,1}-x_{3,2}-2 x_{3,3}-x_{3,4}-1\leq 0 \\
 x_{1,1}+2 x_{2,1}+x_{2,2}+2 x_{2,3}-2 x_{3,1}-2 x_{3,2}-x_{3,3}-x_{3,4}-1\leq 0 \\
 -2 x_{1,1}-4 x_{2,1}-2 x_{2,2}-2 x_{2,3}-x_{3,1}+x_{3,2}+x_{3,3}-x_{3,4}+3\leq 0 \\
 -9 x_{1,1}-9 x_{2,1}-9 x_{2,2}-10 x_{3,1}-10 x_{3,2}-10 x_{3,3}-11 x_{3,4}-x_{3,5}-19\leq 0 \\
 -9 x_{2,1}-10 x_{2,2}-x_{2,3}-x_{3,1}-9 x_{3,2}-10 x_{3,3}-9 x_{3,4}-x_{3,5}+10\leq 0 \\
 -2 x_{1,1}-2 x_{2,1}-x_{2,2}-2 x_{3,1}-3 x_{3,2}-2 x_{3,3}-2 x_{3,4}-x_{3,5}+4\leq 0 \\
 -5 x_{1,1}-10 x_{2,1}-5 x_{2,2}-5 x_{2,3}-4 x_{3,1}+x_{3,2}+x_{3,3}-4 x_{3,4}+x_{3,5}+9\leq 0 \\
 -2 x_{1,1}-5 x_{2,1}-3 x_{2,2}-3 x_{2,3}+x_{3,1}-2 x_{3,2}-x_{3,3}+x_{3,4}+x_{3,5}+4\leq 0 \\
 -5 x_{1,1}-5 x_{2,1}-5 x_{2,2}-x_{3,1}-x_{3,2}-x_{3,3}-x_{3,4}+4 x_{3,5}+6\leq 0 \\
 -x_{1,1}+x_{2,1}+2 x_{2,2}+2 x_{2,3}-2 x_{3,1}-3 x_{3,2}-2 x_{3,3}-2 x_{3,4}-x_{3,5}+1\leq 0 \\
 -2 x_{1,1}-4 x_{2,1}-2 x_{2,2}-x_{2,3}-2 x_{3,1}-3 x_{3,2}-x_{3,3}-2 x_{3,4}-x_{3,5}+6\leq 0 \\
 -3 x_{1,1}-6 x_{2,1}-3 x_{2,2}-3 x_{2,3}-x_{3,1}-4 x_{3,2}-4 x_{3,3}-x_{3,4}-x_{3,5}+7\leq 0 \\
 -x_{1,1}-x_{2,1}-2 x_{2,2}-x_{2,3}-2 x_{3,1}-x_{3,2}-2 x_{3,3}-x_{3,4}-x_{3,5}+3\leq 0 \\ \hline
\end{array}
\end{equation*}
\caption{The orbit polytope of $\ket{\Theta_1}$ in $2\times 4 \times 6$}
\end{table}
\FloatBarrier
\section{Three Qutrits}
Here we give a list of inequalities for the state corresponding to the multiplication of upper triangular $2 \times 2$ matrices. 
\begin{table}[h]
\begin{equation*}
\begin{array}{c}
\textbf{A complete list of inequalities for the Upper Triangular Matrix Multiplication Tensor}\\ \hline 
 x_{1,2}-x_{2,1}\leq 0 \\
 x_{2,2}-x_{1,1}\leq 0 \\
 x_{1,2}-x_{3,1}\leq 0 \\
 x_{2,2}-x_{3,1}\leq 0 \\
 x_{3,2}-x_{1,1}\leq 0 \\
 x_{3,2}-x_{2,1}\leq 0 \\
 -x_{1,1}-x_{2,1}+1\leq 0 \\
 -x_{1,1}-x_{3,1}+1\leq 0 \\
 -x_{2,1}-x_{3,1}+1\leq 0 \\
 -x_{1,2}+x_{2,1}-x_{3,1}\leq 0 \\
 x_{1,1}-x_{2,2}-x_{3,1}\leq 0 \\
 -x_{1,2}-x_{2,1}+x_{3,1}\leq 0 \\
 -x_{1,1}-x_{2,2}+x_{3,1}\leq 0 \\
 x_{1,1}-x_{2,1}-x_{3,2}\leq 0 \\
 -x_{1,1}+x_{2,1}-x_{3,2}\leq 0 \\
 -x_{1,2}-x_{2,2}+x_{3,2}\leq 0 \\
 -x_{1,1}-x_{1,2}-x_{2,2}+1\leq 0 \\
 -x_{1,1}-x_{1,2}-x_{3,2}+1\leq 0 \\
 -x_{2,1}-x_{2,2}-x_{3,2}+1\leq 0 \\
 x_{1,2}+x_{2,2}-x_{3,1}-x_{3,2}\leq 0 \\
 x_{1,2}-x_{2,1}-x_{2,2}+x_{3,2}\leq 0 \\
 -x_{1,1}-x_{1,2}+x_{2,2}+x_{3,2}\leq 0 \\
 x_{1,1}+x_{1,2}-x_{2,2}-x_{3,1}-x_{3,2}\leq 0 \\
 -x_{1,2}+x_{2,1}+x_{2,2}-x_{3,1}-x_{3,2}\leq 0 \\
 -x_{1,1}-x_{1,2}-x_{2,1}-x_{2,2}+x_{3,1}+1\leq 0 \\
 -x_{1,1}-x_{1,2}+x_{2,1}-x_{3,1}-x_{3,2}+1\leq 0 \\
 x_{1,1}-x_{2,1}-x_{2,2}-x_{3,1}-x_{3,2}+1\leq 0 \\
 -x_{1,1}-x_{1,2}-x_{2,1}-x_{2,2}-x_{3,1}+2\leq 0 \\
 -x_{1,1}-x_{1,2}-x_{2,1}-x_{3,1}-x_{3,2}+2\leq 0 \\
 -x_{1,1}-x_{2,1}-x_{2,2}-x_{3,1}-x_{3,2}+2\leq 0 \\
\end{array} 
\end{equation*}
\caption{The entanglement polytope of the Upper Triangular Matrix Multiplication Tensor}
\end{table}
\chapter{Appendix B: SLOCC Classes for $ 2 \times M \times N $ Systems as done by Chen and Chen}
In this appendix we briefly present some of the results by Chen and Chen in \cite{Chen2009} on the SLOCC classes of $2 \times M \times N$ systems collected in a table. For M = 1 the orbits represented by $\ket{\Gamma_7},\ket{\Gamma_{11}},\ket{\Gamma_{13}}$ fall away.
\begin{table}[h]
\begin{equation*}
\begin{array}{|c|l|}
\multicolumn{2}{c}{\textbf{Entanglement classes for selected } 2 \times M \times N \textbf{ systems}} \\\hline 
\text{local ranks} & \text{Representative of SLOCC Class} \\ \hline
2 \times M \times 2M, M \geq 2 & \ket{\Upsilon_0} := \ket{0}\sum_{i=0}^{M-1}\ket{ii}+\ket{1}\sum_{i=0}^{M-1}\ket{i,i+M} \\ \hline 
\multirow{2}{*}{$2 \times M+1 \times  (2M+1), M \geq 1$} & \ket{\Upsilon_1}:=\ket{0,M,2M}+\ket{\Upsilon_0} \\
& \ket{\Upsilon_2}:=\ket{0,M,2M}+\ket{1,M,M-1}+\ket{\Upsilon_0} \\ \hline
\multirow{6}{*}{$2 \times M+2 \times  (2M+2), M \geq 2$} & \ket{\Theta_0}:= \ket{1,M+1,2M+1}+\ket{\Upsilon_1} \\
& \ket{\Theta_1}:=  \ket{0,M+1,2M+1}+\ket{\Upsilon_1} \\
& \ket{\Theta_2}:= \ket{1,M+1,2M+1}+\ket{\Upsilon_2} \\
& \ket{\Theta_3}:=  \ket{0,M+1,2M+1}+\ket{1,M+1,2M}+ \ket{\Upsilon_1} \\
& \ket{\Theta_4}:=  \ket{0,M+1,2M+1}+\ket{1,M+1,0}+ \ket{\Upsilon_2} \\
&  \ket{\Theta_5}:=  \ket{0,M+1,2M+1}+\ket{1,M+1,2M}+ \ket{\Upsilon_2} \\ \hline
\multirow{15}{*}{$2 \times M+3 \times  (2M+3), M \geq 2$} &  \ket{\Gamma_0} := \ket{1,M+2,2M+2}+\ket{\Theta_0} \\
&  \ket{\Gamma_1} := (\ket{0}+\ket{1})\ket{M+2,2M+2}+\ket{\Theta_0} \\
&  \ket{\Gamma_2} := \ket{0,M+2,2M+2}+\ket{\Theta_1} \\
&  \ket{\Gamma_3} := \ket{1,M+2,2M+2}+\ket{\Theta_2} \\
&  \ket{\Gamma_4} := \ket{0,M+2,2M+2}+\ket{\Theta_2} \\
&  \ket{\Gamma_5} := \ket{1,M+2,2M+2}+\ket{\Theta_3} \\
&  \ket{\Gamma_6} := \ket{0,M+2,2M+2}+\ket{\Theta_3} \\
&  \ket{\Gamma_7} := \ket{1,M+2,2M+2}+\ket{\Theta_4} \\
&  \ket{\Gamma_8} := \ket{1,M+2,2M+2}+\ket{\Theta_5} \\
&  \ket{\Gamma_9} := \ket{1,M+2,2M+2}+\ket{0,M+2,2M+1}+\ket{\Theta_2} \\
&\ket{\Gamma_{10}} := \ket{1,M+2,2M+2}+\ket{1,M+2,2M+1}+\ket{\Theta_3} \\
&\ket{\Gamma_{11}} := \ket{1,M+2,2M+2}+\ket{0,M+2,
M+1}+\ket{\Theta_4} \\
&\ket{\Gamma_{12}} := \ket{1,M+2,2M+2}+\ket{0,M+2,
M}+\ket{\Theta_5} \\
&\ket{\Gamma_{13}} := \ket{0,M+2,2M+2}+\ket{1,M+2,2
M+1}+\ket{\Theta_5} \\ \hline
 
\end{array}
\end{equation*}
\caption{The entanglement classes in $2\times 4 \times N$ systems quoted from \cite{Chen2009}}
\end{table}
\chapter{Appendix C: Documentation of the Mathematica Package \emph{Gradient Flow and SIC.m}}
We explain here how  Mathematica package Gradient Flow and SIC.m is used\footnote{Mathematica 9.0 was used to write it and it has not been tested with older versions.}. This is not to be understood as a complete documentation, but rather as a guide to actually computing Entanglement polytopes using Semi-Interactive Polytope Computation, or SIC for short. While the package itself was put together by the author, the main body of the code was written by Michael Walter, with only minor additions and alterations by the author. \\ 
The two main parts of the package are the Gradient Flow and SIC Functions. Their use will be explained in detail here, the theoretical details are explained in chapter \ref{chapter:GradientFlow}. Some three Qubit examples are included, for more examples see the notebooks attached to the package. 
\section{Quantum Tools}
The package computes in the tensor product Hilbert space using the Built-In Mathematica \verb|KroneckerProduct| function. 
\subsection{Entering states}
States must be entered using the provided \verb|BasisKet| function which gets called as \verb|BasisKet[I,dims]| where \verb|I| and \verb|dims| must be lists of the same length. \verb|I| determines the entries of the ket while \verb|dims| specifies the local dimensions $(d_1,\ldots,d_N)$. For example, the three-qubit GHZ and W states can be entered as follows: 
\begin{verbatim}
GHZ = BasisKet[{0,0,0},{2,2,2}] + BasisKet[{1,1,1},{2,2,2}];
W = BasisKet[{1,0,0},{2,2,2}] + BasisKet[{0,1,0},{2,2,2}] + BasisKet[{0,0,1},{2,2,2}];
\end{verbatim}
\verb|BasisKet| returns a column vector: 
\begin{verbatim}
In[1]:= BasisKet[{0,0,0},{2,2,2}]
Out[1]= {{1},{0},{0},{0},{0},{0},{0},{0}}
\end{verbatim}
which is why the built-in Mathematica \verb|Normalize| won't work. Instead, the package provides a function \verb|Nrm|:
\begin{verbatim}
In[666]:= GHZ=BasisKet[{0,0,0},{2,2,2}]+BasisKet[{1,1,1},{2,2,2}];
GHZ //Nrm
Out[667]= {{1/Sqrt[2]},{0},{0},{0},{0},{0},{0},{1/Sqrt[2]}}
\end{verbatim}
\subsection{Eigenvalues and reduced density matrices}
There are functions \verb|LocalEigenvalues| and \verb|MostLocalEigenvalues| which compute the local eigenvalues resp. "Most Local Eigenvalues" (notice one must always include the local dimensions as second argument)  : 
\begin{verbatim}
In[672]:= LocalEigenvalues[GHZ, {2,2,2}]
MostLocalEigenvalues[W,{2,2,2}]
Out[672]= {{1/2,1/2},{1/2,1/2},{1/2,1/2}}
Out[673]= {{2/3},{2/3},{2/3}}
\end{verbatim}
There is also a function \verb|RDM| which computes the $i$-th reduced density matrix. It gets called as \verb|RDM[Rho,dims,i]|. Notice that the first argument is now a density operator $\rho$. The function \verb|Rho[Psi]| converts a state to a density operator. 
\begin{verbatim}
In[674]:= RDM[Rho[W],{2,2,2},1]
Out[674]= {{2/3,0},{0,1/3}}
\end{verbatim}

\section{The GradientFlow function} 
The \verb|GradientFlow| function uses the theoretical ideas of \ref{chapter:GradientFlow} to flow to the point in the entanglement polytope of a given state which is closest to a rational point in $Lie(T)$. 
\subsection{Input to the GradientFlow function}
The GradientFlow function gets called as 
\begin{verbatim}
GradientFlow[initialPsi,initialCoadjointUs,lambdas,dims]
\end{verbatim}.
The arguments must be entered as follows:
\begin{itemize}
\item The argument \verb|initialPsi| is the state in whose entanglement polytope one like to flow. It must be entered using the \verb|BasisKet| function provided with the package. 
\item The argument \verb|initialCoadjointUs| must be a tuple of unitary matrices corresponding to a starting point on the coadjoint orbit. Mostly one will want take a generic starting point, for this the package provides the function \verb|GenericCoadjoint[dims]| which returns a random tuple of unitary matrices. 
\item The argument \verb|lambdas| is a bit special. It does not take directly the point one wants to flow to. Instead, one must enter a tuple of Young diagrams $\lambda =(\lambda_1,\ldots\lambda_N)$, where $\lambda_i$ has $d_i$ rows (possibly empty). A Young diagram with $n_j$ boxes on row $j$ is entered as a list $\{n_1,\ldots n_{d_i}\}$ of non-increasing integers.  Each $\lambda_i$ must have the same number of boxes, $k$, such that the  point one wants to flow to in the polytope is $\frac{\lambda}{k}$. For example, if one wants to flow to the point $(1,1,\frac{1}{2})$ in the three-qubit polytope, one enters \verb|{{2,0},{2,0},{1,1}|. Alternatively, one can use the function \verb|ConvertToLambdas[point,dims]|, which takes a point in the ``Most Local Eigenvalues'' coordinates and transforms it to a tuple of Young diagrams. 
\item \verb|dims| specifies, as always, the local dimensions.
\end{itemize}
For example, a correct call would be: 
\begin{verbatim}
GradientFlow[W, GenericCoadjoint[{2, 2, 2}], {{2, 0}, {1, 1}, {1, 1}}, {2, 2, 2}]
\end{verbatim}
\subsection{Output of the GradientFlow function}
The Output of \verb|GradientFlow| is a list of \emph{replacement rules} which basically allow to access all results of the flow process. The list is very long, so one should not let the output be displayed, but rather address the results one is interested in specifically. The  results that can be accessed are: 
\begin{itemize}
\item \verb|Reached|:Bool which is \verb|True| if the target point was reached, i.e. the distance in $Lie(T)$ between the endpoint and the target point is less than the option value \verb|TargetPrecision|. By default, \verb|TargetPrecision|=0.01. 
\item \verb|FinalPsi| the state where the gradient flow stopped. 
\item \verb|FinalLocalEigenvalues| local eigenvalues of the final state, i.e the endpoint of the polytope trajectory. 
\item \verb|FinalDistance| the distance in $Lie(T)$ between the endpoint and the target point.
\item \verb|PsiTrajectory| the trajectory of the gradient flow in the Hilbert space. 
\end{itemize}
If \verb|Reached| is \verb|False|, then the following three additional parameters are returned.
\begin{itemize}
\item \verb|Inequality| The inequality for the polytope  calculated as explained in \ref{chapter:GradientFlow}
\item \verb|RawInequality| the raw version of this inequality, can be used to meddle around with when the calculated inequality is unprecise (remember we should always get rational inequalities)
\item \verb|suggestedInequality| a scaled and rounded version of \verb|RawInequality|. Its quality varies from perfect over useless to wrong. This will be important for the SIC.  
\end{itemize}
As a first example, we try to flow from the $W$ state to the origin: 
\begin{verbatim}
GradientFlow[W, GenericCoadjoint[{2, 2, 2}], {{1, 1}, {1, 1}, {1, 1}}, {2, 2, 2}]
\end{verbatim}
This returns 
\begin{verbatim}
{Reached -> False, 
 FinalPsi -> {{0.}, {0.57735}, {0.57735}, {0.}, {0.57735}, {0.}, \
{0.}, {0.}}, 
 FinalLocalEigenvalues -> {{0.666667, 0.333333}, {0.666667, 
    0.333333}, {0.666667, 0.333333}}, FinalDistance -> 0.408248, 
 PsiTrajectory -> {{{0.}, {0.57735}, {0.57735}, {0.}, {0.57735}, \
{0.}, {0.}, {0.}}}, 
 RawInequality -> {-0.816497, -0.816497, -0.816497, 1.63299}, 
 Inequality -> 
  1. Subscript[x, 1, 1] + 1. Subscript[x, 2, 1] + 
    1. Subscript[x, 3, 1] >= 2., 
 suggestedInequality -> {-1, -1, -1, 2}}
\end{verbatim}
We see the flow is stationary in this case because the $W$ state already maps to the closest point to the origin in its entanglement polytope. 
A usual call to the GradientFlow will look like this\footnote{Depending on what Mathematica decides to do, one maybe needs to address the parameters of the GradientFlow and the SIC as \texttt{Private`Parameter}.}: 
\begin{verbatim}
In[131]:= g =GradientFlow[AnotherWClassState,GenericCoadjoint[{2,2,2}],
{{1,1},{1,1},{1,1}},{2,2,2}];
FinalDistance /. g
Inequality /. g 
(* Or any other parameters one would like to address*)
Out[132]= 0.408248
Out[133]= 1. Subscript[x, 1,1]+1. Subscript[x, 2,1]+1. Subscript[x, 3,1]>=2.
\end{verbatim}
\subsection{Options of the GradientFlow function}
The options and their default values can be seen in the table below. 

\begin{table}[h]
\centering
\begin{tabular}{|c|c|}
\multicolumn{2}{c}{\textbf{Options of the GradientFlow Function}} \\ \hline 
Option & Default \\ \hline
\texttt{MaxSteps} & $\infty$ \\
\texttt{MinProgress} & $10^{-6}$ \\
\texttt{MinStepSize} & $10^{-6}$ \\
\texttt{InitialStepSize} & 1 \\ 
\texttt{MaxRestarts} & 5 \\
\texttt{TargetPrecision} & $10^{-2}$ \\
\texttt{Verbose} & False \\ \hline

\end{tabular}
\end{table}
Most are rather self-explanatory, only \verb|MaxRestarts|  needs further explanation. When there has been no progress for a while, the gradient flows \emph{restarts}: It resets the step counter and the step size (but keeps the trajectory). This is a numerical hack to ameliorate convergence. The Option \verb|MaxRestarts| gives an upper bound on the number of times this can be done. 
\subsection{Possible Issues}
Even though the gradient flow works well in many examples, especially low-dimensional ones\footnote{Some more are supplied with the package}, as one increases dimension both inefficiency and imprecision can occur. Here are some tips how one can try deal to with this. \\ 
\subsubsection{Qubit systems}
For Qubit systems the flow seems to have no precision problems, at least not for ``easy'' states such as GHZ and symmetric Dicke states. However as one goes up with the number of qubits the flow takes more and more computing time. It seems to be possible to improve on this by decreasing the values of \verb|MaxRestarts|, \verb|MinProgress|, and \verb|TargetPrecision|. 
\subsubsection{Imprecision in higher-level systems}
Here experience has shown that it is most promising to decrease \verb|MinProgress| value rather than the other option parameters. This is also reasonable once one thinks about the nature of the problem: We are trying to use the gradient to flow to a point where the gradient vanishes. It will be very small anyway in the neighbourhood of those points, so there is no real sense in further decreasing the stepsize. However, the descent might go on very slowly, so that one gains a bit on convergence by allowing for lesser progress\footnote{Of course, one can also increase the stepsize again. This however is already done by the restarts, and experience has shown that increasing the number of allowed restarts beyond five does not really add to convergence.}. However, one of course pays a price in terms of efficiency. 
\subsubsection{Inefficiency in higher-level systems}
We include this point just to say that there cannot be done much about it. Since convergence and precision become more fragile as one increases dimensions, the results quickly become useless if one tries to increase efficiency.
\section{Semi-Interactive Polytope Computation}
This is the main tool of the package. It implements the Algorithm presented and discussed in \ref{section:SIC}. The environment 
consists of the ``data structure'' of an SIC \emph{context} and the four functions \verb|SICStart| \verb|SICFlowOn|,
\verb|SICAddInequality|, and \verb|SICStop|. 
\subsection{The SIC context} 
The SIC context contains the data which is needed to run the algorithm presented in \ref{section:SICalgo}. It consists of lists of replacement rules for various data elements. We quickly go through the important ones. 
\begin{itemize}
\item In containts the initial data as replacement rules for \verb|InitialPsi|, \verb|InitialCoadjointUs|, \verb|Dims| determining the starting point and the local dimensions. 
\item It contains the data of the algorithm as replacement rules for \verb|CurrentInequalities|,\verb|VerticesExpected| and \verb|VerticesFound|. 
\item It contains replacement rules \verb|LastResult|, which contains the result of the last gradient flow procedure executed, and \verb|FlowOptions|, which contains a list of options passed on to the gradient flow. 
\end{itemize} 
Any of these parameters can in theory be addressed at any time of the procedure by using a Mathematica construct like 
\begin{verbatim}
c /.(FlowOptions -> _):> (FlowOptions ->{MinProgress ->10^(-10)});
\end{verbatim}
(again, depending on Mathematica, one has to address the parameters as \verb|Private`Parameter|)

\subsection{Setting up an SIC context}
To set up an SIC context the function \verb|SICStart| is used. It has three  mandatory arguments \verb|initialPsi|,
\verb|initialCoadjointUs| and \verb|dims|, which determine a starting point for SIC and the local dimensions (again, one will probabaly want to use \verb|GenericCoadjoint[dims]| as \verb|initialCoadjointUs|).  Furthermore, it has an optional argument \verb|opts| which should be a list of options for the gradient flow.\\ \verb|c = SICStart[initalPsi,initialCoadjointUs,dims,opts]| initialises an SIC context \verb|c| with initial data \verb|initialPsi|,\verb|initialCoadjointUs| and \verb|dims| and \verb|FlowOptions| set to opts.
 Furthermore, it initialises the parameters needed in the algorithm: 
 \begin{itemize}
\item \verb|CurrentInequalities| with a set of local inequalities, i.e. Weyl chamber and eigenvalues conditions. 
\item \verb|VerticesExpected| with the set of vertices corresponding to these inequalities, computed by \texttt{Qhull}.
\item \verb|VerticesFound| with the empty set. 
\end{itemize}
If possible, one will want to include the inequalities for the generic polytope. This can be done by using the function \verb|AddGenericInequalities| on a context \verb|c| after initialising: 
\begin{verbatim}
c = SICStart[GHZ,GenericCoadjoint[{2,2,2}],{2,2,2}];
c = SICAddInequalities[c];
\end{verbatim}

\subsection{Doing the SIC Flow}
After initialising the SIC Flow  can be started by using the function 
\verb|SICFlowOn| on \verb|c|:
\begin{verbatim}
c = SICFlowOn[c];
\end{verbatim}
As described already in \ref{section:SICalgo}, It works in the following way: 
\begin{itemize}
\item It takes a vertex from \verb|VerticesExpected| and runs a gradient flow from our starting to this vertex. Since the vertices are computed numerically, this involves a rounding procedure. 
\item If the vertex is found, it gets printed to the Front End as 
\begin{verbatim}
VERTEX FOUND: {Most Local Eigenvalues of the vertex}
\end{verbatim}
and continues with the next vertex from \verb|VerticesExpected|, or, in the case \verb|VerticesExpected| is empty, terminates with \verb|NO MORE VERTICES EXPECTED|. This means that at this point we have found the entire polytope of our initial state. 
\item If the vertex is not found, it outputs 
\begin{verbatim}
VERTEX NOT FOUND: {Most Local Eigenvalues of the vertex}
\end{verbatim}
to the Front End, followed by the closest point and the inequalities computed by the gradient flow, and interrupts the procedure. 
\end{itemize}
Again, this can best be seen in examples: 
\begin{verbatim}
c = SICStart[GHZ, GenericCoadjoint[{2, 2, 2}], {2, 2, 2}];
c = AddGenericInequalities[c];
c = SICFlowOn[c];
\end{verbatim}
produces the output 
\begin{verbatim}
VERTEX FOUND: {1.,1.,1.}
VERTEX FOUND: {1.,0.5,0.5}
VERTEX FOUND: {0.5,0.5,0.5}
VERTEX FOUND: {0.5,0.5,1.}
VERTEX FOUND: {0.5,1.,0.5}
NO MORE VERTICES EXPECTED
\end{verbatim}
while
\begin{verbatim}
c = SICStart[W, GenericCoadjoint[{2, 2, 2}], {2, 2, 2}];
c = AddGenericInequalities[c];
c = SICFlowOn[c];
\end{verbatim}
produces 
\begin{verbatim}
VERTEX FOUND: {1.,1.,1.}

VERTEX FOUND: {0.5,0.5,1.}

VERTEX NOT FOUND: {0.5,0.5,0.5}

  CLOSEST POINT: {0.666667,0.666667,0.666667}

  INEQUALITY: 1. Subscript[x, 1,1]+1. Subscript[x, 2,1]+1. Subscript[x, 3,1]>=2.

  RAW INEQUALITY: {-0.816497,-0.816497,-0.816497,1.63299}

  SUGGESTED INEQUALITY: {-1,-1,-1,2}
\end{verbatim}

\subsubsection{Adding Inequalities}
Once the SIC Flow failed to find to vertex, we must enter a new inequality by hand and then resume the SIC Flow. 
This is implemented in the function \verb|SICAddInequality|. It gets called as \verb|c = SICAddInequality[c,Ineqs]|, where \verb|c| is a SIC context and \verb|Ineqs| is an inequality or a list of inequalities in ``raw'' or ``Qhull'' format, i.e a list $\{n_1,...,n_d,d\}$ corresponding to the inequality $n.x + d \leq$. For example the inequality $x_{1,1}+ x_{2,1} + x_{3,1} \geq 2$ can be added in the example above as 
\begin{verbatim}
c = SICAddInequality[c{-1,-1,-1,2}]; 
c = SICFlowOn[c];
\end{verbatim}
which leads to the following output: 
\begin{verbatim}
VERTEX FOUND: {0.5,1.,0.5}
VERTEX FOUND: {1.,0.5,0.5}
NO MORE VERTICES EXPECTED
\end{verbatim}
In theory one always can take the inequality that the gradient flow finds. However, this soon fails to be precise. We will discuss this in the issues section. 
\subsubsection{After termination}
Once the SICFlow terminates one can use the function \verb|SICStop| on a context to display the vertices found in the progress. Continuing our W example: 
\begin{verbatim}
SICStop[c];
VERTICES:
1.	1.	1.
0.5	0.5	1.
0.5	1.	0.5
1.	0.5	0.5
\end{verbatim}
\subsection{Issues with the inequalities}
In theory, when having to enter an inequality, one should be able to  always just take the inequality the SIC Flow just calculated. However, remember that we always want integral inequalities with small coefficients. Unfortunately, this fails wildly as we start going to higher dimensions. \\
As a first rule one should check the accordance between the three inequalities displayed. Typically a nice suggested inequality points to good precision, but sometimes there are mysterious rounding errors, so also in this case one should double check. After some exercise, one can often guess the correct inequality to continue by looking at all three inequalities or meddling with the raw inequality (it's often a good idea to divide by the maximal or second heighest entry and then try to have a guess at rational numbers). \\ 
Often however the result is not sensible enough to continue. Often, the most promising thing is to run a \verb|GradientFlow| to the vertex the SIC did not find (SIC has a lower default precision in the flow to speed up the process). In many cases the inequality looks particularly better afterwards. However, it might also happen that with the gradient flow one reaches the vertex not found by SIC. In this case one can try to increase the precision parameters passed to the gradient flow in the SIC Flow. If this is not possible, one can use the function\verb|ConsiderFound|, which manually moves the vertex currently from expected to found. The correct syntax is 
\verb|c = ConsiderFound[c]|.
Lastly, it can also happen geometrically that one runs into an ``uglier'' inequality. For example, we saw that for higher dimensions the coefficients in the closest point inequalities become increasingly large. This means we ran into some corner of the polytope where several faces coincides. In such cases, use just carries these inequalities along. After termination one can try to run the SIC anew, adding some of the simpler inequalities but the ``ugly'' ones. Often the gradient flow then terminates without them, indicating they are implied by some simpler ones. 
\section{The closest point finder}\label{section:closestpointfinder}
The Package contains a function \verb|ClosestPointFinder[Psi,dims]| which tries to find the closest point in the entanglement polytope of \verb|Psi| using the theoretical ideas in chapter \ref{geometrictricks} (i.e. \verb|Psi| should satisfy the assumptions of Theorem \ref{thm:magiclemma}). In practice, what it does is trying to solve th eigenvector equation \ref{prop:free2} for the coefficients of a state with support $\supp$(\verb|Psi|). If it finds a solution (which is not certain, since we have not proven that a solution exists), it outputs the closest point, the corresponding state, and the inequality, in a list of rules for \verb|Point|, \verb|State| and \verb|Inequality|. 
\begin{verbatim}
In[200]:= ClosestPointFinder[GHZ,{2,2,2}]
ClosestPointFinder[W,{2,2,2}]
Out[200]= {Point->{{1/2,1/2},{1/2,1/2},{1/2,1/2}},
State->{{1/Sqrt[2]},{0},{0},{0},{0},{0},{0},{1/Sqrt[2]}},Inequality->True}
Out[201]= {Point->{{2/3,1/3},{2/3,1/3},{2/3,1/3}},
State->{{0},{1/Sqrt[3]},{1/Sqrt[3]},{0},{1/Sqrt[3]},{0},{0},{0}},
Inequality->Subscript[x, 1,1]+Subscript[x, 2,1]+Subscript[x, 3,1]>=2}
\end{verbatim}
Notice that when applying the \verb|ClosestPointFinder| on a state whose entanglement polytope contains the origin, like the GHZ state,  the inequality returned is ``True''.  
\listoffigures
\listoftables
\end{appendix}
\end{document}